\documentclass[msom,nonblindrev]{informs3}

\DoubleSpacedXI

\usepackage{natbib}
 \bibpunct[, ]{(}{)}{,}{a}{}{,}%
 \def\bibfont{\small}%
\setcitestyle{authoryear}
\TheoremsNumberedThrough 
\EquationsNumberedThrough

\usepackage{booktabs} % For formal tables
\usepackage[ruled]{algorithm2e} % For algorithms
\usepackage{floatrow}
\usepackage{multirow}
\usepackage{subcaption}
\usepackage{tabularx}
\usepackage[framemethod=tikz]{mdframed}
\usepackage{subfiles}
\usepackage{pstricks}
\usepackage{hyperref}

\usepackage{dcolumn}
\newcommand{\tablemc}[1]{\multicolumn{1}{c@{}}{#1}}

\usepackage{soul}

\SetAlFnt{\small}
\SetAlCapFnt{\small}
\SetAlCapNameFnt{\small}
\SetAlCapHSkip{0pt}
\IncMargin{-\parindent}

\newcommand{\E}{{\mathbb{E}}}

\newcommand{\jobs}{donations}
\newcommand{\job}{donation}
\newcommand{\worker}{volunteer}
\newcommand{\workers}{volunteers}

\newcommand{\matchmin}{match-min}
\newcommand{\s}{\mu}
\newcommand{\gap}{non-adoption growth benefit}
\newcommand{\adoptp}{\rho}
\newcommand{\temporary}{one-time}
\renewcommand{\problem}{platform design problem}
\newcommand{\regimeA}{AGD regime}
\newcommand{\regimeB}{AGN regime}

\newcommand{\revcolor}{}
\newfloatcommand{capbtabbox}{table}[][\FBwidth]

\usepackage{comment}
\usepackage[showdeletions]{color-edits}
\definecolor{auburn}{rgb}{0.43, 0.21, 0.1}
\addauthor{vm}{auburn}
\definecolor{cadmiumgreen}{rgb}{0.0, 0.42, 0.24}
\definecolor{scottgreen}{rgb}{0.1, 0.72, 0.34}
\addauthor{sr}{scottgreen}
\addauthor{il}{magenta}

\begin{document}

% Title. Note the optional short title for running heads. In the interest of anonymization, please do not include any acknowledgements.
\TITLE{Commitment on Volunteer Crowdsourcing Platforms: Implications for Growth and Engagement}

\RUNTITLE{Commitment on Crowdsourcing Platforms}

\RUNAUTHOR{Lo et al.}

\ARTICLEAUTHORS{%
%\AUTHOR{Irene Lo}
%\AFF{Stanford MS\&E, Stanford, CA, \EMAIL{ilo@stanford.edu}}
%\AUTHOR{Vahideh Manshadi}
%\AFF{Yale School of Management, New Haven, CT, \EMAIL{vahideh.manshadi@yale.edu}}
%\AUTHOR{Scott Rodilitz}
%\AFF{Yale School of Management, New Haven, CT, \EMAIL{scott.rodilitz@yale.edu}}
%\AUTHOR{Ali Shameli}
%\AFF{Stanford MS\&E, Stanford, CA, \EMAIL{shameli@stanford.edu}}
\AUTHOR{Irene Lo\textsuperscript{1}, Vahideh Manshadi\textsuperscript{2}, Scott Rodilitz\textsuperscript{2}, Ali Shameli\textsuperscript{1}}
\AFF{\textsuperscript{1}Stanford MS\&E, Stanford, CA \\ \textsuperscript{2}Yale School of Management, New Haven, CT}

}

% Abstract. Note that this must come before \maketitle.

\ABSTRACT{ 
{\bf Problem Definition:} 
Volunteer crowdsourcing platforms match volunteers with tasks which are often recurring. To ensure completion of such tasks, platforms frequently use a lever known as ``adoption,'' \revcolor{which amounts to a commitment by the volunteer to repeatedly perform the task}. Despite reducing match uncertainty, high levels of adoption can decrease the probability of forming new matches, which in turn can suppress growth. We study how platforms should manage this trade-off.

\noindent {\bf Academic/Practical Relevance:}
Our research is motivated by a collaboration with Food Rescue U.S. (FRUS), a volunteer-based food recovery organization active in over 30 locations. For platforms such as FRUS, success crucially depends on volunteer engagement. Consequently, effectively utilizing non-monetary levers, such as adoption, is  critical.

\noindent { \bf Methodology:} 
Motivated by the volunteer management literature and our analysis of FRUS data, we develop a model for 
two-sided markets which repeatedly match volunteers with tasks. Our model incorporates match uncertainty as well as the negative impact of failing to match on future engagement. 
We study the platform's optimal policy for setting the adoption level to maximize the total discounted number of matches.

\noindent { \bf Results: }
We fully characterize the optimal myopic policy and show that it takes a simple form: depending on volunteer characteristics and market
thickness, either allow for full adoption or disallow adoption. In the long run, we show that such a policy is either optimal or achieves a constant-factor approximation. \revcolor{Our finding is robust to incorporating heterogeneity in volunteer behavior.}

\noindent { \bf Managerial Implications: }
Our work sheds light on how  two-sided platforms need to carefully control the double-edged impacts that commitment levers have on growth and engagement. \revcolor{Setting a misguided  adoption level may result in marketplace decay.} At the same time, a one-size-fits-all solution may not be effective, as the optimal design crucially depends on the characteristics of the volunteer population.
}


\maketitle

\vspace{-20pt}

\section{Introduction}
\label{sec:intro}

%\vc{todo: update title and abstract}
%\vc{Scott: plz bring the MSOM abstract format. }

%\vc{GENERAL WORD CHOICES: use uncertainty instead of randomness, heterogeneity instead of diversity, policy instead of algs, donations (instead of tasks/rescues/scheduled donations/jobs), availabilities or match opportunities (instead of options), volunteers instead of workers, myopic instead of single-period.}

%\vc{todo: in the conclusion, add a discussion on the limitation of data + how our work relates to a broader context in two-sided markets. }

%\vc{once we introduce the model and problem, then we finish by saying that such commitment levers are used in other platforms as well.}

In 2018, there were an estimated 37.2 million Americans living in food insecure households
\citep{food_insecurity}, a number that is expected to grow substantially as a consequence of the COVID-19 pandemic \citep{galvin_2020}. Food insecurity is correlated with a myriad of poor health outcomes, particularly for children, and can have significant impacts on cognitive development and scholastic achievement \citep{ke2015food}. Meanwhile, more than 60 million tons of food go to waste in the U.S. each year, contributing roughly 2\% of the total annual greenhouse gas emissions in the U.S. \citep{venkat2011climate}.

Food recovery solutions seek to combat the problems of food insecurity and food waste simultaneously by transporting food from where it would be wasted to where it is needed. In doing so, they address {\em six} of the seventeen United Nations Sustainable Development Goals: (i) zero hunger; (ii) good health and well-being; (iii) quality education; (iv) sustainable cities and communities; (v) responsible consumption and production; and (vi) climate action.
%\vc{this transition is really abrupt. Maybe here we can just say if we can improve we move towards these goals. Then in next parag. we talk challenges and then platforms as one part of the solution?}
%Despite the fact that food insecurity and food waste coexist in close proximity around the world, 
However, food recovery programs face two key operational challenges. 
First, potential food donors need to be connected with potential food recipients. Second, once a connection has been made, the food must be transported from donor to recipient, and such last-mile transportation can be prohibitively costly.

In recent years, online non-profit platforms have emerged to both facilitate connections between donors and recipients and crowdsource transportation from volunteers.
Our work is motivated by a collaboration with one such food recovery organization, Food Rescue U.S. (FRUS).
In conjunction with FRUS, we determined that the goals of stakeholders, including donors, volunteers, recipients, and staff, are aligned: they all strive to maximize the amount of nutritious food that they put in the hands of food insecure Americans. We also identified that platform growth and volunteer management were two key areas where improved platform operations could have a large impact. 
In the following, we provide a brief background on FRUS, and we then summarize our contributions toward improving the operations of FRUS and similar platforms.

%\vc{right here, you should say this "This work is motivated by our collaboration with a non-profit platform called Food Rescue U.S.(FRUS) that recovers food from local businesses and donates it to non-profit agencies by crowdsourcing the transportation to volunteers" (this is from the other paper)}

\begin{figure}[t]
 \centering
 \begin{subfigure}[b]{0.56\textwidth}
 \includegraphics[width=\textwidth]{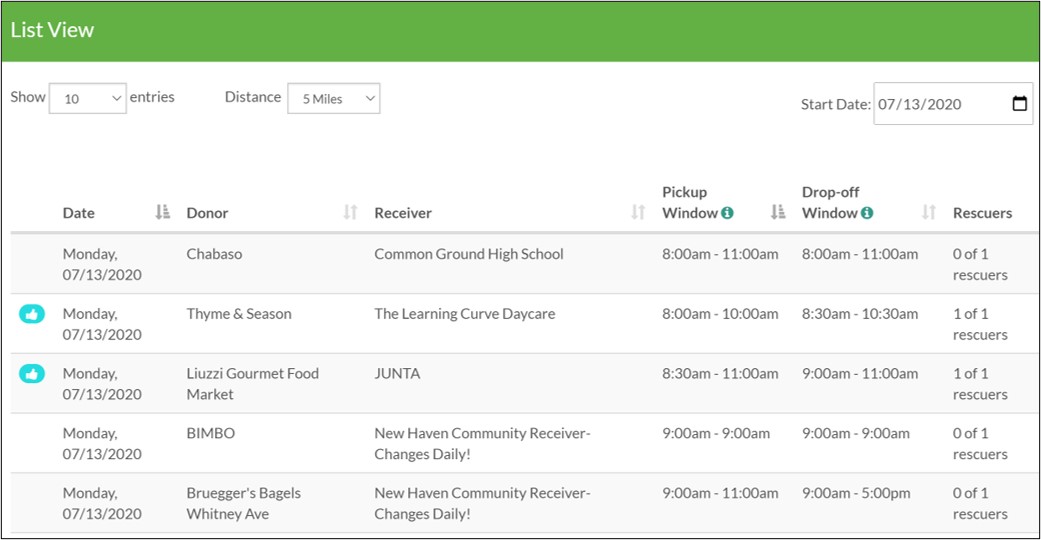}
  %\caption{ }\label{fig:FRUSimage}
 \end{subfigure}
 \begin{subfigure}[b]{0.1841\textwidth}
 \centering
 \includegraphics[width= \textwidth]{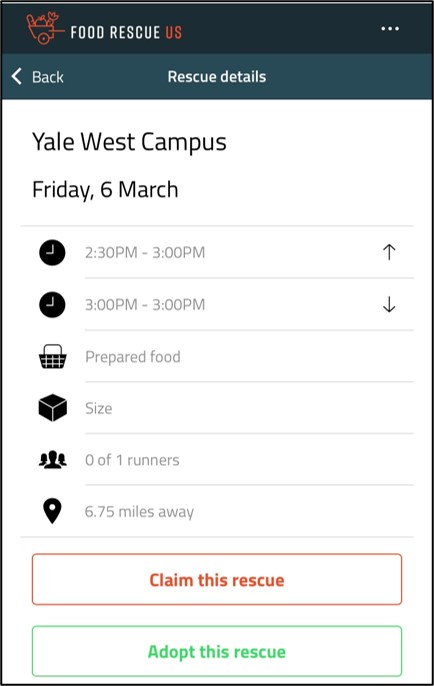}
 \end{subfigure}
  \begin{subfigure}[b]{0.20825\textwidth}
 \centering
 \includegraphics[width= \textwidth]{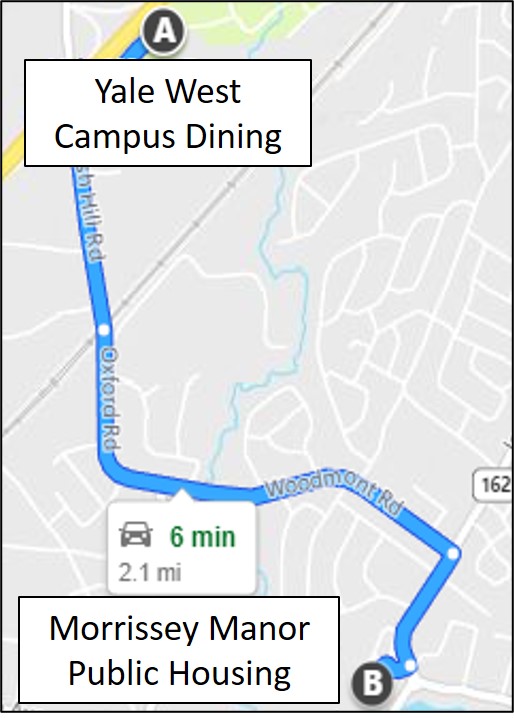}
\end{subfigure}
\caption{{\it Left:} Scheduled donations on the FRUS website. {\it Middle:} Details for a particular donation, including ``Claim'' and ``Adopt'' options for volunteer sign ups. {\it Right:} Map of a donation route.}
\label{fig:FRUSplatform}
\end{figure}

\subsection{Background on Food Rescue U.S.}
\label{subsec:frusbackground}
%{\bf Background on FRUS: }
FRUS is a non-profit platform that facilitates perishable food recovery from local businesses. Since its inception in 2011, FRUS has provided over 60 million meals to food insecure individuals in more than 30 locations across the country, and FRUS continues to grow significantly both as a platform and in individual locations.  In each FRUS location, a site director connects food donors with receiving agencies. Since collecting donated food is time-sensitive, the site director also creates a schedule for pick-up, which is posted on the FRUS website and app (left two panels of Figure \ref{fig:FRUSplatform}). 
FRUS relies on its large base of thousands of volunteers for transportation.
At any time, volunteers can log on and ``match'' with an available donation, i.e., sign up to transport donated food for a specified connection at the scheduled time.

Since most donations are scheduled to recur on a weekly basis, one natural tool for ensuring completion is to encourage volunteers to ``adopt'' donations and commit to repeatedly transporting the same donation each week. Adopted matches reduce uncertainty about whether donations 
will be transported in the subsequent weeks and can lead to an increase in volunteer retention. Consequently, more adoption is generally thought to be desirable for marketplace health and growth.

Based on this line of thinking, FRUS uses the same design for every location: 
 all recurring donations are available for {\em both} \temporary \ sign up or adoption (through the ``Claim this rescue'' and ``Adopt this rescue'' button, respectively, as shown in the middle panel of Figure \ref{fig:FRUSplatform}). \revcolor{This platform design choice -- allowing full adoption -- is mirrored by other food recovery organizations, such as Food Finders and Second Harvest, and is also employed by many other
 volunteer-based organizations which rely on crowdsourcing for simple recurring tasks, such as book sorting in libraries or food packaging at food banks (e.g., \href{https://www.pittsburghfoodbank.org/get-involved/volunteer/individual/}{the Greater Pittsburgh Community Food Bank}).}
 
 \revcolor{However, for platforms in a growth phase such as FRUS, having a larger proportion of volunteers in adopted matches does not guarantee an increase in long-term growth. Growth relies on volunteer engagement in the \emph{spot market} 
(where new matches are formed between donations and volunteers), and adoption may reduce such engagement. To highlight the tension between adoption and engagement, consider the following thought experiment. Suppose that once a volunteer adopts a donation, they will simply complete their weekly donation and never participate in the spot market. On the other hand, volunteers who are not part of an adopted match check the platform frequently and participate in the spot market multiple times each week.  When a new donation joins the system (e.g., an extra donation from a current donor to a different agency), if all volunteers are already in adopted matches the donation will not be matched. However, if some volunteer is not in an adopted match, they could be sufficiently engaged in the spot market to complete \emph{both} their original donation and the new donation. In this hypothetical example, a high adoption level decreases engagement and matching in the spot market.
 {If participants in the spot market cannot find matches, they may be less likely to participate in the future,} which may hinder growth.}

\revcolor{In Section \ref{subsec:data}, we provide evidence in support of the above. Specifically, in accord with previous research on volunteer under-utilization (e.g., \citealt{sampson2006optimization}), we find that donations and volunteers become more likely to leave the platform if they are not matched. 
{Furthermore, we provide suggestive evidence that volunteer engagement depends on match type. Volunteers who successfully find one-time matches on the spot market can be more likely to increase their engagement above one donation in subsequent weeks, which enables platform growth.} Given the importance of attaining large, positive growth rates for the survival of nonprofit platforms \citep{bielefeld1994affects}, the potentially detrimental impact of adoption on growth and engagement raises a key question that motivates our work:}

 \noindent {\em How should volunteer crowdsourcing platforms such as FRUS effectively utilize commitment levers?}

\subsection{ Our Contributions}
%{ \bf Our Contributions: } 
To study the above platform design question, we first 
(i) present our analysis of FRUS data which highlights the relationship between matching and future engagement and 
(ii) 
 develop a model for repeated two-sided matching 
 that captures the key features observed in our analysis of FRUS data. As part of our model, we develop a new matching function for two-sided markets with random compatibility that may be of independent interest.
Using this model of platform operations, we then provide insights into how a platform can maximize the total discounted number of matches, by 
 (iii) {establishing that} the optimal myopic policy is always either {allowing all donations to be adopted, or disallowing adoption.}
 (iv) We show that one of these simple policies provides a constant-factor approximation to the optimal policy, and under many conditions is optimal.

{{\bf Modeling the Platform Operations:}}
We model the FRUS platform as a decentralized 
two-sided matching market in which volunteers and donations repeatedly match. %, and their future engagement with the platform depends on whether they successfully find a match.
Volunteers and donations arrive in discrete periods, and are compatible with users on the other side of the market independently with identical probability. 
Each match is either an adopted match, carried across periods, or a \temporary \ match (``non-adoption''), newly formed in each period's spot market via a greedy matching process. %(we use the terms ``\temporary \ match'' and ``non-adoption'' interchangeably). 
The engagement and growth of donations and volunteers in each period is endogenously governed by the volume of matches in the previous period; furthermore,
motivated by our empirical observations (Section \ref{subsec:data}), volunteer engagement differs based on whether the previous-period match was an adoption or a non-adoption. We use this framework to shed light on the benefits and drawbacks of adoption.

{\bf Characterizing Matching in the Spot Market:}
To analyze the {impact of adoption on the efficiency of the spot market}, we develop a novel matching function that approximates greedy matching in two-sided markets with random compatibility.
Our matching function solves a deterministic differential equation which characterizes the large market limit of our greedy matching process (Proposition~\ref{prop:matching}).
{Unlike commonly-used alternatives, our matching function exhibits the desirable property that
\emph{match efficiency}\textemdash the proportion of users on the minimum side of the market who find a match\textemdash is increasing in \emph{market thickness}\textemdash the ex ante number of compatible matches for each user\textemdash while also respecting the physical limitation that the number of matches never exceeds the size of the minimum side of the market.}

{\bf Characterizing the Platform's Optimal Design:}
In a finite-horizon setting, the platform determines the allowable level of adoption in each period to maximize the total discounted number of transported donations. \revcolor{The adoption level decision can have a dramatic impact on the platform's growth rate, potentially making the difference between marketplace growth and decay, and we formally introduce the platform design problem to help optimize this decision (Figure \ref{fig:market_death} and its related discussion in Section \ref{subsec:opt}).}
Using our framework, we first characterize the 
optimal myopic policy, which  maximizes the number of matches in the next period.
 This captures three trade-offs in deciding platform adoption levels: match type influences volunteer engagement; adoption increases match certainty in the subsequent period; and disallowing adoption increases the size and efficiency of the spot market in the subsequent period.
Somewhat surprisingly, 
we establish that the optimal myopic policy has 
a simple structure: either enable adoption for all donations or disallow adoption.

We use these insights to study policies for maximizing the total discounted number of transported donations. The optimal policy depends on the relative engagement of volunteers in adopted matches compared to volunteers in \temporary\ matches. %lead to more net volunteer growth. 
In the \emph{Adoption Growth Dominance (AGD)} regime, adopted matches lead to relatively more volunteer {engagement}, 
%so there is no trade-off between short-term match efficiency and long-term market growth. W
and we provide the intuitive result that full adoption is optimal 
%maximizing platform adoption levels maximizes total discounted number of transported donations 
(Theorem~\ref{thm:optregimea}). %This is because there is no trade-off between short-term match efficiency and long-term market growth.
In the complementary \emph{Adoption Growth Non-Dominance (AGN)} regime decreasing adoption increases {volunteer engagement} in the subsequent period. We show that when the market is sufficiently thick, a policy of disallowing adoption is both myopically optimal and maximizes total discounted matches (Theorem~\ref{thm:optregimeB}). 
For insufficiently thick markets in the AGN regime, 
a policy of disallowing adoption achieves a constant-factor approximation, where the approximation factor improves with market thickness (Theorem~\ref{thm:approxregimeB}).

\revcolor{In Section \ref{sec:hetero}, we generalize our base model by incorporating volunteer heterogeneity which may arise in practice. In particular, we assume that the volunteer population is comprised of 
``students'' and ``professionals'' who interact with the platform differently. For example, students (professionals) may engage more (less) if not part of an adopted match and they may be less (more) likely to adopt a donation if adoption is allowed.  By analyzing this enriched model, we show that our findings are robust to incorporating such heterogeneous behavior (Propositions \ref{prop:myopichetero} and \ref{prop:approxrepeatmyopic}).}

\revcolor{Taken together, our work provides three main contributions to  volunteer management:  (i) shedding light on how commitment influences growth and engagement in two-sided volunteer-based platforms, (ii) highlighting the importance of carefully designing commitment levers, and (iii) providing simple guidelines on the effective use of such levers. 
There are a growing number of volunteer matching platforms like FRUS (e.g. VolunteerSpot, VolunteerMatch, Food Rescue Hero) as well as  nonprofit organizations (libraries, food banks, museums, etc.) which rely on volunteers to complete simple, recurring tasks. 
These platforms and organizations follow different commitment strategies: many allow for commitment whenever possible (e.g., by having buttons for adoption like FRUS), others such as Food Runners require adoption, while organizations such as the Hollywood Food Coalition only allow for \temporary\ sign-ups. 
Our work suggests that there is no one-size fits-all solution, and we provide insight on when each design approach would be preferred.}

\subsection{Related Literature}

{Food insecurity and waste are issues of immense societal importance, and there is growing research on measuring their causes and impacts \citep{sanders2020dynamic, belavina2021grocery} and using operational tools to mitigate their effects \citep{ahire2018harvest, levi2019optimal}. 
In non-profit food distribution, there is particular scope for improving logistical efficiency. 
One line of work explores how to match food donations to recipients to ensure food allocation is targeted and equitable \citep{fianu2018markov,orgut2018robust}. %\citep{fianu2018markov,orgut2018robust, lee2019webuildai}.
Another strand considers how to solve the routing problem in distributing donated food jointly with equity considerations \citep{balcik2014multi,lien2014sequential}. %, and in isolation \citep[see, e.g.,][]{gunes2010vehicle, yildiz2013planning, davis2014scheduling, solak2014stop, reihaneh2018multi}. 
We focus on how to manage engagement and growth of the pool of volunteers who distribute donated food.}  

{Our study also contributes to the literature on improving labor management. 
%In the non-profit space, 
\citet{musalem2019balancing} provide empirical evidence of the effects of workload distribution on customer conversion and employee retention, 
whereas \citet{song2018closing} empirically show that publicly disclosing relative performance feedback can improve worker productivity.
In other interesting directions,
 \citet{green2013nursevendor} and \citet{dong2017managing} study how to manage labor to maximize completed jobs in call centers and hospitals. 
In volunteer-based work, \citet{ata2019dynamic} consider staffing and managing volunteers in gleaning operations to maximizing the volume of food gleaned, while \citet{manshadi2020online} consider volunteer notification policies for online crowdsourcing platforms. %The importance of managing labor participation rates is similarly important in for-profit operations: 
There is a growing body of work on how volunteer management should differ from traditional labor management given the unique considerations in volunteer labor retention  (see e.g. \citealt{locke2003hold} and \citealt{sampson2006optimization} for overviews). %Our setting is partially motivated by the volunteer labor market, and w
 We include such considerations in studying how to manage volunteers to improve market growth.}

Another related body of literature considers issues of design and management in dynamic matching platforms. A number of works study the trade-off between current matches and future market thickness and match efficiency, such as in  ridesharing \citep{ma2018spatio, ozkan2020dynamic}, kidney assignment \citep{anderson2017efficient, ashlagi2019matching, ata2020achievable}
%\citep{su2005patient, anderson2017efficient, ashlagi2018maximizing, ashlagi2019matching, ata2020achievable} 
and other markets \citep{bimpikis2020managing}. 
%We consider a similar trade-off in a setting where current matches affect future matches through user growth. 
Previous literature finds that limiting information and/or available actions on two-sided platforms can improve welfare %, primarily through reducing search frictions 
\citep{halaburda2017competing, kanoria2020facilitating}. %, arnosti2018managing}. 
We 
%focus on the lever of \emph{commitment} in a repeated  matching market and 
similarly find that controlling the level of \emph{commitment} can increase the long-run number of matches. Most of the previous work on dynamic matching focuses on improving the matching of heterogeneous agents and items, e.g. in housing allocation %\citep{bloch2012optimal, kurino2014house, 
\citep{leshno2017dynamic, arnosti2020design}, healthcare %applications 
\citep{su2005patient, zhang2020truthful}, 
%\citep{zenios2002optimal, su2005patient, dickerson2012dynamic, zhang2020truthful}, 
and other domains \citep{hu2018dynamic, feigenbaum2020dynamic}. %, chen2019pricing}. 
We consider a repeated setting with ex ante homogeneous jobs and workers, and focus on how commitment affects match volume through user engagement and growth.

Within the literature on two-sided marketplace design, the most closely related paper is \citet{lian2019optimal}, which provides a theoretical framework for optimal growth in a two-sided revenue-maximizing marketplace. {Their analysis relies} on a Cobb-Douglas matching function, {and their policy recommendations center around the use of} price levers for both sides of the market %; this formulation allows them to characterize the optimal market balance and pricing policy 
to promote growth and maximize profit. 
In our setting the scope for growth is very different, as physical constraints limit increasing returns to scale to occur at a diminishing rate. We develop a novel matching function to capture this  constraint. We also focus on the lever of commitment in a market without prices. %In order to capture uncertainty in matching and study the effects of adopted matches, we  develop a novel matching function. We similarly provide a theoretical formulation where user growth depends endogenously on the volume of past matches. However the structure of our optimal policies differ substantially, as we consider a setting without prices where the only lever is the type of matching.  
{To our knowledge, our work is one of the first to study growth in a two-sided  matching platform without monetary transactions.}

\section{Model} 
\label{sec:model}
In this section, we present our framework for studying decentralized two-sided matching platforms such as FRUS.
In Section~\ref{subsec:data} we motivate our modeling framework through an analysis of FRUS data.
In Section \ref{subsec:dynamics}, we formally introduce our model of a repeated matching market. We show that, properly scaled, the stochastic process underlying our model converges to a deterministic system of equations. \revcolor{In Section \ref{subsec:opt}, we motivate the importance of the design question that FRUS faces when considering the use of adoption, and formally introduce their optimization problem.}

 \subsection{Analysis of FRUS Data}
 \label{subsec:data}
 We begin this section by providing a brief overview of FRUS data. We then highlight three key findings which inform our modeling framework: (i) donor dropout depends on previous-period matching outcomes, (ii) volunteer dropout also depends on previous-period matching outcomes, and (iii) volunteers 
 engage with the platform asymmetrically depending on their match type.
 
 {\bf Data Overview: } Our analysis is based on nearly two years of FRUS data, consisting of over 35,000 donations from grocery stores, restaurants, and other food providers across the U.S.
 Nearly all donations repeat the following week (92\%), and for each donation, we observe detailed temporal information enabling us to track its availability and match type in any time period. 
 We observe sign-ups (matches between donations and volunteers) from 1652 unique volunteers:
 approximately $2/3$ of  sign-ups are adopted matches, but only 30\% of volunteers are ever part of an adopted match.

 {\bf Match Success Influences Donor Dropout: }
 We find that, unsurprisingly, future donor engagement with the platform depends on whether the current donation is completed.
Since donors spend time preparing food for donation, they can be discouraged from future participation %participating in the system 
when their donation is not picked up and %consequently 
goes to waste. 
 Using a logistic regression analysis, we estimate that the average donor is $3.40$ times more likely to drop out (defined as not offering that donation for at least 6 weeks) if no one picks up their donation. Our analysis includes $22,981$ donor-weeks, and this finding is significant at the $p < 0.001$ level. \revcolor{More details can be found in Appendix \ref{app:emp:dondrop}}

 {\bf Match Success Influences Volunteer Dropout: }
 Next, we turn our attention to the behavior of volunteers on the FRUS platform. Previous work in the context of volunteer labor suggests that if volunteers feel under-utilized (e.g. in the FRUS setting, if they cannot find a match), then they will be less engaged in the future \citep{sampson2006optimization}. %, similar to the effect described above for donors. 
Though volunteers who are part of an adopted match are by definition matched to a compatible donation in the subsequent week, we aim to empirically study how volunteers who are not part of an adopted match react when they fail to find a compatible donation. 
However, a key challenge 
%due to the ``filtered'' nature of our dataset:
is that we do not observe volunteers' visits to the platform unless they sign up for a donation. 
%However, a key challenge is that we do not observe volunteers' visits to the platform unless they sign up for a donation. 
Consequently, we do not directly observe when volunteers drop out of the system
nor how often they search for a match.

To overcome this challenge, 
we  use  the number of available donations  as a proxy for the likelihood of finding a match: if volunteers see more available donations upon visiting the platform, then they are more likely to find a compatible one. 
For any volunteer who transports a donation in a certain week but is not currently part of an adopted match, we count how many available donations they would see if they visited the platform the following week. 
We then keep track of their activity to check whether they drop out (defined as not signing up for another donation for at least 6 weeks). 
Focusing on four large locations (which together comprise a majority of the data) while incorporating location-specific fixed effects (to account for differences in dropout rates across locations), we use a logistic regression analysis to ascertain whether the number of available donations impacts dropout probability. Based on 3204 observations for volunteers who are not part of an adopted match, we find that an additional available donation decreases the dropout probability for the average such volunteer by 2.85\%. This effect is significant at the $p < 0.001$ level.

In contrast, repeating a similar analysis for volunteers currently in an adopted match reveals that
an additional available donation does not have a statistically significant effect on dropout probability.
This supports our assertion that dropout rates depend on whether or not volunteers can find a compatible donation: these volunteers are guaranteed to have a match independent from the number of available donations. Thus, we would not expect an additional available donation to influence the dropout rate of volunteers in an adopted match. For other volunteers, however, fewer available opportunities implies a lower chance of finding a match, which in turn increases their dropout rate. \revcolor{Details on both analyses can be found in Appendix \ref{app:emp:voldrop}}

\begin{figure}[t]
 \centering
 \begin{subfigure}[b]{0.48\textwidth}
 \includegraphics[width=\textwidth]{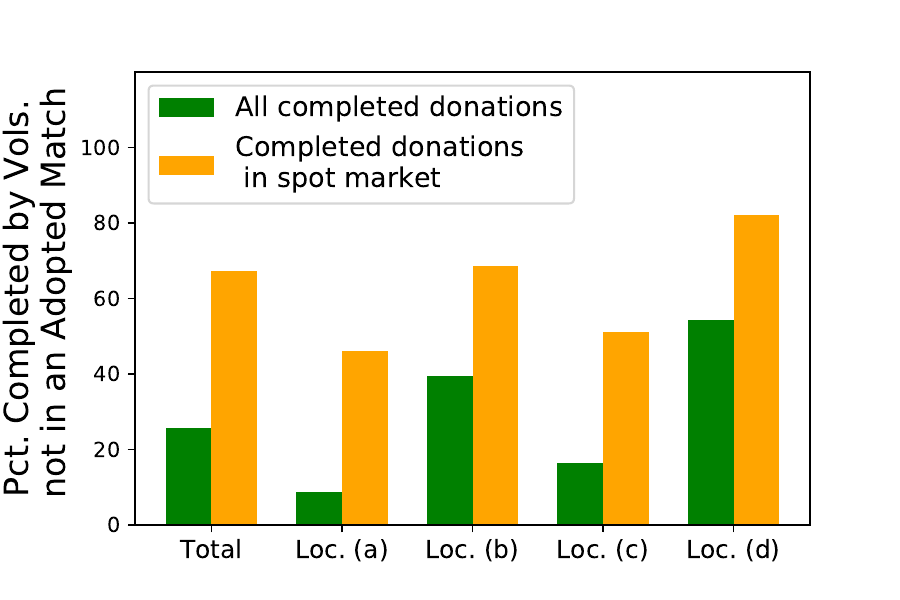}
  %\caption{ }\label{fig:summary2}
 \end{subfigure}
 \begin{subfigure}[b]{0.48\textwidth}
 \centering
 \includegraphics[width= \textwidth]{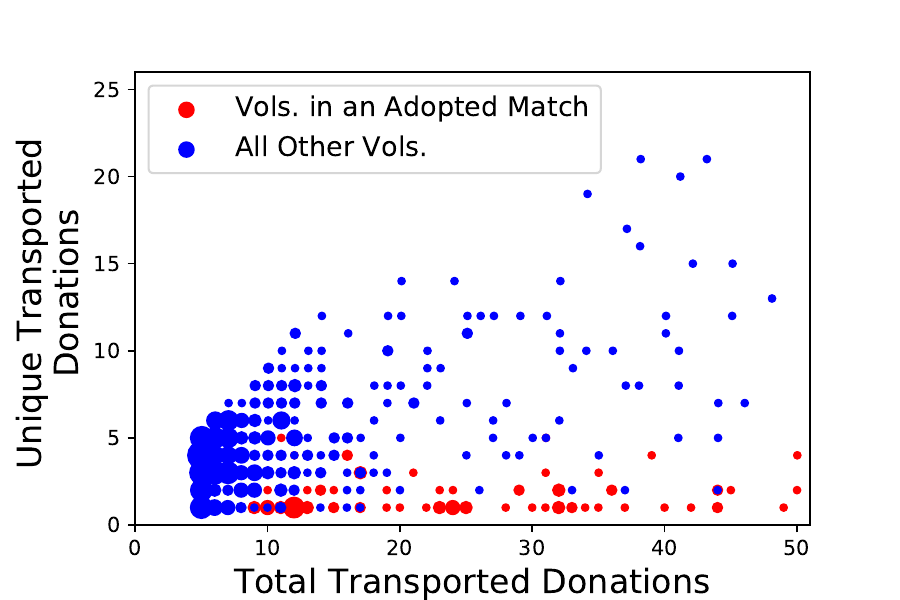}
 %\caption{ }\label{fig:summary1}
\end{subfigure}
\caption{ {\it Left:} Relative contributions of volunteers who are not part of an adopted match. {\it Right:} Heterogeneity of donations transported by volunteers who transported between 5 and 50 donations.}
\label{fig:summaries}
\end{figure}
 
{\bf Volunteers Engage Asymmetrically Based on Match Type:}
So far, we have shown evidence that failing to match negatively impacts both donors' and volunteers' future engagement. Consequently, adoption seems favorable as it reduces uncertainty in matching. However, our analysis of FRUS data also reveals a potential downside for adoption.
Volunteers who are in an adopted match are signed up automatically on a weekly repeating basis, and hence do not need to frequently engage with the platform.
Consequently, they may not be responsive to available donations in the spot market (donations unclaimed within a week of their scheduled pickup time).
These availabilities arise for several reasons, such as when a new donor joins the platform, and account for one-third of all FRUS donations.
Given the importance of successful matching for donor engagement, it is crucial to have volunteers willing to {engage with donations in the spot market. Volunteers who are not currently part of an adopted match} frequently play this vital role.

We demonstrate the differences in engagement based on match type in Figure \ref{fig:summaries}. The left panel shows the share of completed donations involving volunteers who are not part of an adopted match, as well as their share of completed {spot market donations}.
We present these shares on an aggregate level and for four large locations. As shown, volunteers who are not part of an adopted match play an outsized role in completing donations in the spot market.
The scatter plot in the right panel shows the number of unique donations vs. the total number of completed donations by both volunteers in an adopted match (red) and other volunteers (blue), focusing on those that have completed between 5 and 50 donations. Dot size is proportional to the number of volunteers who transported the corresponding number of total and unique donations.
We highlight that volunteers who are not part of an adopted match complete a more diverse set of donations. 
Together, these figures illustrate the importance of having volunteers who engage with the platform frequently.

\revcolor{In Appendix \ref{app:emp:engagement}, we further investigate the potential drawbacks of adoption. Even though we cannot directly observe volunteer engagement, we find suggestive evidence that in some FRUS locations, overall engagement in the subsequent week is higher for volunteers who are not in an adopted match. 
%However, we note that this analysis suffers from significant limitations because we do not directly observe volunteers' engagement unless they sign up for a donation and because we cannot control for the possibility of heterogeneous volunteers with different preferences for adoption. 
In addition, we show that locations with higher rates of adoption also tend to have higher rates of missed donations. 
%We suggest that the inefficiency of the spot market in locations with high levels of adoption is a contributing factor to this relationship.
}

Due to the benefits and disadvantages of adoption, the platform faces a potentially challenging trade-off between engagement and commitment when deciding on the adoption level.
To guide the platform's design, we develop a theoretical framework that captures these data-driven insights.

\subsection{Model Dynamics}
\label{subsec:dynamics}
The market starts at an initial state and evolves over $T$ periods. At the beginning of period $t \in [T]$, there are $D_t$ donations that each require one volunteer, and on the other side of the market, there are $V_t$ volunteers who are willing to transport one donation. For ease of notation, we use $[a]$ to refer to the set $\{0, 1, 2, \dots, a\}$ for any $a \in \mathbb{N}$. Since the platform allows for adoption, some donations and volunteers present at time $t$ have already been matched. In particular, $K_t$ pairs of donations and volunteers begin period $t$ as part of an adopted match, where $K_t \leq \min \{D_t, V_t \}$. 

%We assume that both \jobs \  and \workers \  are homogeneous. As a consequence, we do not need to know the status of individual \jobs \ or \workers. The state space at the start of period $t$ is then sufficiently described by the triplet $(D_t, K_t, V_t)$. 

{\bf Matching Process:}
We now describe how donations and volunteers match on the platform. In a given period, all \workers \ who are not already committed to a \job \ via adoption would like to match with one \job; however, \workers \  are not always compatible with the time or location of a particular \job. We assume that compatibility between \jobs \ and \workers \  is drawn identically and independently across pairs and time periods. %We model this as a compatibility graph which can be described by a random $0-1$ matrix $C_t(p)$ with dimension $(V_t-K_t) \times (D_t - K_t)$, where each entry has probability $p$ of being $1$.

During each period, all donations and volunteers that are not already part of an adopted match become part of that period's spot market, where \workers \  match with \jobs \ in a decentralized manner. The two sides of the spot market have size $D_t - K_t$ and $V_t - K_t$. We model the matching process in the spot market as sequentially greedy: (i) \workers \  arrive to the platform in a random order, (ii) upon arrival, a \worker \  checks the available \jobs \  for compatibility, and (iii) the \worker \  matches if they find a compatible \job \ (breaking ties at random). At the end of the period-$t$ matching process, a total of $M_t$ pairs of donations and volunteers have been matched, either from prior adopted matches or from the spot market in period $t$, where $K_t \leq M_t \leq \min \{D_t, V_t \}$.

{\bf Platform's Match Type Decision:}
\revcolor{The platform determines the fraction of matches in period $t$ which are allowed to be adopted. 
We denote this fraction by $z_t$. {For the sake of simplicity and sharpness of results, we assume} (i) that the platform can terminate previously-formed adopted matches and (ii) that volunteers will form an adopted match whenever it is allowed. Under assumption (i), the policy $z_t = 0$ means that the platform makes \emph{all} donations available in the spot market, even those which were previously part of adopted matches. Under assumption (ii), the policy $z_t = 1$ means that all matches are adopted matches. In Section \ref{sec:hetero}, we relax both assumptions and show that our main insights are unaffected.}

{\bf Donation Dynamics:}
We consider settings where \jobs \ recur, so the number of \jobs \ in the market in one period is closely related to the number of \jobs \ in the subsequent period. 
{In Section~\ref{subsec:data}, we found that} if a \job \ is not completed, it may {drop out} from the platform in the subsequent period. 
To reflect this in the model, we assume that \jobs \  will leave the market independently with probability $\beta$ if they are not completed in the previous period. 
On the other hand, successful matches are likely to increase the number of donations on the platform,
as growth on the FRUS platform (and similar platforms) is largely due to word-of-mouth and good publicity.  
To model this, for each match in period $t$, an additional \job \ enters the market in period $t+1$  with probability $\beta'-\beta>0$. Thus, the expected number of \jobs \  in period $t+1$ is 
\begin{align}
E[D_{t+1}|D_t, M_t] &= D_t + (\beta' -\beta) M_t - \beta(D_t - M_t) %\nonumber
= (1-\beta)D_t + \beta'M_t. \nonumber
\end{align}
Since no money is exchanged on FRUS or similar platforms, the number of matches can be viewed as measuring the surplus earned by the \job \ side of the market. Thus we may interpret our model as letting growth of the \job \ side of the market be proportional to that side's surplus, which is broadly consistent with the matching market literature \citep[e.g.,][]{lian2019optimal}.

{\bf Volunteer Dynamics:}
We now turn our attention to the evolution of the number of \workers \ on the platform. We found in Section \ref{subsec:data} that the future engagement of \workers \ with the platform depends {both} on whether the volunteer is matched (i.e. match success influences volunteer dropout) {and} on the type of their match (i.e. volunteers engage asymmetrically based on match type).
{Our model captures these effects using}
different dropout probabilities for the three possible \worker \ {matching} outcomes. (i) Unmatched volunteers drop out with probability $\alpha$. (ii) Volunteers who are part of an adopted match drop out with probability $\gamma$ (if a volunteer in an adopted match leaves the platform, we assume their corresponding \job \ will enter the spot market in the subsequent period). {W}e expect $\alpha > \gamma$, which is borne out in FRUS data (we omit the details for brevity).
 (iii) Volunteers who are part of a \temporary \ match drop out with probability $(\alpha - \alpha')$. The parameter $\alpha'$ captures the impact of being in a \temporary\ match on volunteer's future engagement with the platform, and it can be thought of as an increase in volunteers' retention or engagement (or a combination of the two). Though we expect $\alpha > \gamma$, the relationship between $\alpha - \alpha'$ and $\gamma$ depends on the characteristics of the volunteer base {and} may vary across markets. 
As we will see in Section \ref{sec:results}, this relationship plays a critical role in determining the optimal platform design. 

As with \job \ growth,
\worker \ growth is proportional to the surplus of the \worker \ side of the market:
as a result of each match a ``new \worker'' joins the market with probability $\gamma'$. %, i.e., growth is proportional to the surplus of the \worker \  side of the market.}
In the FRUS context, this new volunteer can also capture an increase in the engagement of existing volunteers. For example, if a volunteer in an adopted match  decides to volunteer one additional time per week, they will need to enter the spot market to find an additional match; we model this increase in engagement by introducing a new volunteer into the system.
%\ilcomment{Might be helpful to note that differences in dropout rates might also reflect a difference in growth rate between those in adopted vs one-time matches, e.g. `Note that the difference in dropout probabilities based on match type can also capture that volunteers have asymmetric future engagement and growth depending on their match type.'  or `Note that volunteers may exhibit asymmetric future engagement and growth depending on their match type; for simplicity we capture this using the different dropout probabilities, and  do not define different growth rates $\gamma'$ for each match type.'}
%\srcomment{This is a good point. I slightly adjusted your second suggestion, below. NOTE: we may need to fix table spacing.}
{Note that volunteers may exhibit asymmetric future engagement and growth depending on their match type; for simplicity, we capture this using the relationship between $\alpha - \alpha'$ and $\gamma$, and we do not define different growth rates $\gamma'$ for each match type.}

Based on these modeling choices, we can specify the expected number of both adopted matches and \workers \ in each period. As a result of the matching in period $t$ and the platform's adoption policy {$z_t$}, there are $z_tM_t$ adopted matches and $(1-z_t)M_t$ \temporary \ matches. Since volunteers in adopted matches will drop out with probability $\gamma$, the expected number of adopted matches at the start of period $t+1$ is $E[K_{t+1}|M_t, z_t] = (1-\gamma)z_tM_t$.
Similarly, the expected number of \workers \  at the start of period $t+1$ is given by
\begin{align}
E[V_{t+1}|V_t, M_t, z_t] =& (1-\alpha)(V_t - M_t) & \tag{\text{unmatched vols.}}\\ & + (1 + \gamma'-\gamma)z_tM_t& \tag{\text{vols. in adopted matches}} \\ &+ (1+\gamma'+\alpha'-\alpha)(1-z_t)M_t& \tag{\text{other matched vols.}} \\
=& (1-\alpha)V_t + (\alpha'+\gamma')M_t - (\gamma - \alpha + \alpha')z_tM_t.& \nonumber
\end{align}

We highlight that the expected number of \workers \ in period $t+1$ is increasing (decreasing) in the prior fraction of adopted matches $z_t$ if  $\gamma - \alpha+\alpha'$ is negative (positive). We refer to $\gamma - \alpha + \alpha'$ as the \emph{\gap}, since it represents the expected growth in the size of the volunteer pool if an adopted match is replaced with a \temporary \  match.

%Before proceeding, we highlight two key features of our volunteer growth model. First, note there is an asymmetry in our model between volunteers who form fixed matches and volunteers who form temporary matches. The former drop out with a lower rate, but do not increase their activity level in response to higher need. In context, this makes sense: volunteers who form temporary matches need to return to the website to sign up for available opportunities, whereas volunteers who are part of fixed matches are automatically signed up for the following week's donation and thus do not necessarily see the list of available rescues. 

%Second, note that the expected number of volunteers in period $t+1$ is increasing in $z_t$ if and only if $\alpha -\alpha' -\gamma$ is positive. Equivalently, promoting adoption helps grow the volunteer side of the market only if the dropout rate of volunteers in fixed matches (which must be positive) is less than the net dropout rate of volunteers in temporary matches (which may be positive or negative).

Having specified the state variables and their evolution, 
we next show that the state variables $\{D_t, K_t, V_t, : t \in [T]\}$ and the matching outcomes  $\{M_t: t \in [T]\}$ are concentrated around their expectations when the market size is large (Proposition \ref{prop:convergence}). 
Such concentration results enable us to approximate the multi-dimensional stochastic process by simple deterministic differential equations. 
This in turn provides the tractability needed to analyze the impact of the platform's decisions on the dynamics of the system. Such an approach is common when studying design questions in complex stochastic systems \citep[see, e.g.,][]{crapis2017monopoly, afeche2018ride}.
In the statement of this proposition, we use capital (resp. lower case) letters for our state variables (resp. approximations) to highlight that they are random (resp. deterministic) variables. \revcolor{These variables depend on the full history of the platform's decisions (e.g., $v_1$ is a function of $z_0$); however, we suppress this dependence for clarity of presentation.} \begin{proposition}[Deterministic Approximation for System Dynamics]
	\label{prop:convergence}
	For all $n > 1$, some fixed $c > 0$, and all $t \in [T]$, suppose that $(D^n_t, K^n_t, V^n_t)$ represents the state of the system in period $t$ assuming the initial conditions are given by $(nd_0, nk_0, nv_0) \in \mathbb{R}_+^3$, the platform decisions are given by $\mathbf{z}= \{z_t: z_t \in [0,1] \ \forall t \in [T]\}$, and the match probability is given by $\frac{c}{n}$.
	%\footnote{We note that for some $n$, it may not be possible to have a system with these initial conditions, e.g. $nd_0$ is not an integer. Thus a more accurate statement is to consider the subsequence $n_1, n_2, \dots$ such that the initial conditions are all integers and $\frac{c}{n_1} \leq 1$. However, for the sake of brevity we do not mention this in our statements. } 
	In addition, suppose that $M^n_t$ represents the number of completed \jobs \ in period $t$. Then for all $t \in [T]$, the following limits hold almost surely as $n$ approaches infinity:
	\begin{align}
	\frac{M_t^n}{n} \ \rightarrow \ m_t \ &= \s(d_t-k_t, v_t-k_t) + k_t \label{eq:mt}\\
	\frac{D_{t+1}^n}{n} \rightarrow d_{t+1} &= (1-\beta)d_t + \beta'm_t \label{eq:dt} \\
	\frac{K_{t+1}^n}{n} \rightarrow k_{t+1}&= (1-\gamma)z_tm_t \label{eq:kt} \\ \frac{V_{t+1}^n}{n} \rightarrow v_{t+1} &= (1-\alpha)v_t + (\alpha'+\gamma')m_t - (\gamma - \alpha + \alpha')z_tm_t \label{eq:vt}
	\end{align}
	where the matching function $\s(\cdot, \cdot)$ is defined in Proposition~\ref{prop:matching} in Section~\ref{subsec:matchingprocess}.
	%\footnote{We remark that $m_t$ is a deterministic function of $d_t, k_t$, and $v_t$. However, we drop the dependence in order to highlight that $m_t$ is a function of the platform's decision in the previous period, $z_{t-1}$.}
\end{proposition}
We highlight that in this large-market limit, the market thickness does not scale with $n$: the ex ante number of compatible matches for volunteers (resp. donations) in the spot market is given by $(D_t^n - K_t^n)\frac{c}{n} \rightarrow c(d_t - k_t)$ (resp. $(V_t^n - K_t^n)\frac{c}{n} \rightarrow c(v_t - k_t)$). Our approximation for the size of the matching in the spot market, given by $\s(d_t - k_t, v_t-k_t)$, depends on this market thickness (as opposed to the market size).
A proof of Proposition \ref{prop:convergence} can be found in Appendix \ref{proof:prop:convergence}.

\subsection{Platform Optimization Problem}
\label{subsec:opt}
\revcolor{
{To motivate} the platform's optimization problem, we provide an example {that highlights} the importance of the platform's adoption level decision. {In the two panels of Figure \ref{fig:market_death}, we illustrate the growth trajectories of the number of matches (left) and the fraction of missed donations (right)} resulting from three static policies: full adoption ($z = 1$), no adoption ($z = 0$), and partial adoption with $z = 0.5$. The model primitives for the instance shown (denoted $\mathcal{I}_1$) are given by $(\alpha, \alpha', \gamma, \gamma', \beta, \beta') = (0.35, 0.39, 0.04, 0.07, 0.11, 0.12)$, $v_0 = d_0 = 1$, $k_0=0$, and $c=5$.

{In the short-run, a static policy of full adoption ($z=1$, red solid lines) maximizes the number of matches and reduces missed donations due to the benefits of match certainty. However, in the long-run, a static policy of no adoption ($z=0$, blue dotted lines) leads to significantly more growth due to increased volunteer engagement and spot market thickness. In sharp contrast, a static policy which sets a fractional adoption level of $z = 0.5$ (green dashed lines) is significantly suboptimal and ultimately leads to the decay of the marketplace.}
}

\begin{figure}[t]
	\centering
	\includegraphics[width=.89\textwidth]{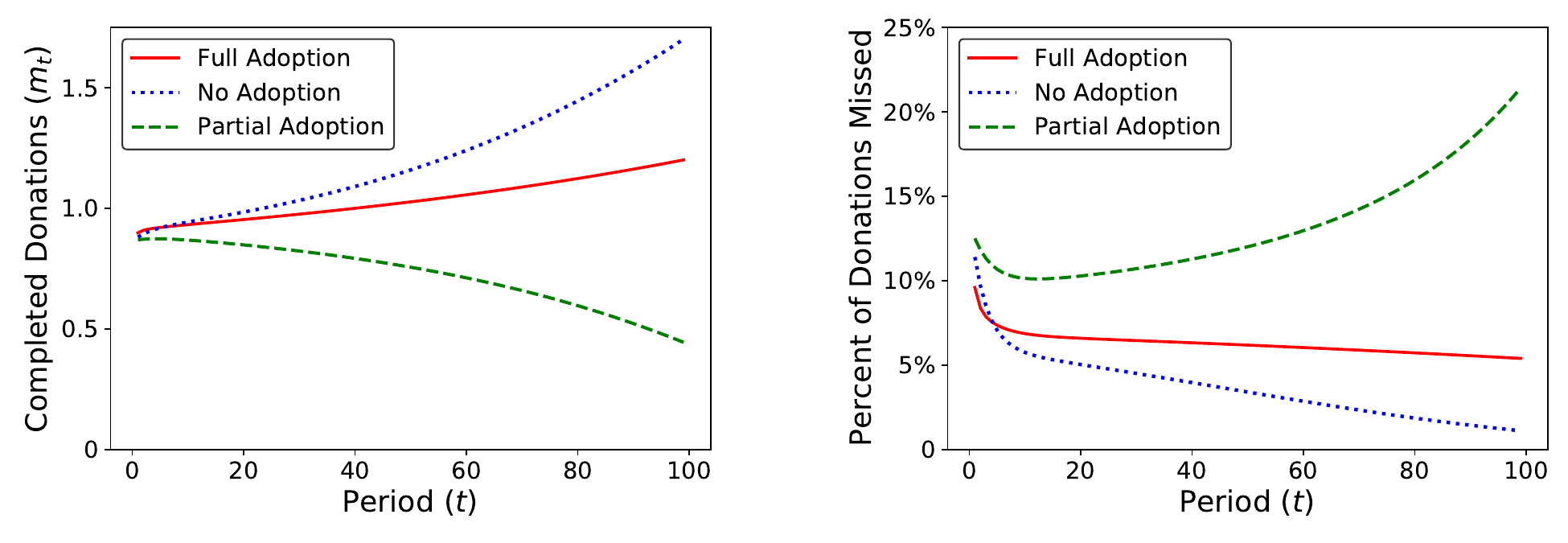}
	\caption{\revcolor{The trajectory of completed donations ({\it left}) and the trajectory of the percentage of missed donations ({\it right}) in instance $\mathcal{I}_1$ for static policies of full adoption ($z=1$), no adoption ($z = 0$), and partial adoption ($z = 0.5$).
	%For model primitives $(\alpha, \alpha', \gamma, \gamma', \beta, \beta') = (0.35, 0.39, 0.04, 0.07, 0.11, 0.12)$, $v_0 = d_0 = 1$, $k_0=0$, and $c=5$, the plot shows the trajectories of completed donations and missed donation percentage for three different static policies: full adoption ($z=1$), no adoption ($z = 0$), and partial adoption ($z = 0.5$).
	}}
	\label{fig:market_death}
		\end{figure}
%Ultimately, this can determine the survival of the nonprofit itself.

We now formalize how the platform can optimize its {adoption level} decision. For any scaled initial condition
$(d, k, v)$,
the \problem \ (PD), which we present below, is to decide on the allowable level of adoption in each period (i.e., $\{z_t: t \in [T]\}$) in order to maximize the total discounted number of completed donations between periods $t=1$ and $t = T$, subject to the deterministic dynamics described in Proposition \ref{prop:convergence}. 
The discount factor, $\delta$, captures the relative weights on current and future success. As in Proposition \ref{prop:convergence}, the state variables depend on the full history of the platform's decisions, but we suppress that dependence for {clarity} of presentation.

%\vspace{-10pt}
\begin{table}[h]
{\small
	$$\begin{array}{ccclr}
	\hline
	\max & & & \sum_{t=1}^T \delta^{t-1} m_t&(PD) \vspace{-10pt} \\
	^{\{z_t: t \in [T]\}}&&&& \\
	\text{s.t.} & m_t & = & k_t + \s(d_t-k_t, v_t-k_t) & \forall t \in [T] \\
	&d_{t+1} &=& (1-\beta)d_t + \beta'm_t &\forall t\in [T-1] \\
	&k_{t+1}&=& (1-\gamma)z_tm_t &\forall t\in [T-1]\\
	&v_{t+1} &=& (1-\alpha)v_t + (\alpha'+\gamma')m_t - (\gamma - \alpha + \alpha')z_tm_t & \qquad \forall t\in [T-1] \\
	&\multicolumn{3}{l}{d_0 = d, \qquad k_0 = k, \qquad v_0 = v, \qquad z_t \in [0,1]}&\forall t \in [T] \\
	\hline
	\end{array}$$
}
\end{table}

%\vspace{-20pt}

\begin{section}{Main Results}
	\label{sec:results}
	We now present our analysis of the \problem. In Section \ref{subsec:matchingprocess}, we derive a deterministic approximation for the size of the matching. In Section \ref{subsec:myopic}, we characterize the platform's optimal policy in a myopic version of the problem. Based on insights from the  myopic solution, in Sections \ref{subsec:adoptiondominance} and \ref{subsec:optiondominance} we study the \problem \ for two regimes depending on the relationship between $\gamma$ and $(\alpha-\alpha')$, i.e. the asymmetric engagement behavior of
	\workers \ in adopted matches {vs.} \temporary \ matches. We show a stark difference in the optimal policy for the two regimes.
	
	\subsection{Characterizing the Matching Function}
	\label{subsec:matchingprocess}
	
	In this section, we provide a characterization of the matching function when the matching process is sequential and greedy (as described in Section \ref{subsec:dynamics}). 
		%Our matching function exhibits \emph{increasing returns to scale at a diminishing rate}, a property that is not held by any member of the commonly-used class of CES production functions (which includes the Cobb-Douglas function and the fully-efficient matching function), while retaining other desirable properties of CES functions. 
		Our matching function admits a tractable closed form, and has desirable properties not captured by commonly-used matching functions which we discuss below.
The following proposition introduces our matching function and establishes that for any fixed market thickness, as the scaling factor $n$ grows, 
the  size of the matching is well-concentrated around its expected value, which is itself the unique solution to a differential equation.

	\begin{proposition}[Deterministic Approximation for the Matching Process]
		\label{prop:matching}
		Given two values $(a, b) \in \mathbb{R}_+^2$, suppose that for all $n>1$, {$S^n$} represents the number of matches in a market with sides of size $na$ and $nb$ where the match probability is given by $\frac{c}{n}$ and the matching is governed by a sequentially greedy process. Then the following limit holds almost surely:
		\begin{equation}
		\frac{S^n}{n} \rightarrow~ \s(a, b) {~:=~} a+b - \frac{1}{c} \log(e^{ca}+e^{cb}-1)
		\label{eq:matchingfunc}
		\end{equation}
	\end{proposition}

Proposition \ref{prop:matching} motivates our definition of the matching function $\s(\cdot,\cdot)$ for a two-sided market when matches are formed greedily and compatibilities are random and homogeneous. As such, it may have application in other studies of matching markets. 
{Note} that in the statement of Proposition \ref{prop:matching}, the size of the volunteer (resp., donation) side of the market has two components:  the scaling factor $n$ and the scaled size $a$ (resp., $b$). While market thickness is independent of the scaling factor, it does increase with the scaled size of each side.

Our matching function $\s(\cdot,\cdot)$ has several advantages over other popular matching functions, while retaining their desirable properties.  
First, it exhibits (i) increasing returns to scale (i.e. $\s(\kappa a, \kappa b) > \kappa \s(a, b)$ for any $\kappa > 1, a,b > 0$), a requisite for market thickness to improve match efficiency. Additionally, the physical constraints of one-to-one matching imply that if there is randomness in compatibility, then (ii) increasing both sides of the market by the same constant should increase the size of the matching by less than the constant (i.e. $\frac{d }{d \epsilon} \s(a+\epsilon, b+\epsilon) < 1$).
%This critical property ensures that the certainty provided by an adopted match is valuable relative to a corresponding increase in the size of the spot market.
Our matching function satisfies both properties, and hence exhibits increasing returns to scale \emph{at a diminishing rate}.
As the market becomes thicker, it approaches the fully-efficient matching function: $\lim_{a,b \rightarrow \infty}\s(a,b) = \min\{a,b\}$. Thus, our matching function is consistent with a market that gets more efficient as it becomes thicker, while respecting the physical {constraint} that the size of the matching cannot exceed the size of either side.

One commonly-used class of functions is the constant elasticity of substitution (CES) production functions \citep{petrongolo2001looking}, which include the Cobb-Douglas function ($\mu_1(a,b) = a^\rho b^{\omega}$ for some $\rho, \omega \geq 0$), and the fully-efficient matching function ($\mu_2(a,b) = \min\{a,b\}$). Many functions in this class are increasing and concave in the size of each side of the market, and result in an empty matching when one side of the market is empty. Our matching function retains these desirable properties. 
However,  unlike  our matching function,
no CES function can satisfy both properties (i) and (ii) defined above.

The proof of Proposition \ref{prop:matching} is a generalization of Theorem 3 in \citet{mastin2013greedy} and makes use of stochastic differential approximation techniques from \citet{wormald1999models}. \revcolor{The full proof details are presented in Appendix \ref{proof:prop:matching}. Further, we numerically show the accuracy of this deterministic approximation for markets of moderate size in Appendix \ref{app:numerics}.}

	%\begin{lemma}
	%Let $f(n(D_t-K_t), n(V_t-K_t), C_t(\frac{c}{n}))$ be the size of the matching if $n(V_t-K_t)$ homogeneous volunteers match in a greedy fashion with $n(D_t-K_t)$ homogeneous donors, where each volunteer is compatible with $c(D_t - K_t)$ donors in expectation. Then a.a.s.
	%\begin{equation}
	%\frac{f(n(V_t-K_t), n(D_t-K_t), \frac{c}{n})}{n} = g(V_t-K_t, D_t-K_t)
	%\end{equation}
	%\end{lemma} 
	
	%A proof can be found in the appendix. In Section ZZZ we provide numerical evidence that $f(V_t-K_t, D_t-K_t, C_t(p))$ is well concentrated around $g(V_t-K_t, D_t-K_t)$ for reasonable values of the market size and the compatibility probability. 

	%At the conclusion of the matching process, we define the total number of matches as $M_t := K_t + f(D_t-K_t, V_t-K_t, p)$. 
	%On a related note, recall that we model the compatibility probability as a new random variable in each period, even if an available volunteer and donation were matched in the previous period. This modeling choice also slightly overstates the value of fixed matches but does not significantly impact our results.

	%Section

	\begin{subsection}{Optimal Myopic Policy}
		\label{subsec:myopic}
		In this section, we consider the myopic 
		problem where $T = 1$, or equivalently 
		$\delta=0$. In such a setting, the platform chooses the adoption level $z_0$ to maximize
		$m_{1}(z_0)$. Note that, with a slight abuse of notation, we augment $m_1$ by $z_0$ to highlight the role of the decision variable. 
		
		Before presenting the structure of the optimal policy, we provide some intuition. First note that $z_0$ does not impact the number of  completed donations in period $0$ (accordingly, $d_1$ as defined in \eqref{eq:dt} is independent of $z_0$). However, $z_0$ impacts $m_1$ in three ways. (i) The match type influences the growth and engagement of the \worker \ pool (as specified in \eqref{eq:vt}).
		{(ii) One-time matches increase the thickness of the spot market in the subsequent period, which improves match efficiency in this spot market. However, (iii) adoption replaces the randomness in matching with certainty.}
		To compare the relative impact of effects (ii) and (iii), note that with certainty, each adoption {carried over to the next period} 
		is  worth one match;
		however, increasing both sides of the spot market by one in the next period is worth less than one additional match due to randomness in matching. 
		Thus if adopted matches are preferable for growing the \worker \ pool (i.e., 
		the \gap\  $\gamma - \alpha + \alpha'$ is negative), we would expect a policy of only adoption ($z_0 = 1$) to be optimal. 
		
		\revcolor{On the other hand, suppose the \gap\ is positive, meaning that there is a trade-off between volunteer engagement and the match certainty provided by adoption. We claim that even in such settings, an interior adoption level $z_0\in (0,1)$ cannot be optimal. For intuition as to why, suppose that the spot market is perfectly efficient for volunteers (i.e., the number of matches increases by $1$ for every additional volunteer in the spot market). In that case, there is no benefit to the certainty of an adopted match. Thus, it would be better to increase volunteer engagement by disallowing adoption. Of course, random matching is not perfectly efficient; 
		however, if less adoption is allowed, the spot market gets thicker and the matching becomes more efficient.
		This is an additional positive impact of decreasing the adoption level, and it is complementary to the positive impact of the \gap: the total matching is essentially a supermodular function of (i) the engagement of volunteers and (ii) the efficiency of the spot market. These qualitative observations suggest that the number of matches in the next period is quasiconvex in $z_0$, meaning that an interior adoption level $z_0 \in (0,1)$ will never be optimal.}

		We illustrate the above intuition using two numerical examples where the \gap \ is positive.
		In the left panel of Figure \ref{fig:m1z0}, the optimal myopic policy is allowing for full adoption ($z_0 = 1$) because even completely disallowing adoption does not generate sufficient market thickness for volunteers in period $1$. %In other words, the expected number of compatible matches for volunteers in the spot market $c(d_1-k_1)$ is too small even when adoption is disallowed and $k_1=0$.
		%spot market is not thick enough and exhibits poor match efficiency. }
		%In that example, the optimal myopic policy is to have only adoption ($z_0 = 1$) even though the \gap \ is positive. 
		In the right panel, we consider the same model primitives but with {a thicker} market for volunteers, {and find that} the optimal myopic policy is to disallow adoption ($z_0 = 0$). In the following theorem, we confirm that the myopically optimal policy is either allowing adoption for all donations or disallowing adoption.

		\begin{theorem}[Optimal Myopic Policy]
			\label{thm:myopic}
			Given an initial state $(d_0, k_0, v_0)$, the optimal myopic policy is either allowing all donations to be adopted or disallowing adoption. In particular,
			\begin{equation}
			z_0^* = \begin{cases}1, &\text{if} \quad \gamma\leq \alpha-\alpha' \quad \text{or} \quad  cd_1 \leq \log\left(\frac{ e^{c({D_k-D_v})} -1}{e^{-c{D_v}} - 1}\right) \\
			0, &\text{otherwise} 
			\end{cases} \label{eq:optmyopic}
			\end{equation}
			where $D_k := \frac{\partial k_1}{\partial z_0}= (1-\gamma)m_0$ and $D_v := \frac{\partial v_1}{\partial z_0}= (\alpha - \alpha' - \gamma)m_0$.
		\end{theorem}

		\begin{figure}[t]
			\centering
			\begin{subfigure}[b]{0.45\textwidth}
			\includegraphics[trim=0 34 0 0, clip, width=\textwidth]{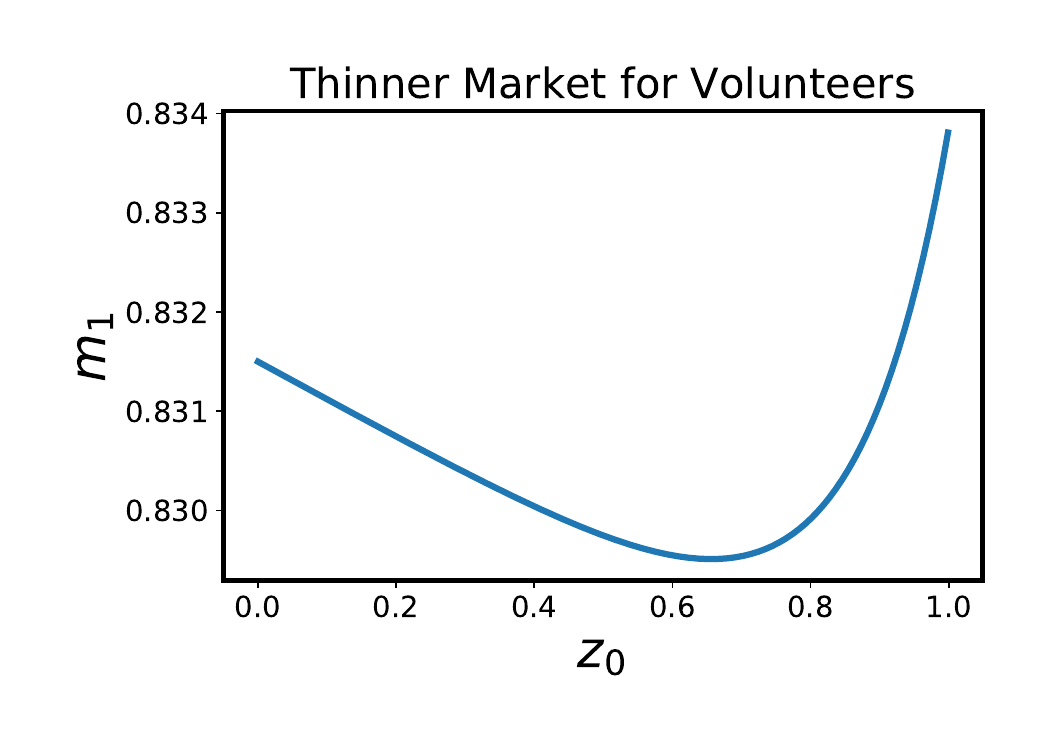}
				%\caption{ }
				%\label{fig:m1z0ex1}
			\end{subfigure}
			\begin{subfigure}[b]{0.45\textwidth}
				\centering
					\includegraphics[trim=0 34 0 0, clip, width= \textwidth]{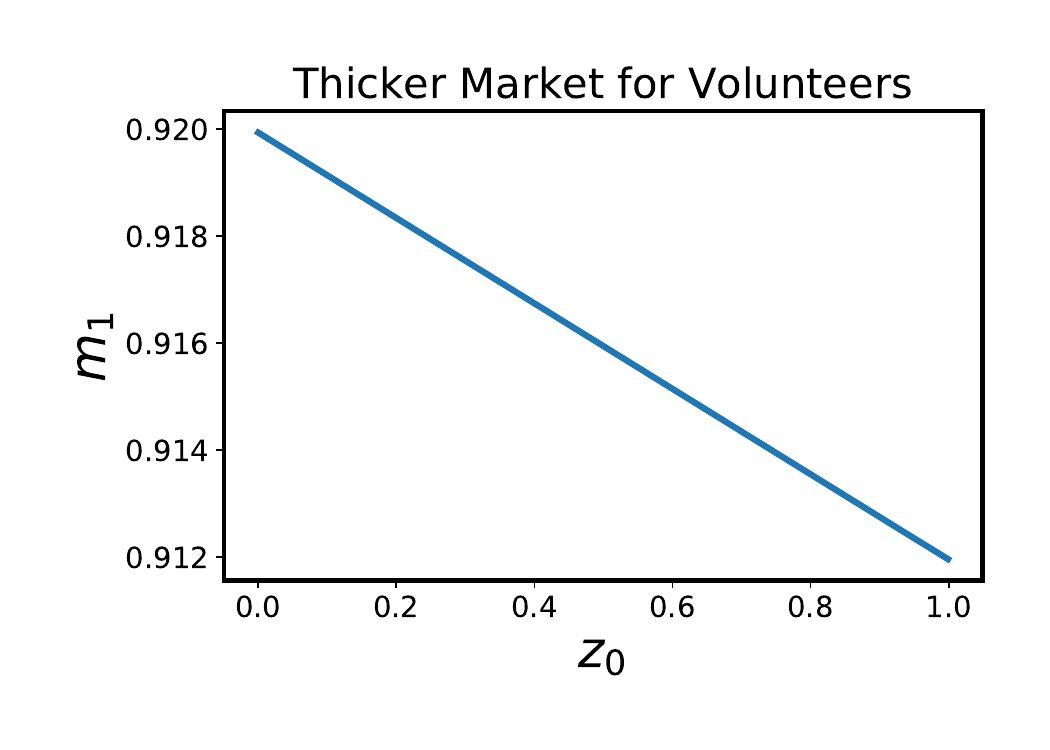}
				%\caption{ }
				%\label{fig:m1z0ex2}
			\end{subfigure}
			\caption{Plotting $m_1(z_0)$ when $(\alpha, \alpha', \gamma, \gamma', \beta, \beta') = (0.2, 0.15, 0.06, 0.2, 0.05, 0.2)$ and $c=10$. {\it Left:} for initial state $(d_0, k_0, v_0) = (0.8, 0, 0.8)$. {\it Right:} for initial state $(d_0, k_0, v_0) = (1.6, 0, 0.8)$.}
			\label{fig:m1z0}
		\end{figure}
		
		If the \gap \ is negative
		(i.e. if $\gamma\leq \alpha-\alpha'$),
		the \worker \ pool grows more per adopted match than per \temporary \ match. When adoption provides both growth and certainty, the platform prefers adopted matches regardless of the state. In Section \ref{subsec:adoptiondominance}, we show that in this regime, a static policy of full adoption is also optimal when the discount factor is nonzero. 
		
		When 
		the \gap \ is positive, a policy of no adoption maximizes the number of volunteers in the subsequent period. However, such a policy is only myopically optimal if the {thickness of the spot market for volunteers with no adoption ($cd_1$)} is above a threshold. {This threshold depends on $D_k$ and $D_v$, which represent the difference in the number of adopted matches ($k_1(1)-k_1(0)$) and the size of the volunteer pool ($v_1(1)-v_1(0)$), respectively, when following a policy of full adoption as opposed to no adoption.} {Note that the threshold is constant in $z_0$, and is decreasing in the \gap\ (via its dependence on $D_v$).}
		
		To prove Theorem \ref{thm:myopic}, we first show that $m_{1}(z_0)$ is quasiconvex,
		implying that the maximum is an extreme point.
		We then show that $m_{1}(1) \geq m_{1}(0)$ either when $\gamma \leq \alpha-\alpha'$ (i.e., when $D_v > 0$) or when $cd_1\leq \log\left(\frac{ e^{c({D_k-D_v})} -1}{e^{-c{D_v}} - 1}\right)$. A complete proof can be found in Appendix \ref{proof:thm:myopic}. In Appendix \ref{app:numerics}, we numerically show that the structure of the optimal myopic policy holds even when we consider the actual expected size of the matching (as opposed to its deterministic approximation).
		
	\end{subsection}

	\begin{subsection}{Optimal Policy in the Adoption Growth Dominance Regime}
		\label{subsec:adoptiondominance}
		In this section, we focus on when %instances where 
		%$\gamma \leq \alpha - \alpha' $
		{the \gap \ is negative (i.e., $\gamma -\alpha+\alpha'\leq 0$)}, which we call the \emph{Adoption Growth Dominance (AGD)} regime. In this regime, increasing adoption increases the {total engagement} of \workers \ in the subsequent period, i.e., $v_{t+1}(z_t)$ as defined in \eqref{eq:vt} is increasing in $z_t$. In addition, Theorem \ref{thm:myopic} states that in {this regime}, $m_{t+1}(z_t)$ is maximized at $z_t = 1$. Thus, there is no trade-off between long-term market growth and short-term certainty in matching, as  increasing the adoption fraction will improve both objectives.

		\begin{theorem}[Optimal Policy in the \regimeA]
			\label{thm:optregimea}
			In the \regimeA \ (i.e., $\gamma \leq \alpha - \alpha'$), the optimal policy is $z^*_t = 1$ for all $t \in [T]$. 
		\end{theorem}
		
		Theorem \ref{thm:optregimea} establishes that a policy of only adoption is optimal in the \regimeA. This implies unless \temporary \ matches provide gains in \worker \ engagement (and hence growth), the platform should promote adoption. To prove Theorem \ref{thm:optregimea}, we use Theorem \ref{thm:myopic} to show that in the \regimeA, the optimal myopic policy is always $z_t = 1$.  
		We then establish that when the optimal myopic policy  increases the size of the \worker \ pool, it is also the optimal policy in the long-run (Lemma~\ref{lemma:myopicisbest}), which completes the proof of Theorem \ref{thm:optregimea}. \revcolor{We defer the proof of Lemma~\ref{lemma:myopicisbest} to Appendix \ref{proof:lemma:myopicisbest}.}
		\begin{lemma}[Optimality Condition for the Myopic Policy]
			\label{lemma:myopicisbest}
			Let $z_\tau'$ be the optimal myopic policy given state $(d_\tau, k_\tau, v_\tau)$. Then $z_\tau'$ is the optimal policy in period $\tau$ if  $z_\tau' \in \text{argmax}_{z_\tau \in [0,1]} v_{\tau+1}(z_\tau)$.
		\end{lemma}
	\end{subsection}
	
	\begin{subsection}{Optimal Policy in the Adoption Growth Non-Dominance Regime}
		\label{subsec:optiondominance}
		As discussed above, in the \regimeA, there is no trade-off between short-term certainty in matching and long-term market growth.  We now turn our attention to the regime where \temporary \ matches are better for long-term market growth, i.e., when 
		%$\gamma > \alpha - \alpha'$.
		the \gap \ is positive.
		We call this the \emph{Adoption Growth Non-Dominance (AGN)} regime. In this regime, the market can choose to promote the growth of the \worker \ side of the market by reducing adoption and correspondingly increasing the size of the spot market in the subsequent period. 
		
		As a consequence of Lemma \ref{lemma:myopicisbest}, when $z_t = 0$ is myopically optimal in the \regimeB, disallowing adoption is the optimal policy in period $t$. Further, as established in Theorem \ref{thm:myopic}, $z_t = 0$ is myopically optimal if the \gap \ is positive and market thickness for volunteers  
		{$(cd_1)$} 
		exceeds a threshold.  
		%Consequently, no adoption is optimal in the \regimeB  \ if the spot market is thick enough that the benefits of short-term certainty are immediately offset by the benefits of \worker \ growth. 
		If the market thickness continues to grow beyond that threshold, a policy of no adoption will remain optimal. Based on this insight, in the following theorem,
		we characterize the optimal policy in the \regimeB \ when the market is sufficiently thick.
		
		\begin{theorem}[Optimal Policy in the \regimeB]
			\label{thm:optregimeB}
			In the \regimeB \ (i.e., $\gamma > \alpha -\alpha'$), if $d_\tau \geq \bar{d}$, and $v_\tau \geq \bar{v}$, then $z_t = 0$ is optimal for all $t \in [T] \setminus [\tau-1]$, where $\bar{d}$ and $\bar{v}$ are the unique solutions to the \emph{minimum market thickness} program defined in \emph{(MMT)}.
		\end{theorem}
%\vspace{-15pt}		
\begin{table}[h]
	$$\begin{array}{llr}
	\hline
	\min_{d,v} & \ \ \quad d &\text{(MMT)} \\
	\ \text{s.t.}& \quad cd \quad \geq \quad \frac{1}{1-\beta}\log\Big(\frac{1+\alpha'-\alpha}{\gamma +\alpha' - \alpha}\Big) & (7) \\ & \quad cd \quad \geq \quad  \frac{1}{1-\beta}\log \Big(\frac{e^{c(1+\alpha' -\alpha)d}-1}{e^{c(\gamma+\alpha'-\alpha )d}-1} \Big) - \frac{\beta'}{1-\beta}cd &(8) \\
	&\quad cd  \quad \geq \quad \log \Big(\frac{e^{cv}-1}{e^{cv(1-\alpha/(\alpha'+\gamma'))}-1} \Big) & (9) \\ & \quad cv \quad \geq \quad \log \Big(\frac{e^{cd}-1}{e^{cd(1-\beta/\beta')}-1} \Big) & (10) \vspace{5pt} \\
	\hline
	\end{array}$$
\end{table}
\stepcounter{equation}
\stepcounter{equation}
\stepcounter{equation}
\stepcounter{equation}

%\vspace{-15pt}
{Theorem \ref{thm:optregimeB} shows that in the \regimeB, once the market gets sufficiently thick, a static no-adoption policy is optimal for the remainder of the time horizon.} The proof of Theorem \ref{thm:optregimeB} relies on appropriately interpreting the constraints in (MMT). Constraints (7) and (8) are sufficient conditions to ensure \temporary \ matches are myopically optimal (according to Theorem \ref{thm:myopic}).  Constraints (9) and (10) ensure that the number of \workers \ and \jobs will increase in the subsequent period assuming there is no adoption. We use these constraints to show that $ \{(d, v) : d \geq \bar{d}, v \geq \bar{v} \}$ is an absorbing state. When applied iteratively, this proves that $z_t = 0$ is optimal for all $t \in [T] \setminus [\tau-1]$. In Appendix \ref{proof:thm:optregimeB}, we present a rigorous proof of Theorem \ref{thm:optregimeB} which also demonstrates that the unique solution to (MMT) can be found via linear search.

To further illustrate the role of market thickness, in Figure \ref{fig:arrow_diagram} we provide an example for how the myopically optimal policy compares to the optimal policy when the discount factor is non-zero {in the \regimeB. At each point $(d,v)$, we consider an initial condition of $(d_0,k_0, v_0)=(d,0,v)$. The model primitives for this instance (denoted instance $\mathcal{I}_2$) are given by  $(\alpha, \alpha', \gamma, \gamma', \beta, \beta') = (0.2, 0.15, 0.06, 0.2, 0.05, 0.2)$, $c=10$, $\delta = 0.6$, and $T = 15$. To the left of the solid curve, $z_0 = 1$ is myopically optimal, while to the right, $z_0 = 0$ is myopically optimal. To the right of the curve, the myopically optimal policy is also exactly optimal (i.e., $z_0^* = 0$, as shown by the red arrows). However, to the left of the curve, the myopically optimal policy of allowing for adoption may be sub-optimal due to the negative impact of adoption on volunteer engagement.  We note that the magnitude and direction of each arrow indicates how the size of the \worker \ and \job \ pools will change in the subsequent period when following the optimal policy $z_0^*$.}
		
		\begin{figure}[ht]
			\centering
			\includegraphics[trim={0cm 1.0cm 0cm 0cm},clip,
			%scale = 0.3
			width = .9\linewidth]{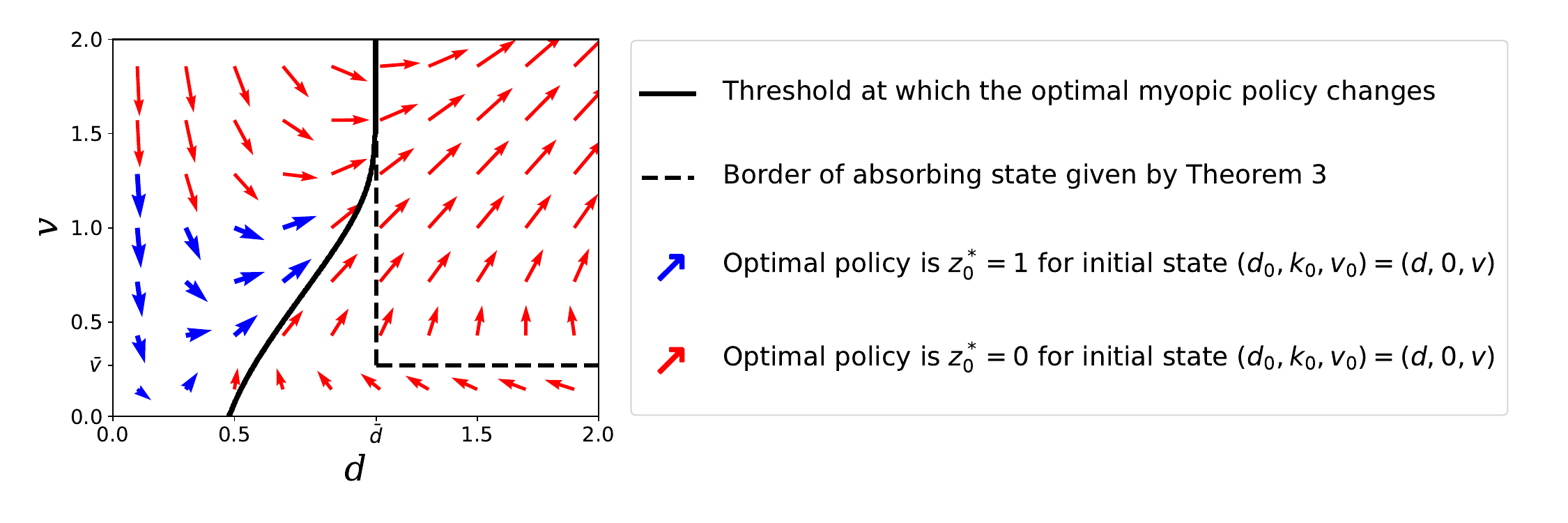}
			\caption{
			Comparing the myopically optimal policy to the optimal policy in instance $\mathcal{I}_2$.
			%For model primitives $(\alpha, \alpha', \gamma, \gamma', \beta, \beta') = (0.2, 0.15, 0.06, 0.2, 0.05, 0.2)$, $c=10$, $\delta = 0.6$, and $T = 15$, the plot describes the following features: 
			%Comparing the myopically optimal and optimal policies, for model primitives $(\alpha, \alpha', \gamma, \gamma', \beta, \beta') = (0.2, 0.15, 0.06, 0.2, 0.05, 0.2)$, $c=10$, $\delta = 0.6$, and $T = 15$,
			%the threshold where the optimal myopic policy changes (solid curve); the absorbing region defined by $\bar{d}$ and $\bar{v}$ (dashed lines); 
			%at each point $(d,v)$, for an initial state $(d_0, k_0, v_0) = (d, 0, v)$, 
			%the long-run optimal policy given initial state $(d, 0, v)$ (arrow color, blue for $z_0^* = 1$,  red for $z_0^* = 0$); and the direction of the state transitions when following the optimal policy (arrow magnitude and direction).
			}
			\label{fig:arrow_diagram}
		\end{figure}

		If the market thickness is below the threshold specified by Theorem \ref{thm:optregimeB}, the optimal decision proves difficult to characterize, since the problem is not convex in the decision variables. However, our numerical analysis indicates that the optimal decisions are always $z_t = 0$ or $z_t=1$, resembling the {structure of the} optimal myopic policy. Further, if the horizon is long enough, we numerically observe that the optimal policy either disallows adoption or it initially allows for full adoption before switching and disallowing adoption (i.e., for some $\tau$, the optimal policy is $z_t^* = 1$ for all $t \in [\tau]$ and $z_{t}^* = 0$ for all $t \in [T] \setminus [\tau]$). In these cases, even though we cannot theoretically characterize the optimal policy, we can offer a guarantee on the performance of a static policy of no adoption. 
		
		\begin{theorem}[Approximately Optimal Policy in the \regimeB]
			\label{thm:approxregimeB}
			{Suppose the discount factor is small enough such that $r:= \delta(1+\max\{\beta'-\beta, \gamma'-\alpha+\alpha'\}) < 1$.}
Then, in the \regimeB \ (i.e., $\gamma > \alpha -\alpha'$), the policy $z_t  = 0$ for all $t \in [T]$ is $1-\kappa$ optimal, where
			$$\kappa =  \min\left\{\frac{ \frac{1}{1-r^T} + \frac{\max\{ \beta', \alpha'+\gamma'\}}{1-r} }{c \min\{d_0, v_0\}} \log 2, 1\right\}. $$
		\end{theorem}
		
		Theorem \ref{thm:approxregimeB} establishes a ratio between the performance of a no-adoption policy and the optimal policy. The ratio approaches $1$ as the thickness {for both sides of the market, namely $c\min\{d_0,v_0\}$,} grows, 
		and  improves the more the platform discounts the future. If 
		$r < 1 - \frac{1+\max\{ \beta', \alpha'+\gamma'\}}{c \min\{v_0, d_0\}} \log 2$, 
		then our approximation factor is non-trivial. In the example shown in Figure \ref{fig:arrow_diagram}, a policy of no adoption at the initial condition $(d_0, k_0, v_0) = (0.5, 0, 0.5)$ achieves an approximation ratio of $0.704$.
		
		To prove Theorem \ref{thm:approxregimeB}, we upper bound the achievable number of completed \jobs \ by considering a setting where the matching process is perfectly efficient, i.e., the number of matches is the minimum of the two sides of the spot market. Using this matching function, we exactly characterize the optimal policy as $z_t^* = 0$ for all $t \in [T]$. We then compare this upper bound on the number of completed \jobs \  to the number of completed \jobs \ when the matching process is governed by \eqref{eq:matchingfunc} {under} the same no-adoption policy. The full proof can be found in Appendix \ref{proof:thm:approxregimeB}.

	\end{subsection}
\end{section}
\revcolor{
\section{Extensions}
\label{sec:hetero}
In our base model, for the sake of simplicity and sharpness of results, we {made} the assumptions that: 
(i) volunteers are homogeneous in their engagement behavior,
(ii) if a donation is allowed to be adopted, the volunteer will adopt it, and (iii) the platform can dissolve an adopted match at any time. In Section \ref{subsec:heteromodel}, we enrich our base model by relaxing these three assumptions. 
Through analytical and numerical results, we illustrate that our key qualitative insights hold in this richer model. 
In particular, in Section \ref{subsec:heteromyopic}, we show that the optimal myopic policy again constitutes either enabling adoption for all donations or disallowing adoption.
Further, in Section \ref{subsec:heterolongrun}, we show that repeatedly following the optimal myopic policy performs well under some conditions.

\subsection{Modeling Heterogeneous Volunteers}
\label{subsec:heteromodel}
Consider a model 
%similar to the one described in Section \ref{subsec:dynamics}, but 
with two types of volunteers: ``professionals'' and ``students''. 
We augment the notation from our base model by adding the superscript $p$ or $s$ (respectively) to indicate the volunteer type.
These volunteers differ in their engagement when part of a \temporary\ match (i.e., $\alpha'^p \neq \alpha'^s$). In addition, we allow volunteers to have more control over the types of matches they form.
As in our base model,
we let $z_t$ denote the %and (ii) their willingness to sign up for an adopted match.To model the latter, we first relax our assumption that the platform has complete control over the types of matches formed in each period. Instead, we assume that the platform's decision $z_t$ represents the 
fraction of donations in period $t$ for which the platform allows adoption. {V}olunteers greedily match with the first compatible donation in the spot market, regardless of whether adoption is allowed. We assume that volunteers can choose to not adopt a match, and their willingness to {adopt} is type-dependent{:} when a professional (resp. student) forms a match with a donation and adoption is allowed, they form an adopted match with probability $\adoptp^p$ (resp. $\adoptp^s$), and a one-time match otherwise. Once an adopted match is formed, the platform cannot dissolve the match; the match remains until it dissolves naturally.

The deterministic dynamics of this extended model, described below, mirror the dynamics given by eqs. \eqref{eq:mt}-\eqref{eq:vt} presented in Proposition \ref{prop:convergence}. (Note that the proof of Proposition \ref{prop:convergence} can easily be adjusted to prove convergence in this setting as well.) {T}he number of adopted matches {reflects} two differences compared to our base model. First, previously-established adopted matches cannot be dissolved by the platform. Second, newly-formed matches are adopted matches if and only if adoption is allowed for the donation (which happens with probability $z_t$) and the volunteer is willing to participate in an adopted match (which happens with probability $\adoptp^j$). Thus,
\begin{equation}
    k_{t+1}^j = (1-\gamma) \left(k_t^j + z_t\adoptp^j(m_t^j-k_t^j)\right), \qquad j \in \{p,s\}. \label{eq:kthetero}
\end{equation}
The total number of adopted matches is the sum of the adopted matches of both volunteer types, i.e., $k_{t} = k^p_{t}+k^s_{t}$. Similarly, we augment the volunteer dynamics by taking into account both the type-specific engagement behavior and the number of adopted matches formed:
\begin{equation}
    v_{t+1}^j = (1-\alpha)v_t^j + (\alpha'^j+\gamma')m_t^j - (\gamma - \alpha + \alpha'^j)\left(k_t^j + z_t\adoptp^j(m_t^j-k_t^j)\right), \qquad j \in \{p,s\}. \label{eq:vthetero}
\end{equation}
The total size of the volunteer pool is given by $v_{t} = v^p_{t}+v^s_{t}$. Since the donor dynamics are unchanged from equation \eqref{eq:dt}, all that remains is to specify the number of matches involving volunteers of each type. This includes both adopted matches and a fraction of the spot market matches. We assume that volunteers arrive to the spot market in a random order (independent of type) and greedily form matches. This leads to volunteers of type $j$ forming a number of matches in the spot market which is proportional to their prevalence in the volunteer side of the spot market,
\begin{equation}
    m_t^j = \left(\frac{v_t^j - k_t^j}{v_t-k_t}\right)\s(d_t-k_t, v_t-k_t) + k_t^j, \qquad j \in \{p,s\}. \label{eq:mthetero}
\end{equation}
The total number of matches is simply {given by} $m_{t} = m^p_{t}+m^s_{t}$. 

{N}ote that we recover our base model by setting $\alpha'^p = \alpha'^s$ and $\adoptp^p = \adoptp^s = 1$, and by allowing the platform to dissolve adopted matches at any point. 
{In addition,} because a volunteer's type is not necessarily observable, we assume that the platform cannot discriminately allow adoption only for a certain type of volunteer, and instead must make a single decision $z_t$ for all volunteer types.

\subsection{Optimal Myopic Policy (Heterogeneous Volunteers)}
\label{subsec:heteromyopic}
Consider the myopic problem, where the platform chooses $z_0$ (the fraction of donations for which adoption is allowed) to maximize $m_1$. Even in this richer model, the optimal myopic policy is always either $z_0 = 1$ (i.e., adoption always allowed) or $z_0 = 0$ (i.e., adoption disallowed).

To see why, note that many important features of the myopic problem remain identical to our base model. First, the matching function remains the same. Second, the number of adopted matches and the size of the volunteer pool are both still linear in the fraction of donations for which adoption is allowed (i.e., $k_1$ and $v_1$ are both linear in $z_0$). Essentially, even though there are two types of volunteers, the next-period dynamics are governed as though there were a single \emph{expected} volunteer type, with adoption probability equal to the average adoption probability of volunteers in the spot market and \gap\ equal to a weighted average of the \gap\ of volunteers in the spot market. Due to these similarities, the same intuition described in Section \ref{subsec:myopic} suggests that either allowing all donations to be adopted or disallowing adoption will be optimal. We formalize this intuition in the following proposition.

\begin{proposition}[Optimal Myopic Policy with Heterogeneous Volunteers]
			\label{prop:myopichetero}
			Given an initial state $(d_0, k_0^p, k^s_0, v_0^p, v_0^s)$, the optimal myopic policy is either allowing all donations to be adopted or disallowing adoption. In particular,
			\begin{equation}
			z_0^* = \begin{cases}1, &\text{if} \quad \gamma \leq \alpha - \frac{\sum_{j \in \{p,s\}}  \adoptp^j(v_0^j - k_0^j)  \alpha'^j}{\sum_{j \in \{p,s\}}\adoptp^j(v_0^j - k_0^j)~~~}  \quad \text{or} \quad  cd_1 \leq (1-\gamma)k_0 + \log\left(\frac{ e^{c(\hat{D}_k - \hat{D}_v)} -1}{e^{-c\hat{D}_v} - 1}\right) \\
			0, &\text{otherwise} 
			\end{cases} \label{eq:optmyopichetero}
			\end{equation}
		where $\hat{D}_k :=  \frac{\partial k_1}{\partial z_0} = (1-\gamma)\sum_{j \in \{p,s\}}\adoptp^j(m_0^j - k_0^j)$ and $\hat{D}_v := \frac{\partial v_1}{\partial z_0} = \sum_{j \in \{p,s\}}(\alpha - \alpha'^j-\gamma)\adoptp^j(m_0^j - k_0^j)$.
		\end{proposition}
	
The above optimality conditions resemble the corresponding conditions in Theorem \ref{thm:myopic}.
These conditions establish that a policy of disallowing adoption is optimal if and only if a weighted version of the \gap\ is positive and if the thickness of the spot market for volunteers (given no additional adoption) is above a threshold. The threshold is decreasing in the \gap\ for each type of volunteer (via its dependence on $\hat{D}_v$). The derivatives $\hat{D}_k$ and $\hat{D}_v$  mirror $D_k$ and $D_v$ in our base model, and represent the difference in $k_1$ and $v_1$ (respectively) when following a policy of fully allowing adoption as opposed to disallowing adoption.

The weighted \gap\ which comprises the first condition in \eqref{eq:optmyopichetero} depends on both the fraction of volunteers in the spot market of each type as well as their willingness to form an adopted match. To illustrate this weighted \gap, consider an instance where professionals have negative \gap\ and students have positive \gap. If students self-select into \temporary\ matches (i.e., if $\adoptp^s = 0$), then the weighted \gap\ places no weight on students' \gap. In such an instance, allowing adoption for every donation would always be myopically optimal because the weighted \gap\ would always be negative. Intuitively, presenting students with the opportunity to adopt does not affect the type of matches that they form, while providing professionals with the opportunity to adopt can only improve their engagement and retention. The proof of Proposition \ref{prop:myopichetero} follows quite naturally from the proof of Theorem \ref{thm:myopic}, and the details can be found in Appendix \ref{app:proof:prop:myopichetero}. 

 Proposition \ref{prop:myopichetero} allows us to conduct comparative statics to evaluate how volunteers' adoption probability impacts the number of completed matches $m_1$. In Appendix \ref{app:comp_statics_adoptp}, we illustrate that increasing volunteers' willingness to adopt can impact the optimal policy and significantly improve the number of completed matches. Many platforms allow for commitment whenever possible (e.g., by having an ``Adopt'' button like FRUS), and sometimes exert significant effort to encourage commitment from their volunteers. Our results provide insight into the value of such efforts.

\subsection{Long-Run Performance of Optimal Myopic Policy}
\label{subsec:heterolongrun}

We next consider the platform's long-run optimization problem in the presence of heterogeneous volunteers, which remains the same in spirit as the optimization problem (PD) presented in Section \ref{subsec:opt}. Given the marketplace dynamics, the platform decides how much adoption to allow (i.e., $\{z_t: t \in [T]\}$) {to maximize} the total discounted number of completed donations. 

Unfortunately, characterizing the optimal long-run policy in this generalized setting is intractable.
Instead, we turn to simulation to shed light on the performance of simple policies. We find that a policy of repeatedly making {myopically optimal} decisions (according to the criteria given in Proposition \ref{prop:myopichetero}) performs quite well. We tested this policy across 10,000 simulated instances with $T=10$ and other reasonable randomly-generated parameters. In Appendix \ref{app:simulations}, we describe the parameter space for our simulations and provide a detailed description of simulation results.

We highlight that in 90.7\% of the simulated instances, the repeated myopically optimal policy was exactly optimal, i.e., it maximized the total discounted number of matches. Across all simulations, the repeated {myopically optimal}  policy achieved on average 99.8\% of the performance of the optimal policy, and it always performed at least 83.2\% as well as the optimal policy. 
To complement our numerical observations, we establish the following approximation guarantee in terms of the gap between the maximum possible growth rate ($B_1 := \max\{\gamma'-\gamma, \gamma'-\alpha +\alpha'^p, \gamma'-\alpha +\alpha'^s, \beta'-\beta\}$) and the minimum growth rate for matched volunteers ($\underline{r} := \min\{\gamma'-\gamma, \gamma'-\alpha +\alpha'^p, \gamma'-\alpha +\alpha'^s\}$).
\begin{proposition}
\label{prop:approxrepeatmyopic}
Suppose the discount factor is small enough such that $\delta(1+B_1) < 1$ where $B_1$ is defined above.
Then, a policy of repeatedly playing the myopically optimal policy given by equation \eqref{eq:optmyopichetero} for all $t \in [T]$ is $1-\kappa_m$ optimal, where
			$$\kappa_m =  \min\left\{\frac{\delta (1-\delta)}{\tilde{m}_1(1-\delta^T)}\left((B_1-\underline{r})\frac{2(1+B_1)\min\{d_0,v_0\}}{(1-\delta(1+B_1))^2} + \frac{(\underline{r} + \max\{\alpha,\beta\}) (1+\underline{r})^2 \log(2)}{\underline{r}(1-\delta (1+\underline{r}))c} \right), 1\right\},$$
			
and  $\tilde{m}_1$ denotes the myopically optimal number of matches in period $1$.  
\end{proposition}

Proposition \ref{prop:approxrepeatmyopic} establishes a ratio between the performance of a repeated {myopically optimal} policy and the optimal policy.
The ratio is decreasing in the difference between the growth rates $B_1$ and $\underline{r}$ {due to the diminished long-term implications for growth of a ``bad'' decision}, and it approaches 1 as the discount rate approaches 0. 
The proof of Proposition \ref{prop:approxrepeatmyopic} resembles the proof of Theorem \ref{thm:approxregimeB};
a complete proof can be found in Appendix \ref{app:proof:prop:approxrepeatmyopic}. 
}

\section{Conclusion}\label{sec:discussion}
 The literature on designing two-sided matching markets is rich and varied, and there have been a multitude of papers studying how to optimize centralized matching as well as how to design levers such as pricing, search, and information. To our knowledge, our work is one of the first to suggest that the level of \emph{commitment} to a match is also a design lever that can be influenced by the platform. Such a lever is potentially valuable in a wide range of settings where matchings are repeated over time, and it is particularly useful for nonprofit organizations which often cannot rely on monetary levers to promote platform growth. {Volunteer-based organizations such as FRUS and Food Rescue Hero tend to assume that more commitment will amplify the growth of their platform. However, increasing commitment for some users may reduce the engagement of other users, which can have mixed effects on overall engagement and platform growth.}
 
{Given the compounding importance of platform growth, it is vital that organizations (especially nonprofit organizations) carefully consider the double-edged impacts that commitment can have on growth and engagement. Our work sheds light on this trade-off, and we provide a simple framework for assessing whether commitment will be helpful or potentially harmful.}

\ACKNOWLEDGMENT{The authors gratefully acknowledge the Simons Institute for the Theory of Computing, as this work was done in part while attending the program on Online and Matching-Based Market Design.}

\bibliographystyle{informs2014_orig}

\bibliography{bibliography}
\newpage
\renewcommand{\theHchapter}{A\arabic{chapter}}

\begin{APPENDICES}
\begin{section}{Omitted Details from Section \ref{sec:model}}

\subsection{Additional Details on Empirical Analysis of FRUS Data (Section \ref{subsec:data})}
\label{app:empirics}
\subsubsection{Donor Dropout}
\label{app:emp:dondrop}
\revcolor{
We use a logistic regression to estimate the impact that matching outcomes have on donor dropout. We want to focus our analysis on the (large) subset of donations which are \emph{intended} to recur; however, since we do not directly observe donors intentions, we instead restrict our attention to donations which have occurred for at least two consecutive weeks. Our analysis consists of 22,981 donation-weeks, which comprises 97 weeks and 906 different donations.  For each donation $d$ and week $t$, we observe the binary matching outcome $missed_{d,t}$, which is $1$ if the donation was ``missed'' (i.e. if a donation route and time was posted to the FRUS website and app, but no volunteer signed up) and $0$ if the donation was ``picked up.'' A donation $d$ is said to ``drop out'' after week $t$ if it does not occur again for at least 6 weeks. This binary outcome is our response variable, $dropout_{d,t}$. We estimate the following model, where $\mathbf{X}_{d,t}$ are control variables which allow for location and seasonal fixed effects:
\begin{equation}
    \text{logit}[\text{Pr}(dropout_{d,t})] = \lambda_0 + \lambda_1 \cdot missed_{d,t} + \mathbf{\lambda}_2 \cdot \mathbf{X}_{d,t}.
\end{equation}
In Table \ref{table:donordropouttable}, we report the results of our logistic regression, which estimates the impact of a missed donation on the dropout probability. The estimated effect is statistically significant at the $p < 0.001$ level, and is robust to controls on location and season. 
%In the second column, we incorporate fixed effects for the seven largest locations, and in the third column, we incorporate fixed effects for the first three quarters of the year. In the fourth column, we incorporate both sets of fixed effects. 
It is also robust to changes in our definition of dropout (e.g. {if we assume a donation drops out if it does not occur again for 4 weeks}). We omit further results for the sake of brevity.
\begin{table}[ht]
\centering 
\revcolor{\footnotesize
\begin{tabular}{@{} l *{4}{d{5.6}} @{}} 
\cmidrule[1pt]{1-5}
&\tablemc{(1)}&\tablemc{(2)}&\tablemc{(3)}&\tablemc{(4)} \\[-1ex]
& \tablemc{$missed$} & \tablemc{\qquad \ \ $missed$ + \qquad \qquad } & \tablemc{\qquad \ \ $missed$ + \qquad \qquad }  & \tablemc{$missed$ +}
\\ [-1.5ex]
& \tablemc{only} & \tablemc{Location FEs} & \tablemc{Seasonal FEs} & \tablemc{Loc./Seas. FEs} 
\\ 
\cline{1-5}
 $missed$ & 1.232^{***} & 1.375^{***} & 1.221^{***}&1.373^{***} \\ [-1.5ex]
  & (0.267) & (0.273) & (0.267)&(0.273) \\ [1ex]
 Constant & -5.759^{***} & -5.065^{***} & -5.783^{***}&-5.219^{***} \\  [-1.5ex]
  & (0.122) & (0.219)  & (0.216)& (0.278)\\
\cline{1-5}
   Df Model & \multicolumn{1}{c}{1} & \multicolumn{1}{c}{8} & \multicolumn{1}{c}{4}&\multicolumn{1}{c}{11} \\[-1.5ex]
   $n$ & \multicolumn{1}{c}{22981} & \multicolumn{1}{c}{22981} & \multicolumn{1}{c}{22981}&\multicolumn{1}{c}{22981} \\
\cmidrule[1pt]{1-5}
\multicolumn{3}{@{}l@{}}{\footnotesize Note: $^{***}\, p<0.001$}
\end{tabular} 
}
  \caption{\revcolor{Logistic regression estimating the impact of a missed donation on the likelihood of dropout.}}
  \label{table:donordropouttable}
\end{table}
}
\medskip

\subsubsection{Volunteer Dropout}
\label{app:emp:voldrop}
\revcolor{
Due to the challenges of our censored dataset described in Section \ref{subsec:data}, we use the availability of donations as a link between match outcomes and volunteer dropout. If there are fewer availabilities, we expect volunteers who are not part of an adopted match to be less likely to find a compatible match, and thus be more likely to drop out. In contrast, we do not expect volunteers who are part of an adopted match to be impacted by the number of availabilities, since they are guaranteed a compatible match.

We focus our analysis on four large locations which together comprise a majority of the data. This allows us to track the size of the spot market, which is defined as the number of donations scheduled for within the next week that are currently available for sign-up. We use a logistic regression to estimate the impact of availabilities (i.e., the size of next week's spot market) on volunteer dropout, both for volunteers which are not part of an adopted match and for volunteers which are. For each volunteer $v$ in location $j$ that transports a donation in week $t$, we observe (i) whether they are part of an adopted match in week $t$, (ii) the size of the spot market in their location in week $t+1$ ($spot_{j, t+1}$), and (iii) whether they drop out following week $t$ ($dropout_{v,j,t}$), which is defined as not transporting another donation for the next six weeks. Due to varying marketplace characteristics, we incorporate location-based fixed effects along with seasonal fixed effects into the vector $\mathbf{X}_{v,j,t}$. For volunteers who are not part of an adopted match in week $t$, we estimate the following:
\begin{equation}
\text{logit}[\text{Pr}(dropout_{v,j,t})] = \lambda_0 + \lambda_1 \cdot spot_{j, t+1} + \mathbf{\lambda}_2 \cdot \mathbf{X}_{v,j,t}.
\end{equation}

\begin{table}[ht]
\centering 
\revcolor{\footnotesize
\begin{tabular}{@{} l *{5}{d{5.6}} @{}} 
\cmidrule[1pt]{1-6}
&\multicolumn{2}{c}{Vols. in \temporary\ match} &&\multicolumn{2}{c}{Vols. in adopted match} \\
\cline{2-3} \cline{5-6}
&\tablemc{(1)}&\tablemc{(2)}&&\tablemc{(3)}&\tablemc{(4)} \\[-1ex]
& \tablemc{$spot$ + } & \tablemc{$spot$ + } && \tablemc{$spot$ + }  & \tablemc{$spot$ + }
\\ [-1.5ex]
& \tablemc{Location FEs} & \tablemc{\quad Loc./Seas. FEs\quad} && \tablemc{Location FEs} & \tablemc{\quad Loc./Seas. FEs\quad} 
\\ 
\cline{1-6}
 $spot$ & -0.035^{***} & -0.031^{***} && 0.001&0.009 \\ [-1.5ex]
  & (0.008) & (0.009) && (0.016)&(0.019) \\ [1ex]
 Constant & -1.019^{***} & -1.123^{***} && -3.258^{***}&-3.394^{***} \\  [-1.5ex]
  & (0.134) & (0.176)  && (0.256)&(0.353) \\
\cline{1-6}
   Df Model & \multicolumn{1}{c}{4} & \multicolumn{1}{c}{7} && \multicolumn{1}{c}{4}&\multicolumn{1}{c}{7} \\[-1.5ex]
   $n$ & \multicolumn{1}{c}{3204} & \multicolumn{1}{c}{3204} && \multicolumn{1}{c}{5156}&\multicolumn{1}{c}{5156} \\
\cmidrule[1pt]{1-6}
\multicolumn{3}{@{}l@{}}{\footnotesize Note: $^{***}\, p<0.001$}
\end{tabular} 
}
  \caption{\revcolor{Logistic regression estimating the impact of the number of available donations in the spot market on the likelihood of volunteer dropout.}}
  \label{table:voldropouttable}
\end{table}

As reported in columns (1) and (2) of Table \ref{table:voldropouttable}, we find that an additional donation in the spot market significantly decreases the dropout probability for volunteers who are not in an adopted match. We then repeat this analysis for volunteers in adopted matches, and we find that the impact of an additional available donation is negligible and not statistically significant (as seen in columns (3) and (4)). This supports our hypothesis that available donations only impact the dropout probability of volunteers who may not be able to find a match.
As shown, these results are robust to the inclusion of seasonal fixed effects, and furthermore, the results are robust to changes in our definition of dropout (e.g. {if we assume a volunteer drops out if they does not volunteer again for 4 weeks}) and our temporal definition of availabilities (e.g. {if we consider the size of the spot market over a slightly different seven-day window}). We omit such results for the sake of brevity.
}

\subsubsection{Asymmetric Volunteer Engagement}
\label{app:emp:engagement}
\revcolor{
{We would like to directly assess the engagement of volunteers as a function of their match type. Unfortunately, we do not observe volunteers' visits to the platform unless they sign up for a donation (see Section \ref{subsec:data}), which means that we can only measure a volunteer's engagement by the number of donations they complete.} 
{Despite this limitation, in the following, we provide aggregate as well as location-specific evidence that suggests volunteers in non-adopted matches can be more engaged. We 
also discuss how the censored nature of our data may impact the analysis.}

{If volunteers are highly engaged but cannot find a compatible match, we will not observe their high level of engagement. To sidestep this limitation, we focus on the subset of volunteers who transported one donation $d$ in week $t$ and have the option to sign up for that same donation $d$ in week $t+1$ (for volunteers in an adopted match, this is almost always the case). We then assess the number of donations completed by these volunteers in week $t+1$ as a function of their match type in period $t$. We focus on four large locations which together comprise a majority of the data, and for each location, we use a linear regression where the dependent variable $donations_{v, t+1}$ is the number of donations completed by volunteer $v$ during week $t+1$.} 

{Before presenting the regression analysis, in Figure \ref{fig:histograms} we show the distribution of 
$donations_{v, t+1}$ separated by match type. 
We observe that (i) a vast majority of volunteers in an adopted match only complete one match -- the one they are committed to; (ii) those volunteers also have a higher retention rate; (iii) however, interestingly, volunteers in non-adopted matches are more likely to complete $2+$ donations, compelling evidence that they are more engaged. In fact, due to this extra engagement, on average, those volunteers complete more donations in the subsequent week compared to their counterparts in adopted matches (1.09 vs 0.94).}

{
For this subset, we now
estimate a simple model which depends on an indicator variable $\textit{non-adopt}_{v,t}$ that is $1$ if $v$ is not part of an adopted match in week $t$,} 
\begin{equation}
    donations_{v, t+1} = \lambda_0 + \lambda_1 \cdot \textit{non-adopt}_{v,t}. \label{eq:NAGB}
\end{equation}

{Due to potential differences in volunteer characteristics across markets, we conduct this estimation separately for each of the four locations. The coefficient on $\textit{non-adopt}_{v,t}$ represents the increased engagement provided by volunteers who are not part of an adopted match relative to volunteers who are part of an adopted match. These results are robust to the inclusion of a seasonal fixed effect term, though we omit such results for the sake of brevity.} 
%{We make the following observations based on Table~\ref{table:NAGB}: (i) The coefficient $\lambda_1$ is statistically significant in three out of the four region implying that match type impacts the number of donations a volunteer completed in the subsequent week. (ii) $\lambda_1$ varies substantially across locations suggesting that volunteer characteristics can vary across markets. } 
{Table~\ref{table:NAGB} shows that in two of the four locations, the coefficient on $\textit{non-adopt}_{v,t}$ is positive and statistically significant, which implies that in those locations, volunteers who were in one-time matches in period $t$ are more engaged in period $t+1$ than volunteers who were in adopted matches.}

\begin{figure}[t]
 \centering
 \includegraphics[width=.65\textwidth]{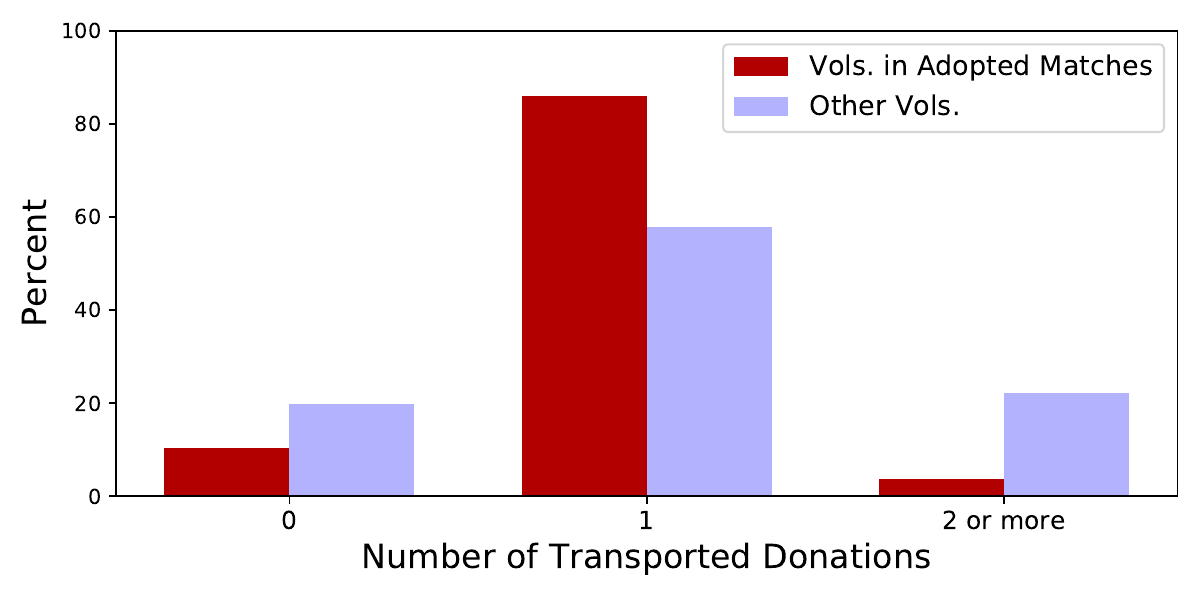}
  \caption{\revcolor{Of the volunteers in Locations (a), (b), (c), and (d) who transported one donation in week $t$ and have the same donation available to them in week $t+1$, we show the percentage who sign up for 0, 1, or 2+ donations in week $t+1$, separated by their match type in week $t$.}}\label{fig:histograms}
\end{figure}

\begin{table}[t]
\centering 
\revcolor{\footnotesize
\begin{tabular}{@{} l *{4}{d{5.6}} @{}} 
\cmidrule[1pt]{1-5}
&\multicolumn{1}{c}{Location (a)} &\multicolumn{1}{c}{Location (b)}&\multicolumn{1}{c}{Location (c)}&\multicolumn{1}{c}{Location (d)} \\
\cline{1-5}
 $\textit{non-adopt}$ & -0.292^{**} & 0.249^{***} & -0.034&0.489^{***} \\ [-1.5ex]
  & (0.110) & (0.061) & (0.054)&(0.095) \\ [1ex]
 Constant & 1.019^{***} & 0.899^{***} & 0.948^{***}&0.900^{***} \\  [-1.5ex]
  & (0.013) & (0.014)  & (0.011)&(0.030) \\
\cline{1-5}
   Df Model & \multicolumn{1}{c}{1} & \multicolumn{1}{c}{1} & \multicolumn{1}{c}{1}&\multicolumn{1}{c}{1} \\[-1.5ex]
   $n$ & \multicolumn{1}{c}{753} & \multicolumn{1}{c}{1254} & \multicolumn{1}{c}{1468}&\multicolumn{1}{c}{357} \\
\cmidrule[1pt]{1-5}
\multicolumn{3}{@{}l@{}}{\footnotesize Note: $^{**}\, p < 0.01$; \quad $^{***}\, p<0.001$}
\end{tabular} 
}
  \caption{\revcolor{Linear regression estimating the impact of match type on the number of completed donations.}}
  \label{table:NAGB}
\end{table}

{We claim that the estimates in Table \ref{table:NAGB} represent a lower bound on the difference in engagement between volunteers in one-time matches and volunteers in adopted matches (for each location). To see why, note that the number of completed donations, $donations_{v, t+1}$, can only serve as a lower bound on volunteer $v$'s true engagement level in period $t+1$, $engagement_{v, t+1}$. There is the potential for a gap between these two quantities when the the engagement level of the volunteer is greater than $1$: if the volunteer is willing to complete multiple donations, they may be unable to find enough compatible donations in the spot market. Volunteers in one-time matches are, on average, more likely to complete multiple donations (this is true in each location and shown in aggregate in Figure \ref{fig:histograms}). As a result, if we were able to directly observe $engagement_{v,t}$ and use it as the dependent variable in eq. \eqref{eq:NAGB} instead of $donations_{v,t}$, then we would expect our estimate for the coefficient on $\textit{non-adopt}_{v,t}$ to be larger. In other words, because we cannot always observe the true engagement level of volunteers when their engagement is large, we are disproportionately underestimating the engagement level of volunteers in one-time matches.}
}

\revcolor{
To more accurately estimate the engagement difference between volunteers in one-time matches and volunteers in adopted matches, we would need more granular volunteer engagement data. As one direction for future research, we intend to encourage the platform to keep track of this data, which would allow us to estimate the impact of match type in different locations. (We note that our estimation procedure corresponds to estimating the \gap\ in our base model. If there are heterogeneous volunteer types as described in Section \ref{sec:hetero}, then we would want to estimate the \gap\ for each volunteer type separately.)
}

\revcolor{
 To shed more light on the effects of asymmetric volunteer engagement, we would like to be able to directly attribute higher rates of missed donations to higher adoption rates. Unfortunately, there is a notable lack of temporal variation in the adoption level within locations. Since there is insufficient variation in the adoption level \emph{within} locations, we instead look for differences \emph{across} locations in terms of adoption rates and missed donations. We focus our analysis on one full year for the seven locations with at least 1000 total donations in our dataset.
 
For these locations, we find a correlation of $\rho = 0.604$ between the average adoption level (i.e., the fraction of matches which were adopted) in a location and the fraction of donations which were missed in that location in that year. %Each region is given a different color, and the green and red circles correspond to Locations (a) and (b), respectively, which are described in the introduction and Figure \ref{fig:FRUSlocations}. We find a correlation of $\rho = $ for the data shown (and a correlation of $\rho = $ when considering the overall average rates for each region).
  We expect that the relationship between the adoption fraction and missed donations is caused by a lack of engagement in the spot market. To lend credibility to this hypothesis, we find a correlation of $\rho = 0.824$ between the adoption level in a location and the fraction of donations \emph{in the spot market} which were missed in that location. {We note that} there are limitations to this analysis {since} we only have sufficient data on a handful of locations, and there {can be} other differences across locations which could confound the relationship. 
}

\subsection{Proof of Proposition \ref{prop:convergence}}
	\label{proof:prop:convergence}

To aid in this proof, we will use constants $\Delta_t$ and $\Phi_t$ in bounding the deviation of the state variables from their expectation. We will also use of the following lemma, which we prove in Appendix~\ref{proof:prop:matching}.
\begin{lemma}\label{lem:Delta}
{Let $S_t^n = M_t^n - K_t^n$. For each $t \in [T]$, if there exists some $C$ such that $D_t^n \leq Cn$ and $V_t^n \leq Cn$,} then there exists a constant $\Delta_t$ such that with probability $1-O(n^{1/4}e^{-c^6n^{1/4}})$,
\begin{equation*}
\left|\frac{S_{t}^n}{n}-\s\left(\frac{D_{t}^n-K_{t}^n}{n}, \frac{V_{t}^n-K_{t}^n}{n}\right)\right| \leq \Delta_t n^{-1/4}.
\end{equation*}
\end{lemma}

Given $\Delta_t$, we define $\Phi_{-1} = 0$ and  $\Phi_t := \sum_{\tau = 0}^t 6^{t-\tau}\Delta_\tau$ for all $t\in[T]$. We now prove by induction that for all $t \in [T]$, with probability $1-O(n^{1/4}e^{-c^6n^{1/4}})$,
{\small
\begin{align}
    \left|\frac{K_{t}^n}{n}-k_t\right| \leq \Phi_{t-1} n^{-1/4}, \quad \,
    \left|\frac{D_t^n}{n}-d_t\right| \leq 2\Phi_{t-1} n^{-1/4}, \quad \,
    \left|\frac{V_t^n}{n}-v_t\right| \leq 3\Phi_{t-1} n^{-1/4}, \quad \, 
    \left|\frac{M_t^n}{n}-m_t\right| \leq \Phi_{t} n^{-1/4}. \label{proof:eq:mconverge}
\end{align}
}

As a base case, when $t=0$ the first three inequalities hold by definition and the final inequality follows {from an application of Lemma \ref{lem:Delta}, using the facts that $\frac{M_0^n}{n}-m_0 = \frac{S_0^n + K^n_0}{n}-(s_0+k_0) = \frac{S_0^n}{n}-s_0$ and $\Phi_0=\Delta_0$.} 
We now assume that all four inequalities hold for $t = \tau$. We will prove that they also hold for $t = \tau+1$. Now
 \begin{align*}
 \left|\frac{K_{\tau+1}^n}{n}-k_{\tau+1}\right| 
 &\leq \left|\frac{K_{\tau+1}^n}{n}-\E\left[\frac{K_{\tau+1}^n}{n}\mid M_\tau^n\right]\right|
 + \left|\E\left[\frac{K_{\tau+1}^n}{n}\mid M_\tau^n\right] - k_{\tau+1}\right| \\
 %&= \left|\frac{K_{\tau+1}^n}{n}-(1-\gamma)z_\tau \frac{M_\tau^n}{n}\right| 
 %+ \left|(1-\gamma)z_\tau \frac{M_\tau^n}{n} - (1-\gamma)z_\tau m_\tau\right| \\ 
 &= \frac{1}{n}\left|K_{\tau+1}^n-(1-\gamma)z_\tau M_\tau^n\right| 
 + (1-\gamma)z_\tau\left|\frac{M_\tau^n}{n} - m_\tau\right|.
 \end{align*}
Recall that $K_{\tau+1}^n$ is a binomial random variable. Thus, using a Chernoff bound, with probability $1-2e^{n^{-1/2}}$ the first term is at most $ n^{-5/4} \sqrt{3(1-\gamma)z_\tau M_\tau^n}$. By the inductive hypothesis, the second term is at most $(1-\gamma)z_\tau \Phi_{\tau} n^{-1/4}$ with probability $1-O(n^{1/4}e^{-c^6n^{1/4}})$, and we note that this implies $M_\tau = O(n)$. We now take a union bound over those two upper bounds. With probability $1 - O(e^{n^{-1/2}}) -O(n^{1/4}e^{-c^6n^{1/4}})$, the distance between $\frac{K_{\tau+1}^n}{n}$ and $k_{\tau+1}$ is at most the sum of our two upper bounds. For large enough $n$, the second upper bound dwarfs the first, implying that with probability $1-O\left(n^{1/4}e^{-c^6n^{1/4}}\right)$,
$\left|\frac{K_{\tau+1}^n}{n}-k_{\tau+1}\right| \leq \Phi_{\tau} n^{-1/4}$.
This proves by induction that the first inequality in Line \eqref{proof:eq:mconverge} holds for all $t \in [T]$. 

For the sake of brevity, we omit the nearly identical proofs for the second and third inequalities, which use Chernoff bounds two and three times, respectively, due to the additional binomial processes. For the final inequality, 
%we use techniques from \citet{wormald1999models} to prove the convergence of the matching function conditional on the current state, and then 
we leverage the sublinearity of the matching function to bound the distance between the matching resulting from the actual state and the matching resulting from the deterministic approximation of the state.
{\footnotesize
\begin{align}
    \left|\frac{M_{\tau+1}^n}{n}-m_{\tau+1}\right| \leq& \left|\frac{M_{\tau+1}^n}{n}-\frac{K_{\tau+1}^n}{n}-\s\left(\frac{D_{\tau+1}^n-K_{\tau+1}^n}{n}, \frac{V_{\tau+1}^n-K_{\tau+1}^n}{n}\right)\right| \nonumber \\&+
    \left|\frac{K_{\tau+1}^n}{n}+\s\left(\frac{D_{\tau+1}^n-K_{\tau+1}^n}{n}, \frac{V_{\tau+1}^n-K_{\tau+1}^n}{n}\right) - m_{\tau+1}\right| \nonumber \\
    \leq& 
    \left|\frac{S_{\tau+1}^n}{n}-\s\left(\frac{D_{\tau+1}^n-K_{\tau+1}^n}{n}, \frac{V_{\tau+1}^n-K_{\tau+1}^n}{n}\right)\right|
    +\left|\frac{D_{\tau+1}^n}{n} - d_{\tau+1}\right| +\left|\frac{K_{\tau+1}^n}{n} - k_{\tau+1}\right| +\left|\frac{V_{\tau+1}^n}{n} - v_{\tau+1}\right| \label{eq:threeterms}
\end{align}
}
The first term of \eqref{eq:threeterms} comes from cancelling out the certain adopted matches. The next three terms come from the fact that $\frac{ \partial m_t}{\partial d_t}$, $\frac{ \partial m_t}{\partial k_t}$ and $\frac{ \partial m_t}{\partial v_t}$ are all in [0,1] (see Appendix \ref{proof:lemma:myopicisbest}).  We have concentration results for all three of these terms, and we thus take a union bound over those concentration results. With probability $1-O(n^{1/4}e^{-c^6n^{1/4}})$, the sum of those three terms is upper-bounded by $6\Phi_{\tau}n^{-1/4}$, {which also implies that $D_{\tau+1}^n \leq Cn$ and $V_{\tau +1}^n \leq Cn$ for some large $C$. This allows us to apply Lemma \ref{lem:Delta} to the top line: with probability $1-O\left(n^{1/4}e^{-c^6n^{1/4}}\right)$, it is at most $\Delta_{\tau+1}n^{-1/4}$.}
Combining these two bounds via a union bound, we have with probability $1-O(n^{1/4}e^{-c^6n^{1/4}})$,
{\small
\begin{equation}
    \left|\frac{M_{\tau+1}^n}{n}-m_{\tau+1}\right| \leq
    6\Phi_{\tau}n^{-1/4} + \Delta_{\tau+1}n^{-1/4} = \Phi_{\tau+1}n^{-1/4}. \nonumber
\end{equation}
}
This completes the proof by induction (we note that the number of union bounds taken is $O(T) = O(1)$). To show that our convergence results hold for all $t \in [T]$, we take a union bound over the $4T$ state variables. Again noting that $T = O(1)$, we have with probability $1-O(n^{1/4}e^{-c^6n^{1/4}})$, for all $t \in [T]$, 
{\small 
\begin{align}
    \left|\frac{K_{t}^n}{n}-k_t\right| \leq O\left(n^{-1/4}\right), \qquad     \left|\frac{D_t^n}{n}-d_t\right| \leq O\left(n^{-1/4}\right), \qquad 
    \left|\frac{V_t^n}{n}-v_t\right| \leq O\left(n^{-1/4}\right), \qquad 
    \left|\frac{M_t^n}{n}-m_t\right| \leq O\left(n^{-1/4}\right). \nonumber
\end{align}
}
This completes the proof of almost sure convergence.

\if false
We will prove this proposition in two parts. First, we will show that if $|\frac{D_t^n}{n}-d_t|$, $|\frac{K_t^n}{n}-k_t|$, and $|\frac{V_t^n}{n}-v_t|$ are all less than $O(t n^{-1/4})$ with probability $blank$, then $|\frac{M_t^n}{n} - m_t| \leq O(n^{-1/4})$ with probability $blank$.

Then we will prove via induction that for a series of constants $\{\Phi_t : t \in [T]\}$, $|\frac{D_t^n}{n}-d_t| \leq \Phi_t t n^{-1/4}$, $|\frac{K_t^n}{n}-k_t| \leq \Phi_t t n^{-1/4}$, $|\frac{V_t^n}{n}-v_t| \leq \Phi_t t n^{-1/4}$, and $|\frac{M_t^n}{n} - m_t| \leq \Phi_t (t+1) n^{-1/4}$ with probability $blank$. 

Applying these steps inductively starting from $t=0$ to $t = T$, we then take a union bound to show that with probability $blank$, for all $t \in [T]$
\begin{align}
|\frac{D_t^n}{n}-d_t| \leq O(n^{-1/4}) \\
|\frac{K_t^n}{n}-k_t| \leq O(n^{-1/4}) \\
|\frac{V_t^n}{n}-v_t| \leq O(n^{-1/4}) \\
|\frac{M_t^n}{n}-m_t| \leq O(n^{-1/4})
\end{align}
which completes the proof. Note that we require $T$ as well as $d_t, k_t,$ and $v_t$ to be of order $O(1)$  with respect to $n$.

We begin by showing the convergence of the matching process at any given period $\tau$. 
Using the proof of Proposition \ref{prop:matching}, we know that conditional on $D^n_\tau, K_\tau^n,$ and $V^n_\tau$, then with probability $1-O(n^{1/4}e^{-c^6n^{1/4}})$,
\begin{equation}
|\frac{M^n_\tau}{n} - \frac{K_\tau^n}{n} - \s(1/n(D^n_\tau - K_\tau^n), 1/n (V^n_\tau - K_\tau^n))| = O(n^{-1/4})
\end{equation}
Since $|\frac{\partial \s(a,b)}{\partial a}| \leq 1$ (and equivalently for $b$), with probability $1-O(n^{1/4}e^{-c^6n^{1/4}})$
\begin{equation}
|\frac{M^n_\tau}{n} - k_\tau - \s(d_\tau - k_\tau, v_\tau - k_\tau)| \leq |d_\tau - \frac{D^n_\tau}{n}| + 3|k_\tau - \frac{K^n_\tau}{n}|+ |v_\tau - \frac{V^n_\tau}{n}|+ O(n^{-1/4})
\end{equation}
By assumption, with probability $blank$, $|d_\tau - \frac{D^n_\tau}{n}| \leq \Phi_\tau \tau n^{-1/4}$ (similarly for $k_\tau$ and $v_\tau$). Thus, using a union bound, we have with probability $1-O(n^{1/4}e^{-c^6n^{1/4}})$,
\begin{equation}
|\frac{M^n_\tau}{n} - m_\tau| \leq (5 \Phi_\tau + 1) 
\end{equation}
This concludes the first part of the proof.

We now show that for a series of constants $\Phi_t$, $|\frac{D_t^n}{n}-d_t| \leq \Phi_t t n^{-1/4}$, $|\frac{K_t^n}{n}-k_t| \leq \Phi_t t n^{-1/4}$, $|\frac{V_t^n}{n}-v_t| \leq \Phi_t t n^{-1/4}$, and $|\frac{M_t^n}{n} - m_t| \leq \Phi_t (t+1) n^{-1/4}$ with probability $blank$. We proceed by induction. Clearly, this is true for $t=0$, since the first three results hold by equality, and the last was proven above.

By the dynamics described in \eqref{eq:dt}, $$D^n_{t+1} = D_t^n + \text{Binomial}(M_t^n, \beta') - \text{Binomial}(D_t^n - M_t^n, \beta)$$
Using a Chernoff bound, we have that with probability $1-O(e^{\frac{-n^(1/2)\beta'}{2}})$,
$$\frac{1}{n}|\text{Binomial}(M_t^n, \beta')-\beta' M_t^n| \leq n^{-1/4} \beta' (\frac{M_t^n}{n})^{1/2}$$ Similarly, with probability $1-O(e^{\frac{-n^{1/2}(1-\beta)}{2}})$,
$$\frac{1}{n}|\text{Binomial}(D_t^n-M_t^n, 1-\beta)-(1-\beta)(D_t^n M_t^n)| \leq n^{-1/4} (1-\beta)(\frac{D_t^n}{n}-\frac{M_t^n}{n})^{1/2}$$
Combining the two using a union bound, for some constant $\omega_1$, we have that with probability $1-O(e^{\frac{-n^(1/2)}{\omega_1}})$,
\begin{align}
\frac{1}{n}|D^n_{t+1} - (1-\beta)D_t^n - (\beta'-\beta)M_t^n  | &\leq n^{-1/4} ((\frac{M_t^n}{n})^{1/2} + (\frac{D_t^n}{n}-\frac{M_t^n}{n})^{1/2}) \\
&\leq n^{-1/4} ((m_t + t\Phi_t n^{-1/4} )^{1/2} + (d_t - m_t + 2t\Phi_t n^{-1/4})^{1/2}) \\
&\leq \Phi_{t+1} n^{-1/4}
\end{align}
By the triangle inequality, this implies that with probability $blank$
\begin{align}
|\frac{D^n_{t+1}}{n} - (1-\beta)d_n^t - (\beta'-\beta)d_t^n  | &\leq O(n^{-1/4}) + (1-\beta)|\frac{D_t^n}{n}-d_t| +  (\beta'-\beta)|\frac{M_t^n}{n}-m_t| \\
|\frac{D^n_{t+1}}{n} - d_{t+1}| &\leq \Phi_{t+1} (t+1) n^{-1/4}
\end{align}
Identical arguments using the dynamics in \eqref{eq:kt} and \eqref{eq:vt} show that with probability $blank$
\begin{align}
|\frac{K^n_{t+1}}{n} - k_{t+1}| &\leq \Phi_{t+1} (t+1) n^{-1/4} \\
|\frac{V^n_{t+1}}{n} - v_{t+1}| &\leq \Phi_{t+1} (t+1) n^{-1/4}
\end{align}
Given these three results, we can the first step of this proof implies that with probability $blank$, $|\frac{M^n_{t+1}}{n} - m_{t+1}| \leq O(n^{-1/4})$.
By the union bound, all three concentration results must hold with probability $blank$. As previously noted, repeating these steps from $t=0$ to $t = T$ and taking a union bound yields the result in Proposition \ref{prop:convergence}.
\fi
\end{section}

\begin{section}{Omitted Details from Section \ref{sec:results}}

	\begin{subsection}{Proof of Proposition \ref{prop:matching}}
		\label{proof:prop:matching}
		Based on the greedy random matching process, if there are $na$ available \jobs \ when a \worker \  arrives, the volunteer will form a match with probability $1-(1-\frac{c}{n})^{na}$. As a consequence, the expected number of available \jobs \ when the next \worker \ arrives is given by $na -1+(1-\frac{c}{n})^{na}$. This enables us to write a stochastic differential equation for the number of matched \jobs. We then show that the size of the matching converges to the solution of a deterministic differential equation as the market size grows, holding the ex ante number of compatible matches for donations and volunteers (i.e., the market thickness) fixed.

		Suppose that there are $na$ donations and $nb$ volunteers in the spot market, where the match probability is given by $\frac{c}{n}$. %We remark the the ex ante number of compatible matches for donations (resp. volunteers) is $cb$ (resp. $ca$). 
		Let $Y(Z) \in [na]$ be the number of matched \jobs \ right after  the $Z^{th}$ \worker \ arrives, where $Z \in [nb]$.  At the beginning of the process, there are no matches ($Y(0) = 0$), and according to the greedy matching process described above, the expected number of matches evolves as follows: 
		\begin{equation*} 
		\E[Y(Z+1)|Y(Z)] = Y(Z) + 1 - \left(1-\frac{c}{n}\right)^{n(a - \frac{Y(Z)}{n})}
		\end{equation*}
		If we define $y(z) = \frac{Y(nz)}{n}$, then for large $n$,
		$\frac{1}{n}\E[y(z + \frac{1}{n}) - y(z)|y(z)] = 1 - e^{-c(a-y(z))} + o(1).$
		{A}s $n \rightarrow \infty$, this corresponds to the differential equation $\frac{d y}{d z} = 1 - e^{-c(a-y(z))}$.  Given the initial condition $y(0) = 0$, this differential equation has the unique solution
		$y(z) = a + z - \frac{1}{c}\log(e^{ca} + e^{cz} - 1)$.

		To show convergence (which will complete the proof of Proposition \ref{prop:matching}), we apply Theorem 5.1 from \citet{wormald1999models}, and use the fact that the expected change in the number of matches can be characterized by
		\begin{equation*}
		\E\left[Y(Z+1) - Y(Z)\mid Y(Z)\right] = 1 - \left(1-\frac{c}{n}\right)^{n\left(a - \frac{Y(Z)}{n}\right)}.
		\end{equation*}

		{	\subsubsection*{Theorem 5.1 \citep{wormald1999models}}
		\textit{Let $Q^{(n)+}$ represent the set of all possible sequences of matching outcomes for the $nb$ volunteers, and let $g:Q^{(n)+} \rightarrow \mathbb{R}$ and $f:\mathbb{R}^2 \rightarrow \mathbb{R}$ such that $|g(h_z)| \leq nb$ for all $h_z \in Q^{(n)+}$ for all $n$. Let $Y(Z)$ denote the random counterpart of $g(h_z)$, and assume the following three conditions hold:}
		
\noindent \textit{(i) (Boundedness Hypothesis) For some $\beta \geq 1$, $|Y(Z+1)-Y(Z)| \leq \beta$.}

	\noindent \textit{(ii) (Trend Hypothesis) For some $\lambda_1 = o(1)$, $|E[Y(Z+1)-Y(Z)|Y(Z)] - f(\frac{Z}{n}, \frac{Y(Z)}{n})| \leq \lambda_1$ for all $Z \leq bn$.} 
		
		\noindent \textit{(iii) (Lipschitz Hypothesis) The function $f$ is continuous and satisfies a Lipschitz condition for a domain $D := \{(Z,Y(Z)) : Z \in (-\delta, bn+\delta), Y(Z) \in (-\delta, an + \delta) \}$.}
\smallskip
	
	\noindent	\textit{Then the following are true:}
		
		\noindent \textit{(a) The differential equation $\frac{dy}{dz} = f(z, y(z))$ has a unique solution passing through $y(0) = 0$ which extends arbitrarily close to the boundary of $D$.}
		
	\noindent	\textit{(b) Let $\lambda > \lambda_1$ with $\lambda = o(1)$. For a sufficiently large constant $C$, with probability $1-O(\frac{\beta}{\lambda}\text{exp}(-\frac{n\lambda^3}{\beta^3}))$, $\frac{Y(Z)}{n} = y(z) + O(\lambda)$ for all $Z \in [0, \sigma n]$ where $y(z)$ is the solution described in (a) and $\sigma$ is the supremum of those $z$ to which the solution can be extended before reaching within $l^\infty$-distance $C\lambda$ of the boundary of $D$.}}
		
		%To apply this, we must show that (i) the boundedness hypothesis holds, (ii) the trend hypothesis holds, and (iii) the Lipschitz hypothesis holds. 
		
		We first note that for all $Z$, $\E[Y(Z+1) - Y(Z)] \leq 1$, since no more than one match can occur for each volunteer who visits the platform. This satisfies (i) the boundedness hypothesis. For (ii) the trend hypothesis, we define $f(\frac{Z}{n}, \frac{Y(Z)}{n}) = e^{(-c)\left(a - \frac{Y(Z)}{n}\right)}$, and we claim that $1 - \left(1-\frac{c}{n}\right)^{n\left(a - \frac{Y(Z)}{n}\right)}$ is always within $\lambda_1 := \frac{c^2}{n}$ of $f(\frac{Z}{n}, \frac{Y(Z)}{n})$. To see this, note that  for any $x \geq 0$, $e^{-x-x^2} \leq 1-x \leq e^{-x}$. Setting $x = \frac{c}{n}$ and raising each term to the $n\left(a - \frac{Y(Z)}{n}\right)$ power yields  
		\begin{align*}
		e^{\left(-c-\frac{c^2}{n}\right)\left(a - \frac{Y(Z)}{n}\right)} &\leq \left(1-\frac{c}{n}\right)^{n\left(a - \frac{Y(Z)}{n}\right)} \leq e^{-c\left(a - \frac{Y(Z)}{n}\right)} \\
		e^{(-c)(a - \frac{Y(Z)}{n})}\left(1-\frac{c^2(a - \frac{Y(Z)}{n})}{n}\right) &\leq \left(1-\frac{c}{n}\right)^{n\left(a - \frac{Y(Z)}{n}\right)} \leq e^{-c\left(a - \frac{Y(Z)}{n}\right)}\\
		e^{(-c)\left(a - \frac{Y(Z)}{n}\right)}(1-\frac{c^2}{n}) &\leq \left(1-\frac{c}{n}\right)^{n\left(a - \frac{Y(Z)}{n}\right)} \leq e^{-c\left(a - \frac{Y(Z)}{n}\right)}
		\end{align*}
		Noting that $\frac{c^2}{n}e^{(-c)\left(a - \frac{Y(Z)}{n}\right)} \leq \frac{c^2}{n}$ for $Y(Z) \in [0, an]$ validates our claim that $1 - \left(1-\frac{c}{n}\right)^{n\left(a - \frac{Y(Z)}{n}\right)}$ is always within $\lambda_1$ of $f(\frac{Z}{n}, \frac{Y(Z)}{n})$.

		Finally, we note that the constant $L=c^2e^{c\delta}$ satisfies (iii) the Lipschitz hypothesis for the function $1-e^{-c(a-y(z))}$ over domain $D$. Given that the three conditions of Theorem 5.1 in \citet{wormald1999models} are met, we define $\lambda = c^2n^{-1/4} > \lambda_1$ and apply the theorem to yield that for all $z \in [0,b]$, with probability $1-O(n^{1/4}e^{-c^6n^{1/4}})$,
		$\frac{Y(nz)}{n} = y(z) + O(n^{-1/4})$.
		Plugging in $z = b$ completes that proof that the expected matching converges almost surely to $a+b - \frac{1}{c}\log(e^{ca}+e^{cb}-1)$.

	\subsubsection*{Proof of Lemma~\ref{lem:Delta}}
	{As a consequence of the proof of Proposition \ref{prop:matching}, we know that for any constants $a$ and $b$, there exists a constant $\Delta(a,b)$ such that with probability $1-O(n^{1/4}e^{-c^6n^{1/4}})$, $\left|\frac{S^n}{n}-\s(a, b)\right| \leq \Delta(a, b) n^{-1/4}$. If $D_t^n \leq Cn$ and $V_t^n \leq Cn$, then $\frac{D_t^n-K_t^n}{n}\leq C$ and $\frac{V_t^n-K_t^n}{n} \leq C$. In that case, if we define $\Delta_t = \max_{a, b \in [0, C]} \Delta(a, b)$, then with probability $1-O(n^{1/4}e^{-c^6n^{1/4}})$, $\left|\frac{S^n_t}{n}-\s\left(\frac{D_t^n-K_t^n}{n}, \frac{V_t^n-K_t^n}{n}\right)\right| \leq \Delta_t n^{-1/4}$. This completes the proof.}
	
	\if false
	set $a = \frac{D^n_t-K^t_n}{n} = O(1)$ and $b = \frac{D^n_t-K^t_n}{n} = O(1)$.   
[IL NOTE: Proof not quite finished; I think we have to be a little bit more careful in defining the constants $\Delta_t$. @Scott let's sync?]

		As in the proof of Proposition~\ref{proof:prop:matching} , we apply Theorem 5.1 in \citet{wormald1999models} to show that that for fixed $a,b$ as $n\rightarrow\infty$ we have that $\frac{Y(nb)}{n}=y(b)+O\left(n^{-1/4}\right)$. This implies that there exists a constant $\Delta_t(a,b)$ such that $\left|\frac{S^n_t(a,b)}{n}-\s(a,b)\right|\leq \Delta_t(a,b)n^{-1/4}$. 
\begin{itemize}
		\item Want: take $\Delta_t = \max_n \Delta\left(\frac{D^n_t-K^n_t}{n},\frac{V^n_t-K^n_t}{n}\right)$. 
		\item Problem: max may not exist. 
		\item Shouldn't be an issue because: concentration, so max over compact set.
		\item Problem: we are using this to prove concentration.
		\item It's okay, because we can just make sure to prove concentration of K, D, V before M (since induction is on t)
\end{itemize}
\fi
	\end{subsection}
	
	\begin{subsection}{Proof of Theorem \ref{thm:myopic}}
		\label{proof:thm:myopic}
		We begin by showing that $m_1(z_0)$ does not attain a local maximum. From \eqref{eq:dt}-\eqref{eq:vt}, we have $\frac{\partial d_{1}}{\partial z_0} = 0$, $\frac{\partial k_{1}}{\partial z_0} = (1-\gamma)m_0{= D_k}$, and $\frac{\partial v_1}{\partial z_0} = (\alpha - \alpha' - \gamma)m_0 {= D_v}$.
		Combining these derivatives with \eqref{eq:mt} and \eqref{eq:matchingfunc}, we have 
		\begin{align}
		m_1 \ \ &= d_1 - k_1 + v_1  - \frac{1}{c} \log(e^{c(d_1-k_1)}+e^{c(v_1 - k_1)}-1) \nonumber  \\
		\frac{\partial m_1}{\partial z_0} \ &= \frac{{D_v}e^{c(d_1-k_1)} + {D_k-D_v}}{e^{c(d_1-k_1)}+e^{c(v_1 - k_1)}-1} \label{eq:dmdz}\\
		\frac{\partial^2 m_1}{\partial z_0^2} &= \left( \frac{{D_k^2}e^{c(d_1-k_1)} + {(D_k-D_v)^2}e^{c(v_1-k_1)}-{D_v^2}e^{c(d_1+v_1-2k_1)}}{(e^{c(d_1-k_1)}+e^{c(v_1 - k_1)}-1)^2}\right)c \label{eq:d2mdz2}
		\end{align}
		The first-order condition prescribed by \eqref{eq:dmdz} is equivalent to $e^{c(d_1-k_1)} = {\frac{D_v-D_k}{D_v}}= -\frac{1 +\alpha' - \alpha}{\alpha - \alpha'-\gamma}$. When the FOC holds, the numerator of \eqref{eq:d2mdz2} reduces to ${D_k^2}e^{c(d_1-k_1)} + {D_k(D_k-D_v)}e^{c(v_1-k_1)}$. {This must be positive, since $D_k > D_v$ and $D_k > 0$. Thus, $m_1$ is convex at any critical point and consequently} can have no local maxima as a function of $z_0$. This implies that the optimal solution must be at a boundary, i.e., $z_0 = 0$ or $z_0 =1$.

		To compare the values $m_1(0)$ and $m_1(1)$, we use a slight abuse of notation to augment $(d_1, k_1, v_1)$ by $z_t$ to compare them for $z_t = 0$ and $z_t = 1$. Since $d_1(0) = d_1(1) := d_1$ and $k_1(0) = 0$, the following are necessary and sufficient conditions for $m_1(1) \geq m_1(0)$:
		\begin{align}
0 &\geq e^{c(v_1(0) - v_1(1) + k_1(1))} - \frac{e^{cv_1(0)}+e^{cd_1} -1}{e^{c(v_1(1)-k_1(1))}+e^{c(d_1-k_1(1))} -1}  \label{eq:proof:myopic0} \\
\Leftrightarrow ~~ 0&\geq e^{c({D_k-D_v})}(e^{cd_1}e^{-c{D_k}} -1) - e^{cd_1} +1  \label{eq:proof:myopic1} \\
\Leftrightarrow ~~ 0&\geq e^{cd_1}(e^{{-cD_v}} - 1) - (e^{c({D_k-D_v})}-1) \label{eq:proof:myopic2}
		\end{align}
		Line \eqref{eq:proof:myopic0} comes from applying the definitions in \eqref{eq:mt}-\eqref{eq:matchingfunc},
		Line \eqref{eq:proof:myopic1} is algebraic, using the equalities $v_1(0) - v_1(1) + k_1(1) = {D_k-D_v}$ and $k_1(1) = {D_k}$, and Line \eqref{eq:proof:myopic2} follows from rearranging. It is satisfied when either of the two conditions in \eqref{eq:optmyopic} is met, which completes the proof.
	\end{subsection}
	
\begin{subsection}{Accuracy of Deterministic Approximations}
\label{app:numerics}
Our insights in Section \ref{sec:results} are based on deterministic approximations to the underlying multi-dimensional Markov process. Propositions \ref{prop:convergence} and \ref{prop:matching} imply that our deterministic approximations are accurate when the system is large.  In this section, we numerically demonstrate the accuracy of our approximations for moderate system sizes, and we discuss how some of our main structural results continue to apply.

Our numerical analysis shows that the matching function $\s(a,b)$ defined in \eqref{eq:matchingfunc} is always an underestimate of the actual expected matching, but becomes increasingly accurate as the scaling factor $n$ gets large. {This is evident in the left panel of Figure \ref{fig:numerics}, in which we 
show the normalized gap between the expected size of the matching in a market with two sides of equal size $n$ and its corresponding deterministic approximation ($\mu(1,1)$, as defined in \eqref{eq:matchingfunc}).}
We observe that the normalized gap is always positive, and is below $2\%$ when the size of both sides is at least $30$. 

\begin{figure}
 \centering
 \begin{subfigure}[b]{0.4\textwidth}
 \includegraphics[width=\textwidth]{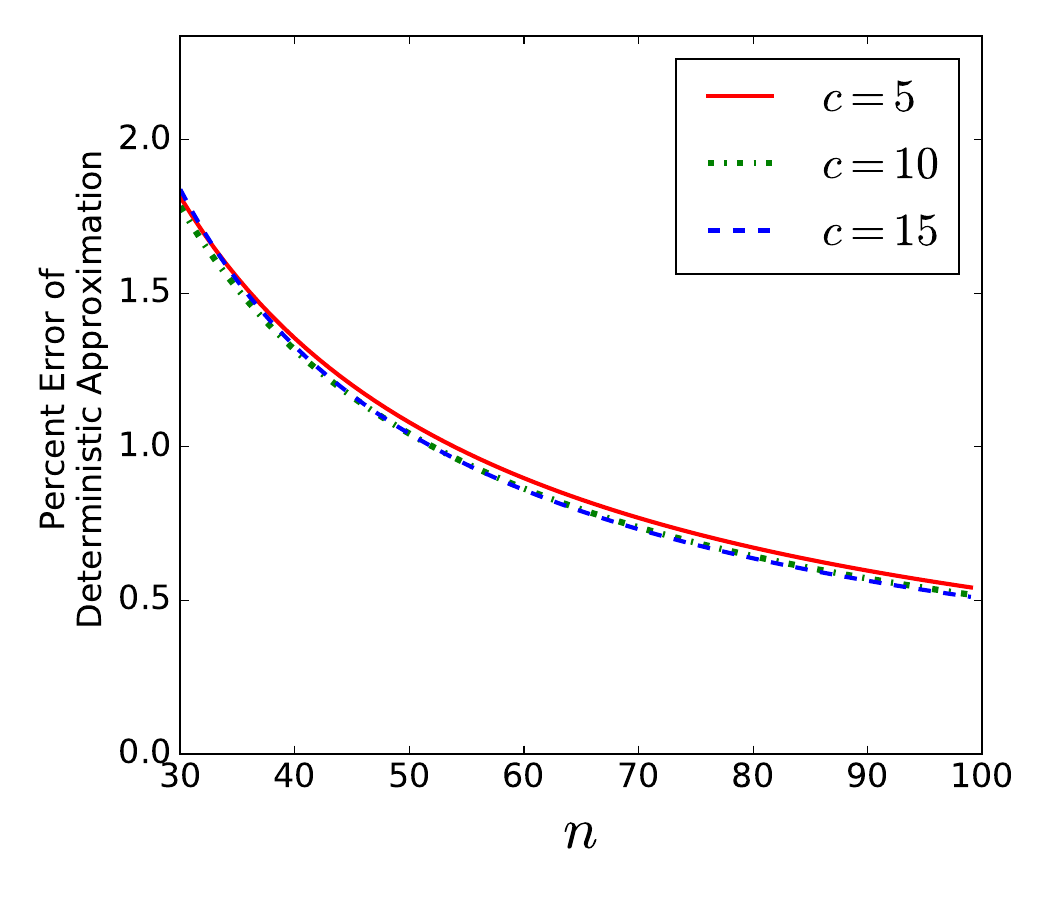}
  %\caption{ }\label{fig:numerics1}
 \end{subfigure}
 \begin{subfigure}[b]{0.4\textwidth}
 \centering
 \includegraphics[width= \textwidth]{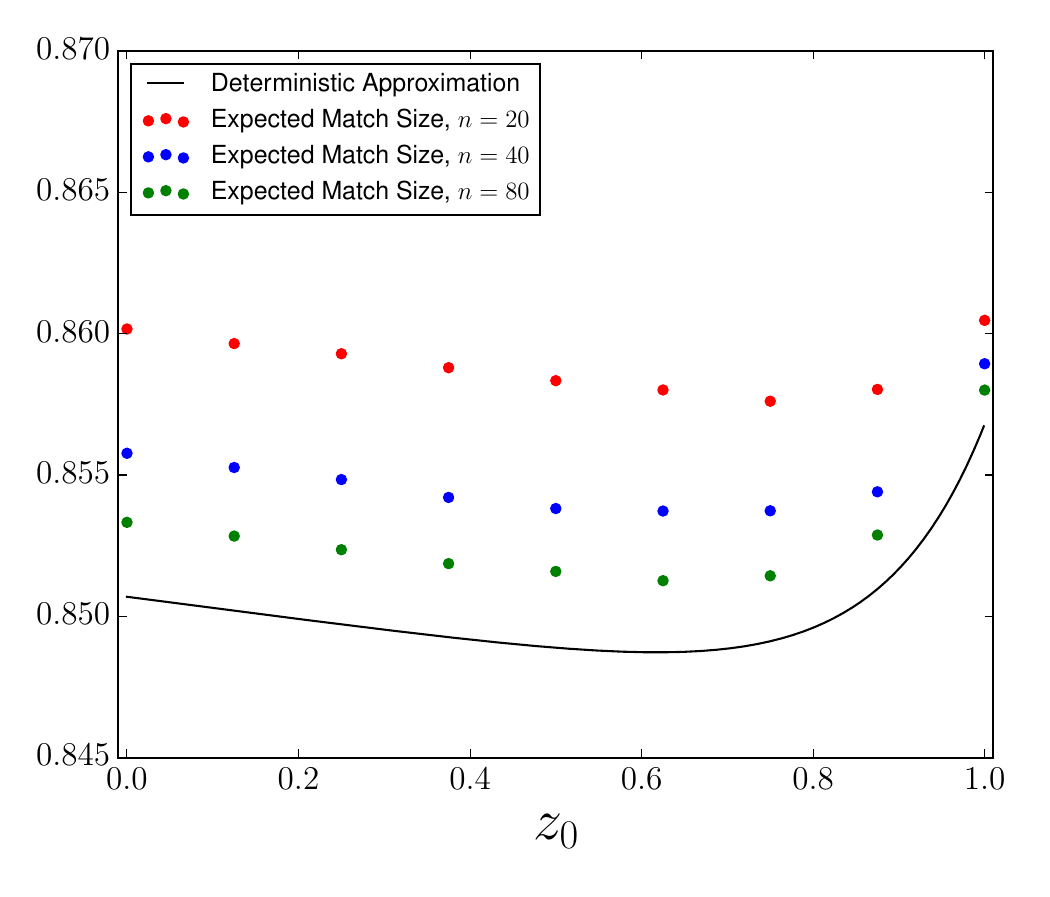}
 %\caption{ }\label{fig:numerics2}
\end{subfigure}
\caption{{\it Left}: Convergence to the deterministic approximation of the matching function when both sides of the market are of size $n$ and compatibility probability is $c/n$. {\it Right}: Comparing the expected value of $m_1(z_0)$ to its deterministic approximation in an identical setting to the left panel of Figure \ref{fig:m1z0}, except with initial condition $(d_0, k_0, v_0) = (0.8, 0.8, 0.8)$ to ensure the integrality of $nm_0$.}
\label{fig:numerics}
\end{figure}

Our numerical analysis also suggests that the structure of the optimal myopic policy remains unchanged: {either fully allowing adoption or disallowing adoption} will be optimal. In the right panel of Figure \ref{fig:numerics}, we use simulation to find the expected number of matches for various levels of adoption {and different values of the scaling factor $n$}, and we compare those results to our deterministic approximation. We make the following observations from the plot: (i) the shape of $m_1(z_0)$, i.e., the deterministic approximation, is broadly consistent with the simulated results. However, (ii) the deterministic matching function is always an underestimate, and (iii) the gap is smallest at $z_0 = 1$, when most of the matching is pre-determined due to adoption. Observations (ii) and (iii) imply that the deterministic approximation relatively overvalues adoption. 
Hence we conjecture if the optimal {myopic} policy according to the deterministic approximation is $z_t = 0$ (i.e., if neither condition in \eqref{eq:optmyopic} holds), then the optimal {myopic} policy is $z_t = 0$ when considering the actual expected matching.

\end{subsection}

	\begin{subsection}{Proof of Lemma \ref{lemma:myopicisbest}}
		\label{proof:lemma:myopicisbest}
		
			Consider two policies $\mathbf{\hat{z}} = \{\hat{z}_t : t \in [T]\} $ and $\mathbf{\tilde{z}}= \{\tilde{z}_t : t \in [T]\}$ that differ only at time $\tau$, where $\hat{z}_\tau = z_\tau'$ and $\tilde{z}_\tau \neq z_\tau'$. Using identical notation to denote the state variables when following each policy, we have $\hat{m}_{\tau+1} \geq \tilde{m}_{\tau+1}$. When the condition in Lemma \ref{lemma:myopicisbest} holds, we have $\hat{v}_{\tau+1} \geq \tilde{v}_{\tau+1}$, and by the dynamics described in \eqref{eq:dt}, $\hat{d}_{\tau+1} \geq \tilde{d}_{\tau+1}$. To complete the proof, we use Claim~\ref{claim:statecoupling} to show these conditions imply $\hat{m}_{t} \geq \tilde{m}_{t} \, \forall t \in [T] \setminus [\tau]$.
			
			\begin{claim}[Pathwise Dominance of Superior State]
			\label{claim:statecoupling}
            Suppose the platform's decisions $z_t$ are fixed for all $t \in [T']$ where $T' \leq T$. In that case, for any two initial states $(d_0, k_0, v_0)$ and $(d_0', k_0', v_0')$, if $d_0 \geq d_0'$, $v_0 \geq v_0'$, and $m_0 \geq m_0'$, then for all $t \in [T']$, $d_t \geq d_t'$, $v_t \geq v_t'$, and $m_t \geq m_t'$.
			\end{claim}
					\begin{proof}{Proof of Claim \ref{claim:statecoupling}}
We proceed via induction. We have as a base case $d_0 \geq d_0'$, $v_0 \geq v_0'$, and $m_0 \geq m_0'$. Suppose this holds for all periods $t \leq \tau$. We will show it holds for $t = \tau+1$.
		Based on the dynamics described in \eqref{eq:dt}-\eqref{eq:vt}, since the platform's policy decision $z_\tau$ is assumed to be fixed, we have
		\begin{align*}
		    d_{\tau+1} &= (1-\beta)d_\tau + \beta'm_\tau \quad \quad \geq \quad  (1-\beta)d_\tau' + \beta'm_\tau' &\geq d_{\tau+1}' \\
		   k_{\tau+1} &= (1-\gamma)z_\tau m_\tau \qquad \qquad \geq \quad (1-\gamma)z_\tau m_\tau' &\geq k_{\tau+1}' \\
		   v_{\tau+1} &= (1-\alpha)v_\tau + (\alpha'+\gamma')m_\tau + (\alpha - \alpha' - \gamma)z_\tau m_\tau & \nonumber \\ &\geq (1-\alpha)v_\tau' + (\alpha'+\gamma')m_\tau' + (\alpha - \alpha' - \gamma)z_\tau m_\tau' &\geq v_{\tau+1}'
		\end{align*}
		
		To complete the proof of the claim, we need to show that $m_{\tau+1}$ is non-decreasing in $d_{\tau+1}$, $v_{\tau+1}$, and $k_{\tau+1}$. We establish this directly, using the matching function defined in \eqref{eq:matchingfunc}.
		\begin{align*}
		\frac{\partial m_{\tau+1}}{\partial d_{\tau+1}} &= \ \  1 - \frac{e^{c(d_{\tau+1}-k_{\tau+1})}}{e^{c(d_{\tau+1}-k_{\tau+1})}+e^{c(v_{\tau+1}-k_{\tau+1})}-1} &= \frac{e^{c(v_{\tau+1}-k_{\tau+1})}-1}{e^{c(d_{\tau+1}-k_{\tau+1})}+e^{c(v_{\tau+1}-k_{\tau+1})}-1} \in [0,1] \\
		\frac{\partial m_{\tau+1}}{\partial v_{\tau+1}} &= \ \  1 - \frac{e^{c(v_{\tau+1}-k_{\tau+1})}}{e^{c(d_{\tau+1}-k_{\tau+1})}+e^{c(v_{\tau+1}-k_{\tau+1})}-1} &= \frac{e^{c(d_{\tau+1}-k_{\tau+1})}-1}{e^{c(d_{\tau+1}-k_{\tau+1})}+e^{c(v_{\tau+1}-k_{\tau+1})}-1} \in [0,1] \\
		\frac{\partial m_{\tau+1}}{\partial k_{\tau+1}} &= -1 + \frac{e^{c(d_{\tau+1}-k_{\tau+1})}+e^{c(v_{\tau+1}-k_{\tau+1})}}{e^{c(d_{\tau+1}-k_{\tau+1})}+e^{c(v_{\tau+1}-k_{\tau+1})}-1} &= \frac{1}{e^{c(d_{\tau+1}-k_{\tau+1})}+e^{c(v_{\tau+1}-k_{\tau+1})}-1} \in [0,1]
		\end{align*}
		This shows that $m_{\tau + 1} \geq m_{\tau+1}'$, which completes the proof of Claim \ref{claim:statecoupling}.  \halmos
			
					\end{proof}

		Since $\hat{z}_t = \tilde{z}_t$ for all $t \in [T] \setminus [\tau]$, we can apply Claim \ref{claim:statecoupling} to the initial states $(\hat{d}_{\tau+1}\, \hat{k}_{\tau+1}, \hat{v}_{\tau+1})$ and $(\tilde{d}_{\tau+1}, \tilde{k}_{\tau+1}, \tilde{v}_{\tau+1})$ for $T' = T - \tau-1$ to show that $\hat{m}_{t} \geq \tilde{m}_{t}$ for all $t \in [T] \setminus [\tau]$. Therefore, by replacing the decision in period $\tau$ with $z_\tau'$, the platform can only increase the total number of completed \jobs. Using a contradiction argument, no policy with $z_\tau \neq z_\tau'$ can be strictly optimal. This completes the proof that $z_\tau'$ (the optimal myopic policy) is the optimal policy in period $\tau$ (in the presence of multiple optimal policies, we follow the convention of choosing the myopically optimal one). 
	\end{subsection}
	
	\begin{subsection}{Proof of Theorem \ref{thm:optregimeB}}
		\label{proof:thm:optregimeB}
		We first show that (MMT) has a unique solution in the \regimeB. Then we show that if $d_{\tau} \geq \bar{d}$, $z_{\tau} = 0$ is optimal. Finally, we show that the thresholds define an absorbing state, i.e., if $d_{\tau} \geq \bar{d}$ and $v_{\tau} \geq \bar{v}$, then $d_{\tau+1} \geq \bar{d}$ and $v_{\tau+1} \geq \bar{v}$. This implies that once $d_t \geq \bar{d}$ and $v_t \geq \bar{v}$, $z_\tau = 0$ must be optimal for all $\tau \geq t$.
		
		To show that (MMT) has a unique solution, it is sufficient to prove that the right hand sides of (7) and (8) have slope less than $c$ and that the right hand sides of (9) and (10) have slopes in $[0,c]$ with respect to the decision variables. In that case, (10) must be tight at any optimal solution. Applying the chain rule, the right hand sides of (7)-(9) must have slopes less than $1$ as a function of $d$. Thus, a linear search for the first value of $d$ to satisfy all three constraints solves the problem.
		
		Clearly the right hand side of (7) is increasing with slope less than $c$ (it is a constant). For the other three constraints, we can exploit their common structure. We will prove that the general function $f(x) = \frac{1}{c}\log(\frac{e^{c\zeta_1x}-1}{e^{c\zeta_2x} - 1})$ satisfies $f'(x) < 1$ when $0 < \zeta_2 < \zeta_1 < 1+ \zeta_2$.
		\begin{align}
		\frac{d f}{dx} &= \frac{\zeta_1 e^{c\zeta_1x}}{e^{c\zeta_1x} - 1} - \frac{\zeta_2 e^{c\zeta_2x}}{e^{c\zeta_2x} - 1} \nonumber \\
		&= \zeta_1 + \frac{\zeta_1}{e^{c\zeta_1x} - 1} - \zeta_2  - \frac{\zeta_2}{e^{c\zeta_2x} - 1} 
		\qquad  \leq \qquad \zeta_1 - \zeta_2  \qquad < \qquad 1\label{eq:proof:mms}
		\end{align}
		Line \eqref{eq:proof:mms} comes from noting that $\frac{\zeta}{e^{cx \zeta}-1}$ is decreasing in $\zeta$, so $\frac{\zeta_1}{e^{c\zeta_1x} - 1} - \frac{\zeta_2}{e^{c\zeta_2x} - 1} \leq 0$. This immediately shows that the right hand sides of (9) and (10) both have slopes less than $c$. Similarly, the right hand side of (8) must have a slope less than $\frac{1-\beta'}{1-\beta}c < c$. Note also that $\zeta + \frac{\zeta}{e^{cx \zeta}-1}$ is increasing in $\zeta$, which means that according to \eqref{eq:proof:mms}, $\frac{d f}{dx} \geq 0$. This implies that (9) and (10) are weakly increasing in both decision variables. This completes the proof there is a unique solution to (MMT) which can be found via a linear search on $d$.
		
		We now show that if $d_\tau \geq \bar{d}$, $z_\tau = 0$ is optimal. By appealing to Lemma \ref{lemma:myopicisbest} in the \regimeB, it is sufficient to show that $cd_{\tau+1} \geq \log\left(\frac{ e^{c(1+\alpha'-\alpha)m_\tau} -1}{e^{c( \gamma + \alpha' - \alpha)m_\tau} - 1}\right)$.
		Using \eqref{eq:dt} to rewrite $d_{\tau+1}$ in terms of $d_\tau$ and $m_\tau$, we find that the above condition is equivalent to
		$cd_\tau \geq \frac{1}{1-\beta}\log\left(\frac{ e^{c(1+\alpha'-\alpha)m_\tau} -1}{e^{c( \gamma + \alpha' - \alpha)m_\tau} - 1}\right) - \frac{\beta'}{1-\beta}cm_\tau$.
		{It can be shown that this threshold is convex in $m_\tau$,} so its maximum value over its domain must occur when $m_\tau = d_\tau$ or when $m_\tau = 0$. Constraint (8) ensures that $cd_\tau$ exceeds the threshold when $m_\tau = d_\tau$, while constraint (7) ensures that $cd_\tau$ exceeds the threshold when $m_\tau = 0$ (recall that we define the threshold to equal its limiting value when $m_\tau = 0$).
		This means that if these two constraints are satisfied, the optimal myopic policy is $z_\tau = 0$. Since we are in the \regimeB, applying Lemma \ref{lemma:myopicisbest} proves that the optimal policy is $z_\tau = 0$.
		
		We now show that $\bar{d}$ and $\bar{v}$ define an absorbing state when the platform follows a policy of no adoption. Suppose $d_t \geq \bar{d}$ and $v_t \geq \bar{v}$. Since the matching function in \eqref{eq:matchingfunc} is increasing in each side of the market, $m_t$ is lower-bounded by 
		$\bar{m} := \s(\bar{d}, \bar{v}) = \bar{d}+\bar{v} - \frac{1}{c} \log(e^{c\bar{d}}+e^{c\bar{v}}-1)$. When following a policy of $z_t = 0$, $v_{t+1} = (1-\alpha)v_t + (\alpha' + \gamma ')m_t \geq (1-\alpha) \bar{v} +(\alpha' + \gamma') \bar{m}$. Though we omit the details for the sake of brevity, constraint (9) is algebraically equivalent to a constraint $(1-\alpha) \bar{v} +(\alpha' + \gamma')\bar{m} \geq \bar{v}$. Thus, $v_{t+1} \geq \bar{v}$. 
		
		Similarly, $d_{t+1} = (1-\beta)d_t +\beta'm_t \geq (1-\beta) \bar{d} +\beta' \bar{m}$. Constraint (10) is algebraically equivalent to a constraint $(1-\beta) \bar{d} +\beta' \bar{m} \geq \bar{d}$ (we again omit the details for brevity). Thus, when following a policy of no adoption, $\bar{v}$ and $\bar{d}$ define an absorbing state where no adoption remains the optimal policy.

		\if false
		By constraint \eqref{eq:mmscon2}, 
		$\bar{d} \geq \frac{1}{c} \log \left(\frac{e^{p\bar{v}}-1}{e^{p\bar{v}(1-\frac{\alpha}{\alpha'+\gamma'})}-1} \right)$. Rearranging terms, we have
		\begin{align}
		e^{c(\bar{d} + \bar{v}(1-\frac{\alpha}{\alpha'+\gamma'}))} 
		&\geq e^{c\bar{d}} + e^{c\bar{v}} - 1 \\
		&\geq e^{c(\bar{d} + \bar{v} - \bar{m})} \label{eq:proof:mms2}
		\end{align}
		Line \eqref{eq:proof:mms2} comes from rearranging the definition of $\bar{m}$ in \eqref{eq:proof:mms1.5}. This inequality is equivalent to $\bar{d} + \bar{v}(1-\frac{\alpha}{\alpha'+\gamma'} \geq \bar{d} + \bar{v} - \bar{m}$. Rearranging terms, we see that
		\begin{align}
		\bar{v} &\leq (1-\alpha)\bar{v} + (\alpha'+\gamma')\bar{m}
		\end{align}
		If $v_\tau \geq \bar{v}$ and $d_\tau \geq \bar{v}$, we must have $m_\tau \geq \bar{m}$ since the matching function is increasing in both sides of the matching market. Thus, \begin{equation}
		\bar{v} \leq (1-\alpha)v_\tau + (\alpha'+\gamma')m_\tau = v_{\tau+1}(0)
		\end{equation}
		An identical proof using constraint \eqref{eq:mmscon3} shows that $\bar{d} \leq d_{\tau+1}(0)$ when $v_\tau \geq \bar{v}$ and $d_\tau \geq \bar{v}$. Thus, when following a policy of no adoption, $\bar{v}$ and $\bar{d}$ define an absorbing state where no adoption remains the optimal policy.
		\fi
	\end{subsection}
	
	\begin{subsection}{Proof of Theorem \ref{thm:approxregimeB}}
		\label{proof:thm:approxregimeB}
		For ease of notation throughout the proof, we define constants $A_1:= \max\{1-\beta+\beta', 1-\alpha+\alpha'+\gamma'\}$, $A_2:= \min\{1-\beta+\beta', 1-\alpha+\alpha'+\gamma'\}$ and $A_3:=\max\{\beta', \alpha'+\gamma'\}$. Note that $A_1 \geq A_2 \geq 1$. 
		
		%We first establish an upper bound on the value of the \problem. Then, for each period, we upper bound the difference between the number of completed \jobs \ when following a no-adoption policy and the number of \jobs \ completed in the upper bound. Ultimately, this allows us to lower bound the performance of a static no-adoption policy as a function of the upper bound on the value of the \problem.
		
		We first establish an upper bound on the value of the \problem. Consider a modified setting where the matching in each period is equal to the small side in the matching market, e.g., $\hat{m_t}(d_t, k_t, v_t)  = \min\{d_t, v_t\}$. We will call this the \emph{\matchmin} setting, and we will use consistent notation to indicate the state variables in the \matchmin \ setting, e.g. $\hat{d}_{t+1} = (1-\beta)\hat{d_t} + \beta'\hat{m}_t$. In this setting, the platform's decision at time $t$ only impacts $\hat{m}_{t+1}$ through its impact on $\hat{v}_{t+1}$. Thus, intuitively in the \regimeB, the platform's optimal policy is to set $\hat{z}_t = 0$ for all $t \in [T]$, which maximizes the growth of the \worker \ side of the market. Based on a similar argument as in the proof of Lemma \ref{lemma:myopicisbest}, we claim the value of completed \jobs \ in the \matchmin \ setting when $\hat{z}_t = 0$ for all $t \in [T]$ is an upper bound on the value of the \problem. The value of this upper bound is given by $\sum_{t=1}^T \delta^{t-1}\hat{m}_t$.
		
		\if false
		We remark that the the size of the matching as given in \eqref{eq:matchingfunc} is less than the minimum side of the market:
		\begin{align}
		m_t(d_t, k_t, v_t) &= d_t + v_t - k_t - \frac{1}{c}\log(e^{c(d_t-k_t)} + e^{c(v_t-k_t)} -1) \\
		&= \min\{d_t, v_t\} -\frac{1}{c}\log(1 + e^{c(\min\{d_t, v_t\}-\max\{d_t, v_t\})} - e^{-c(\max\{d_t, v_t\}-k_t)}) \label{eq:proof:approx1}\\
		&\leq \min\{d_t, v_t\}
		\end{align}
		Thus, given the same initial state $(d_0, k_0, v_0)$ and the same policy choice $z_0$, we have $\hat{m}_1 \geq m_1$, $\hat{d}_1 \geq d_1$, and $\hat{v}_1 \geq v_1$. Using a similar argument as in the proof of Lemma \ref{lemma:myopicisbest}, these inequalities will continue to hold for all $t \in [T]$ as long as the policies $z_t$ are identical. As a consequence, the value of completed \jobs \ in the \matchmin \ setting when $z_t = 0$ for all $t \in [T]$ is an upper bound on the value of the \problem. The value of this upper bound is given by $\sum_{t=1}^T \delta^{t-1}\hat{m}_t$.
		\fi
		
		Now suppose $m_t$ represents the number of matches in period $t$ when the platform follows a static policy of no adoption and the matching output is determined by \eqref{eq:matchingfunc}. We aim to lower-bound the performance ratio of this policy by lower-bounding $1- \frac{\sum_{t=1}^T \delta^{t-1}(\hat{m}_t-m_t)}{\sum_{t=1}^T \delta^{t-1}\hat{m}_t}$. First, we place a lower bound on the denominator:
		\begin{align}
		\sum_{t=1}^T \delta^{t-1}\hat{m}_t \geq \sum_{t=1}^T \delta^{t-1}(A_2)^t\min\{d_0, v_0\} =A_2\min\{d_0, v_0\} \frac{1-(A_2\delta)^{T}}{1-A_2\delta} \geq \min\{d_0, v_0\} \frac{1-(A_1\delta)^{T}}{1-\delta} \label{eq:proof:approxub}
		\end{align}
		
		Next, we place an upper bound on $\sum_{t=1}^T \delta^{t-1}(\hat{m}_t-m_t)$. 
		Using the matching function defined in \eqref{eq:matchingfunc},
		%From Line \eqref{eq:proof:approx1}, we see that $m_t \geq \min\{d_t, v_t\} - \frac{1}{c}\log(2)$. 
		%This allows us to upper bound the difference between $\hat{m}_t$ and $m_t$ in any period $t \in [T]$:
		\begin{align}
		\hat{m}_t - m_t &= \min\{\hat{d}_t, \hat{v}_t\} - d_t - v_t + k_t + \frac{1}{c}\log(e^{c(d_t-k_t)} + e^{c(v_t-k_t)} -1) \nonumber \\
		&\leq  \min\{\hat{d}_t, \hat{v}_t\} - \min\{d_t, v_t\} + \frac{\log(2)}{c} \qquad \leq \qquad \max\{\hat{d}_{t} - d_{t}, \hat{v}_{t} -v_{t}\} + \frac{\log(2)}{c} \label{eq:proof:matchdiff}
		\end{align}
		
		We will now show via induction that $\max\{\hat{d}_{t} - d_{t}, \hat{v}_{t} -v_{t}\} \leq \frac{A_3\log(2) (A_1^t-1)}{c(A_1-1)}$. This holds by equality at $t = 0$. Now we assume it holds for $t \leq k$ and we aim to show that it continues to hold at $t = k+1$. We begin with the volunteer side of the market:
		\begin{align}
		    \hat{v}_{k+1} - v_{k+1} &= (1-\alpha)(\hat{v}_k - v_k) + (\alpha'+\gamma')(\hat{m}_k - m_k) \nonumber \\
		    &\leq (1-\alpha)\max\{\hat{v}_k - v_k, \hat{d}_k - d_k\}+ (\alpha'+\gamma')(\max\{\hat{d}_{k} - d_{k}, \hat{v}_{k} -v_{k}\} + \frac{\log(2)}{c}) \nonumber \\
		    &\leq A_1\max\{\hat{d}_{k} - d_{k}, \hat{v}_{k} -v_{k}\} + \frac{A_3 \log(2)}{c} \nonumber
		\end{align}
		Repeating a nearly identical process for donors, we can show $\hat{d}_{k+1} - d_{k+1} \leq A_1\max\{\hat{d}_{k} - d_{k}, \hat{v}_{k} -v_{k}\} + \frac{A_3 \log(2)}{c}$. Thus, applying our inductive hypothesis:
		\begin{align*}
		    \max\{\hat{d}_{k+1} - d_{k+1}, \hat{v}_{k+1} -v_{k+1}\} \quad \leq \quad  A_1 \left(\frac{A_3\log(2) (A_1^k-1)}{c(A_1-1)}\right) + \frac{A_3 \log(2)}{c} \quad = \quad \frac{A_3\log(2) (A_1^{k+1}-1)}{c(A_1-1)}
		\end{align*}

		\if false
		%%%%%
		To place an upper bound on $\max\{\hat{d}_{t} - d_{t}, \hat{v}_{t} -v_{t}\}$, note that when the platform follows a static policy of no adoption, the \worker \  dynamics are given by
		\begin{align}
		v_{t+1} &=(1-\alpha)v_t +(\alpha'+\gamma')m_t \\
		&\geq (1-\alpha)v_t +(\alpha'+\gamma')(\min\{d_t, v_t\} - \frac{1}{c}\log(2))\\  
		&= (1-\alpha)\hat{v}_t +(\alpha'+\gamma')(\min\{\hat{d}_t, \hat{v}_t\})
		-(1-\alpha)(\hat{v}_t - v_t) - \nonumber \\& \quad (\alpha'+\gamma')(\min\{\hat{d}_t, \hat{v}_t\} - \min\{d_t,v_t\})  - \frac{\alpha' + \gamma'}{c}\log(2) \\
		&\geq \hat{v}_{t+1} -(1-\alpha)(\hat{v}_t - v_t)- (\alpha'+\gamma')(\max\{\hat{d}_t - d_t, \hat{v}_t -v_t\})  - \frac{\alpha' + \gamma'}{c}\log(2) \\
		&\geq \hat{v}_{t+1} - (1-\alpha+\alpha'+\gamma')(\max\{\hat{d}_t - d_t, \hat{v}_t -v_t\}) - \frac{\alpha' + \gamma'}{c}\log(2)
		\end{align}
		Similarly,
		\begin{align}
		d_{t+1} &\geq (1-\beta)d_t +(\beta')(\min\{d_t, v_t\} - \frac{1}{c}\log(2))\\  
		&= (1-\beta)\hat{d}_t +(\beta')(\min\{\hat{d}_t, \hat{v}_t\})
		-(1-\beta)(\hat{d}_t - d_t) - \nonumber \\& \quad (\beta')(\min\{\hat{d}_t, \hat{v}_t\} - \min\{d_t,v_t\}) - \frac{\beta'}{c}\log(2) \\
		&\geq \hat{d}_{t+1} -(1-\beta)(\hat{d}_t - d_t)- (\beta')(\max\{\hat{d}_t - d_t, \hat{v}_t -v_t\}) - \frac{\beta'}{c}\log(2) \\
		&\geq \hat{d}_{t+1} - (1-\beta+\beta')(\max\{\hat{d}_t - d_t, \hat{v}_t -v_t\}) - \frac{\beta'}{c}\log(2)
		\end{align}
		We can combine these two inequalities to yield
		\begin{align}
		\max\{\hat{d}_{t+1} - d_{t+1}, \hat{v}_{t+1} -v_{t+1}\} &\leq \max\{(1-\beta + \beta'), (1-\alpha +\alpha' +\gamma')\}\max\{\hat{d}_t - d_t, \hat{v}_t -v_t\} \nonumber \\ &\quad + \frac{\max\{\beta', \alpha'+\gamma'\}}{c} \log (2) \\
		&=A_1\max\{\hat{d}_t - d_t, \hat{v}_t -v_t\}+\frac{A_3}{c}\log(2)
		\end{align}
		
		Given initial condition $\max\{\hat{d}_{0} - d_{0}, \hat{v}_{0} -v_{0}\}=0$, we have
		\begin{align}
		\max\{\hat{d}_{t} - d_{t}, \hat{v}_{t} -v_{t}\} &\leq \frac{A_3\log(2)}{c} \sum_{\tau=0}^{t-1} A_1^\tau \\
		&= \frac{A_3\log(2) (A_1^t-1)}{c(A_1-1)}
		\end{align}

		%%%
		\fi
		
		\noindent	This completes the proof by induction. We now plug our upper bound on $\max\{\hat{d}_{t} - d_{t}, \hat{v}_{t} -v_{t}\}$ into \eqref{eq:proof:matchdiff} and take a discounted sum over all $t \in [T]$:
		\begin{align}
		\sum_{t=1}^T \delta^{t-1}(\hat{m}_t - m_t) \quad \leq \quad  \sum_{t=1}^T \delta^{t-1}\left(\frac{\log(2)}{c} + \frac{A_3\log(2) (A_1^t-1)}{c(A_1-1)}\right) \quad  \leq \quad  \frac{\log(2)}{c(1-\delta)} + \frac{A_3 \log(2) (1-(\delta A_1)^{T})}{c(1-\delta)(1-\delta A_1)} \label{eq:multeverythingout}
		\end{align}
		
				\if false
		\begin{align}
		\sum_{t=1}^T \delta^{t-1}(\hat{m}_t - m_t) &\leq \sum_{t=1}^T \delta^{t-1}\left(\frac{\log(2)}{c} + \frac{A_3\log(2) (A_1^t-1)}{c(A_1-1)}\right) \\
		&= \frac{\log(2)}{c(A_1-1)}\sum_{t=1}^T \delta^{t-1}(A_1 - 1 - A_3 +A_3A_1^t)) \\
		&= \frac{\log(2)}{c(A_1-1)}\left(\frac{(A_1-1-A_3 )(1-\delta^{T})}{1-\delta} + \frac{A_1A_3(1-(\delta A_1)^{T})}{1-\delta A_1} \right) \\
		&\leq \frac{\log(2)}{c(1-\delta)} + \frac{A_3 \log(2) (1-(\delta A_1)^{T})}{c(1-\delta)(1-\delta A_1)}
		\end{align}
		\fi 
		
		The second inequality in \eqref{eq:multeverythingout} is purely algebraic. We omit the details for brevity. Using this result and the lower bound from \eqref{eq:proof:approxub}, we can lower bound the approximation factor of a static no-adoption policy (assuming $\delta A_1 < 1$). 
		\begin{align*}
	    1-	\frac{\sum_{t=1}^T \delta^{t-1}(\hat{m}_t - m_t)}{\sum_{t=1}^T \delta^{t-1}(\hat{m}_t)}
		\qquad \geq \qquad 1 - \frac{\frac{\log(2)}{c(1-\delta)} + \frac{A_3 \log(2) (1-(\delta A_1)^{T})}{c(1-\delta)(1-\delta A_1)}}{\min\{d_0, v_0\} \frac{1-(A_1\delta)^{T}}{1-\delta}} \qquad \geq \qquad 1 - \frac{\log(2)(\frac{1}{1-(\delta A_1)^T} + \frac{A_3}{1-\delta A_1})}{c\min\{d_0, v_0\}}
		\end{align*}
		This is equivalent to the approximation factor reported in Theorem \ref{thm:approxregimeB}.
	\end{subsection}
\end{section}

\begin{section}{Omitted Details from Section \ref{sec:hetero}}
\begin{subsection}{Proof of Proposition \ref{prop:myopichetero}}
\label{app:proof:prop:myopichetero}
\revcolor{The proof of Proposition \ref{prop:myopichetero} follows quite naturally from the proof of Theorem \ref{thm:myopic}. In this extended setting, we can once again show that $m_1(z_0)$ does not attain a local maximum. 

As before, $\frac{\partial d_{1}}{\partial z_0} = 0$. Based on equation \eqref{eq:kthetero} and the definition of $\hat{D}_k$, $\frac{\partial k_{1}}{\partial z_0} = (1-\gamma)\sum_{j \in \{p,s\}}\adoptp^j(m_0^j - k_0^j) = \hat{D}_k$. Based on equation \eqref{eq:vthetero} and the definition of $\hat{D}_v$, $\frac{\partial v_{1}}{\partial z_0} = \sum_{j \in \{p,s\}}(\alpha - \alpha'^j-\gamma)\adoptp^j(m_0^j - k_0^j) = \hat{D}_v$. We emphasize that $\hat{D}_k$ and $\hat{D}_v$ depend only on instance primitives and are constant with respect to $z_0$.

Combining these derivatives with \eqref{eq:mt} and \eqref{eq:matchingfunc}, we have 
		\begin{align}
		m_1 \ \ &= d_1 - k_1 + v_1  - \frac{1}{c} \log(e^{c(d_1-k_1)}+e^{c(v_1 - k_1)}-1) \nonumber  \\
		\frac{\partial m_1}{\partial z_0} \ &= \frac{\hat{D}_ve^{c(d_1-k_1)} + \hat{D}_k-\hat{D}_v}{e^{c(d_1-k_1)}+e^{c(v_1 - k_1)}-1} \label{eq:dmdzhetero}\\
		\frac{\partial^2 m_1}{\partial z_0^2} &= \left( \frac{\hat{D}_k^2e^{c(d_1-k_1)} + (\hat{D}_k-\hat{D}_v)^2e^{c(v_1-k_1)}-\hat{D}_v^2e^{c(d_1+v_1-2k_1)}}{(e^{c(d_1-k_1)}+e^{c(v_1 - k_1)}-1)^2}\right)c \label{eq:d2mdz2hetero}
		\end{align}
		The first-order condition prescribed by \eqref{eq:dmdzhetero} is equivalent to $e^{c(d_1-k_1)} = {\frac{\hat{D}_v-\hat{D}_k}{\hat{D}_v}}$. When the FOC holds, the numerator of \eqref{eq:d2mdz2hetero} reduces to $\hat{D}_k^2e^{c(d_1-k_1)} + \hat{D}_k(\hat{D}_k-\hat{D}_v)e^{c(v_1-k_1)}$. This must be positive, since $\hat{D}_k > \hat{D}_v$ and $\hat{D}_k > 0$. Thus, $m_1$ is convex at any critical point and consequently can have no local maxima as a function of $z_0$. This implies that the optimal solution must be at a boundary, i.e., $z_0 = 0$ or $z_0 =1$.

		To compare the values $m_1(0)$ and $m_1(1)$, we use a slight abuse of notation to augment $(d_1, k_1, v_1)$ by $z_t$ to compare them for $z_t = 0$ and $z_t = 1$. Since $d_1(0) = d_1(1) := d_1$ and $k_1(0) = 0$, the following are necessary and sufficient conditions for $m_1(1) \geq m_1(0)$:
		\begin{align}
0 &\geq e^{c(v_1(0) - v_1(1) + k_1(1))} - \frac{e^{cv_1(0)}+e^{cd_1} -1}{e^{c(v_1(1)-k_1(1))}+e^{c(d_1-k_1(1))} -1}  \label{eq:proof:myopic0hetero} \\
\Leftrightarrow ~~ 0&\geq e^{c({\hat{D}_k-\hat{D}_v})}(e^{cd_1}e^{-c{\hat{D}_k}} -1) - e^{cd_1} +1  \label{eq:proof:myopic1hetero} \\
\Leftrightarrow ~~ 0&\geq e^{cd_1}(e^{{-c\hat{D}_v}} - 1) - (e^{c({\hat{D}_k-\hat{D}_v})}-1) \label{eq:proof:myopic2hetero}
		\end{align}
		Line \eqref{eq:proof:myopic0hetero} comes from applying definitions,
		Line \eqref{eq:proof:myopic1hetero} is algebraic, using the equalities $v_1(0) - v_1(1) + k_1(1) = {\hat{D}_k-\hat{D}_v}$ and $k_1(1) = {\hat{D}_k}$. Inequality \eqref{eq:proof:myopic2hetero} follows from rearranging, and it is satisfied when either $\hat{D}_v$ is non-negative or when the second condition in \eqref{eq:optmyopichetero} is met. We conclude the proof by using the definition of $m_t^j$ in Line \eqref{eq:mthetero} to show that $\hat{D}_v \geq 0$ is equivalent to the first condition in Line \eqref{eq:optmyopichetero}:
	\begin{align}
0 &\leq \sum_{j \in \{p,s\}}(\alpha - \alpha'^j-\gamma)\adoptp^j(m_0^j - k_0^j)
\\
\Leftrightarrow ~~ 0 &\leq \sum_{j \in \{p,s\}}(\alpha - \alpha'^j-\gamma)\adoptp^j\left(\frac{v_0^j - k_0^j}{v_0-k_0}\right)\s(d_0-k_0, v_0-k_0) \nonumber \\
\Leftrightarrow ~~  0&\leq \sum_{j \in \{p,s\}}(\alpha - \alpha'^j-\gamma)\adoptp^j(v_0^j - k_0^j) \nonumber \\
\Leftrightarrow ~~  0&\leq \alpha\left(\sum_{j \in \{p,s\}}\adoptp^j(v_0^j - k_0^j)\right) - \gamma \left(\sum_{j \in \{p,s\}}\adoptp^j(v_0^j - k_0^j)\right) - \sum_{j \in \{p,s\}} \alpha'^j\adoptp^j(v_0^j - k_0^j) \nonumber \\
\Leftrightarrow ~~  \gamma &\leq \alpha - \frac{\sum_{j \in \{p,s\}}  \adoptp^j(v_0^j - k_0^j)\alpha'^j}{\sum_{j \in \{p,s\}}\adoptp^j(v_0^j - k_0^j)~~~} \nonumber
\end{align}
}

\end{subsection}

\revcolor{
\begin{subsection}
{Sensitivity to Volunteers' Adoption Probability}
\label{app:comp_statics_adoptp}
In practice, when platforms such as FRUS \emph{allow} adoption, volunteers may still choose to form a \temporary\ match. In such settings, the optimal policy (and its performance) depends on the probability that a volunteer will form an adopted match when such an option is allowed. %(i.e., $\adoptp$).
In the following, we study the impact of adoption probability on the optimal policy and its performance in a special case of the model introduced in Section~\ref{subsec:heteromodel} where students and professionals have the same parameters, i.e. $\alpha'^s=\alpha'^p$ and $\adoptp^s=\adoptp^p=\adoptp$.

{Intuitively, a lower adoption probability can only harm the performance of the optimal policy. If volunteers are less willing to adopt, this gives the platform less control over the type of matches which are formed, which is equivalent to reducing the action space of the platform. 
For instance, in a myopic setting, suppose allowing adoption for all donations is the optimal policy when the adoption probability is given by $\adoptp$. Then, if the adoption probability decreases to $\hat{\adoptp} < \adoptp$, the performance of the 
{policy that allows adoption for all donations will decrease, to the extent that such a policy may even become sub-optimal.} However, if disallowing adoption is the optimal policy when the adoption probability is given by $\adoptp$, then a decrease in the adoption probability to $\hat{\adoptp} < \adoptp$ will not impact the performance or the optimality of the policy that disallows adoption.}
%When disallowing adoption is myopically optimal even when $\rho=1$, then a decrease in the adoption probability will not impact the optimal policy or its performance. In contrast, when allowing adoption for all donations is the myopically optimal policy under $\rho=1$, then a decrease in the adoption probability will reduce the performance of that policy, and may even make it sub-optimal.

We provide numerical support for this intuition in Figure \ref{fig:comp_statics_adoptp}, which plots the myopic number of matches achieved by policies of $z_0 = 1$ (fully allowing adoption) and $z_0=0$ (disallowing adoption) as a function of the adoption probability. Specifically, we use our model for heterogeneous volunteers (presented in Section \ref{sec:hetero}) which includes a parameter for adoption probability. For simplicity, we assume that the two volunteer types have identical parameters (thus, we drop the superscript from our notation). In the image shown, $(\alpha, \alpha', \gamma, \gamma', \beta, \beta') = (0.1, 0.1, 0.05, 0.1, 0.05, 0.1)$, $v_0 = d_0 = 1$, $k_0=0$, and $c=5$. We note that these parameters are equal to the expected value of the simulation space described in Appendix \ref{app:simulations}. 

\begin{figure}[t]
	\centering
	\includegraphics[width=.45\textwidth]{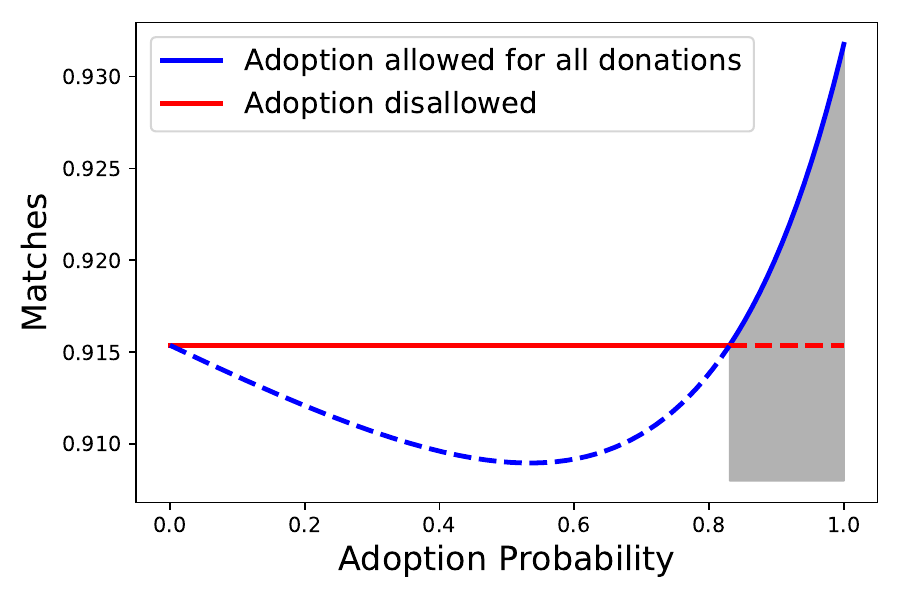}
	\caption{\revcolor{For model primitives $(\alpha, \alpha', \gamma, \gamma', \beta, \beta') = (0.1, 0.1, 0.05, 0.1, 0.05, 0.1)$, $v_0 = d_0 = 1$, $k_0=0$, and $c=5$, the plot shows the myopic number of matches achieved by the only two potentially-optimal policies, as a function of the adoption probability $\adoptp$.}
	}
	\label{fig:comp_statics_adoptp}
		\end{figure}

As shown by the solid blue curve, fully allowing adoption is the myopically optimal policy when $\adoptp = 1$ (i.e., when the platform has complete control over the types of matches formed). As long as the adoption probability lies within the shaded region on the right, this policy remains optimal, but we highlight that its performance degrades noticeably as the adoption probability decreases. If the adoption probability is below a threshold  (which occurs at an adoption probability of $0.831$ in this instance), then disallowing adoption is myopically optimal, as shown by the solid red line. Below this threshold, changes in the adoption probability do not impact the performance of the optimal policy.

We remark that the platform has significant incentive to try to increase the adoption probability whenever fully allowing adoption is optimal, and each percentage-point improvement is increasingly valuable. Improving the number of matches, even by a small amount, can make a significant difference in the growth trajectory of the platform.
\end{subsection}
}

\revcolor{
\subsection{Numerical Performance of a Repeated Optimal Myopic Policy in Extended Model}
\label{app:simulations}

\begin{table}[t]
    \caption{\revcolor{Distribution for model primitives in simulated instances described in Section \ref{subsec:heterolongrun}. All primitives are drawn independently for each instance, aside from $\gamma'$, $\beta'$, and $\delta$ which are drawn from ranges which depend on other primitives. }}
    \centering
    \revcolor{\footnotesize
    \begin{tabular}{clccl}
    & \\
    Model Primitive &  \multicolumn{1}{c}{Distribution}&\qquad \qquad & Model Primitive &  \multicolumn{1}{c}{Distribution} \\
    \cline{1-2} \cline{4-5} 
    $d_0$ & Deterministically 1&&  $c$ & Uniform$[0,10]$  \\ 
    $k_0^p$ & Deterministically 0 && $k_0^s$ & Deterministically 0  \\
    $v^p_0$ & Uniform$[0,1]$ && $v^s_0$ & Uniform$[0,1]$ \\
   $\beta$&Uniform$[0,0.1]$  &&  $\beta'$&Uniform$[\beta, \beta+0.1]$ \\ $\gamma$&Uniform$[0,0.1]$ &&$\gamma'$&Uniform$[\gamma,\gamma+0.1]$ \\
   $\adoptp^p$ & Uniform$[0,1]$&& $\adoptp^s$ & Uniform$[0,1]$ \\
        $\alpha'^p$&Uniform$[0,0.2]$
          &&$\alpha'^s$&Uniform$[0,0.2]$ \\
          $\alpha$&Uniform$[0,0.2]$&& $\delta$ & Uniform$[0.70, 1-B_1]$\\
    \end{tabular}
    }
    \label{table:simspace}
\end{table}

In Section \ref{subsec:heterolongrun}, we compare the performance of a repeated optimal myopic policy to the performance of the optimal policy across 10,000 simulations over a horizon $T=10$. For each of these simulations, we generated instance primitives according to Table \ref{table:simspace}.  We normalized the initial  donation side of the market to size $1$, and for simplicity, we assumed no initial adoption (i.e., $k_0^p = k_0^s = 0$). Most instance primitives were generated independently. The only exceptions are $\gamma'$ and $\beta'$, which were chosen to be large enough to ensure growth is possible, as well as the discount rate $\delta$, which is upper bounded by $1-B_1$ (where $B_1 := \max\{\gamma'-\gamma, \gamma'-\alpha +\alpha'^p, \gamma'-\alpha +\alpha'^s, \beta'-\beta\}$ is the maximum possible growth rate) to ensure that the discounted size of the market remains bounded.

For each of the 10,000 instances, we found the {policy that, in each period, optimally chooses either full adoption or no adoption via exhaustive search}. We used this solution as a proxy for the optimal policy. Separately, we randomly selected and tested 500 of the instances and verified that in each one, {the optimal policy $\vec{z}^*$ was indeed always either either full adoption or no adoption in each period, i.e., $\vec{z}^* \in \{0,1\}^{T}$.}

{We highlight that a repeated optimal myopic policy performs quite well: in 90.7\% of the simulated instances, it is exactly optimal.} In Figure \ref{fig:simhist}, we focus our attention on the subset of 929 instances where a repeated optimal myopic policy is suboptimal. We display a histogram showing the performance ratio of a repeated optimal myopic policy compared to the optimal policy. Even in those instances, the discounted number of matches when following a repeated optimal myopic policy is quite close to optimal. In the worst case that we observe, the repeated optimal myopic policy achieves 83.2\% of the optimal discounted number of completed matches.

\begin{figure}
 \includegraphics[width=.7\textwidth]{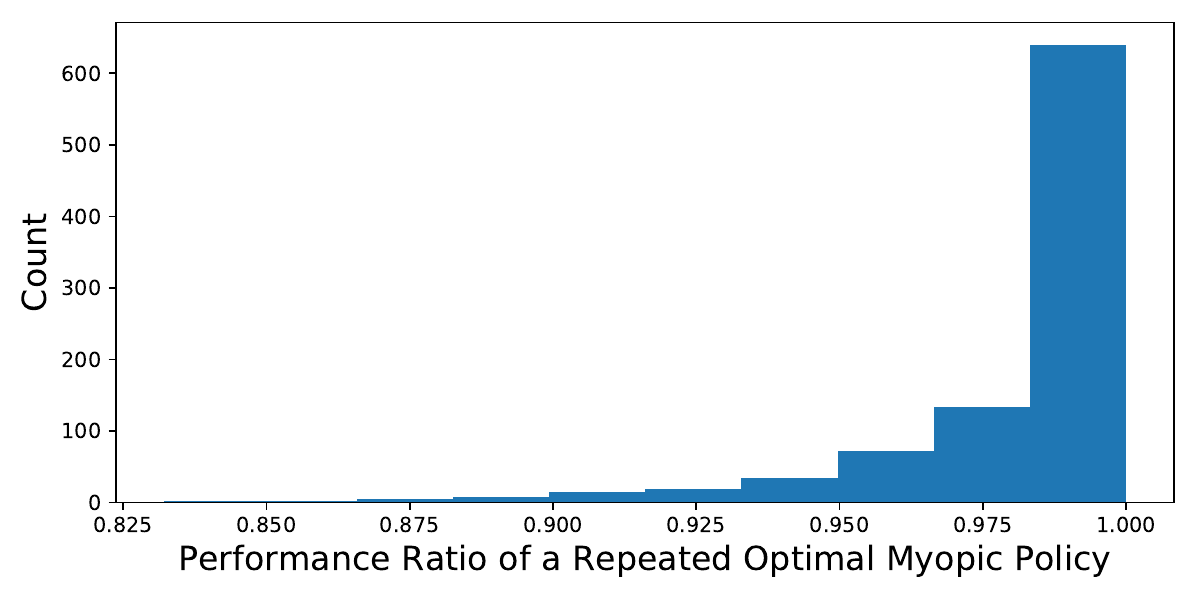}
  \caption{\revcolor{Histogram of the 929 sub-optimal performance ratios (out of 10,000 simulated instances) when following a repeated optimal myopic policy.}}%
  \label{fig:simhist}
\end{figure}

}

%%%%%%%%%%%%%%%%%%%%%%%%%%%%%%%%%%%%%%%%%%%%%%%

%   New approx proof

%%%%%%%%%%%%%%%%%%%%%%%%%%%%%%%%%%%%%%%%%%%%%%%

\begin{subsection}{Proof of Proposition \ref{prop:approxrepeatmyopic}}
\label{app:proof:prop:approxrepeatmyopic}
\revcolor{
In this proof, we will establish a lower bound on the ratio between the discounted number of matches when following a repeated optimal myopic policy (given by $\sum_{t=1}^T \delta^{t-1}\tilde{m}_t$) and the discounted number of matches when following the optimal policy (given by $\sum_{t=1}^T \delta^{t-1}m_t^*$). To do so, we will upper-bound the achievable number of matches in each period (denoted $\bar{m}_t$) and we will also lower-bound the number of matches achieved by a repeated optimal myopic policy in each period (denoted $\underline{m}_t$).

%Consider a setting where the growth rate of a matched volunteer is given by $\bar{r} := $, regardless of match type.  In such a setting, there is no trade-off between match certainty and volunteer growth (the \gap\ is equal to 0), so full availability for adoption is the optimal policy in every period. Further, suppose that volunteers are willing to adopt with probability $\bar{\adoptp} := \max\{\adoptp^p, \adoptp^s\}$. We use $\bar{m}_t$ to denote the number of matches achieved in period $t$ when following a policy of full availability for adoption in this setting.

Consider a setting with perfect matching where the growth rate of a matched volunteer is given by $\bar{r} := \max\{\gamma'-\gamma, \gamma'-\alpha +\alpha'^p, \gamma'-\alpha +\alpha'^s\}$. We use $\bar{m}_t$ to denote the number of matches achieved in period $t$ in this setting.

\begin{claim}[Upper Bound on Optimal Policy] \label{clm:myopicapproxupper}
For all $t \in [T]$, the number of matches $m^*_t$ in period $t$ when following the optimal policy is at most $\bar{m}_t$.
\end{claim}

\begin{proof}{Proof:}
The proof of Claim \ref{clm:myopicapproxupper} is similar to the proof of Claim \ref{claim:statecoupling}, which demonstrates the pathwise dominance of a superior state in a homogeneous volunteer setting. In this case, we consider two settings with different dynamics. Nevertheless, we are able to show a similar pathwise dominance of $\bar{m}_t$ over $m^*_t$. 

Let $(\bar{d}_t, \bar{v}_t)$ denote the state which arises at time $t$ in the setting described above. Note that in this setting, volunteers are homogeneous, and the growth rate of all matched volunteers is $\bar{r}$. Similarly, we will use $(d^*_t, v^*_t)$ to denote the state which arises at time $t$ when following an optimal policy in the heterogeneous volunteer setting. We will show via induction that for all $t \in [T]$, $\bar{d}_t \geq d_t^*$, $\bar{v}_t \geq v^*_t$, and $\bar{m}_t \geq m^*_t$. 

The first two inequalities in the base case hold by assumption: $\bar{d}_0 = d_0^*$ and $\bar{v}_0 = v^*_0$. As a consequence of perfect matching, we must have $\bar{m}_0 \geq m^*_0$.

We now assume the inductive hypothesis holds for $t = \tau$. The donation dynamics in both settings are identical, and the number of donations in period $\tau+1$ is increasing in both the prior period's number of donations and number of matches. By our inductive hypothesis, $\bar{d}_\tau \geq d^*_\tau$ and $\bar{m}_\tau \geq m^*_\tau$. Thus, $\bar{d}_{\tau+1} \geq d^*_{\tau+1}$. 

For volunteers, the dynamics are slightly different in the two settings. Starting with the definition of $v^*_{\tau+1}$ given by equation \eqref{eq:vthetero}, we have
\begin{align}
    v^*_{\tau+1} &= \sum_{j \in\{p,s\}}(1-\alpha)v_\tau^{j*} + (\alpha'^j+\gamma')m_\tau^{j*} - (\gamma - \alpha + \alpha'^j)\left(k_\tau^{j*} + z_\tau^*\adoptp^j(m_\tau^{j*}-k_\tau^{j*})\right) \nonumber \\
    &\leq (1-\alpha)v_\tau^* + (\bar{r} + \alpha)m_\tau^* \label{eq:vtmyopicapprox1} \\
    &\leq (1-\alpha)\bar{v}_\tau + (\bar{r} + \alpha)\bar{m}_\tau \label{eq:vtmyopicapprox2} \\
    &= \bar{v}_{\tau+1} \nonumber
\end{align}
The inequality in Line \eqref{eq:vtmyopicapprox1} comes from the definition of $\bar{r}$ as an upper bound on the growth rate of a matched volunteer, regardless of the volunteer type and the match type. The subsequent inequality, Line \eqref{eq:vtmyopicapprox1}, comes from our inductive hypothesis, and the final equality holds due to the dynamics in that setting.

We have shown that $\bar{d}_{\tau+1} \geq d^*_{\tau+1}$ and $\bar{v}_{\tau+1} \geq v^*_{\tau+1}$. Again due to perfect matching, $\bar{m}_{\tau+1} \geq m^*_{\tau+1}$. This completes the proof of Claim \ref{clm:myopicapproxupper}. \halmos
\end{proof}

\if false
\begin{proof}{Proof:}
The proof of Claim \ref{clm:myopicapproxupper} is similar to the proof of Claim \ref{claim:statecoupling}, which demonstrates the pathwise dominance of a superior state in a homogeneous volunteer setting. In this case, we consider two markets where the platform may make different decisions and where the dynamics are different. Nevertheless, we are able to show a similar pathwise dominance of $\bar{m}_t$ over $m^*_t$. 

Let $(\bar{d}_t, \bar{k}_t, \bar{v}_t)$ denote the state which arises at time $t$ when following a policy of full adoption in the setting described above. Note that in this setting, volunteers are homogeneous, and the growth rate of all matched volunteers is $\bar{r}$. Similarly, we will use $(d^*_t, k^*_t, v^*_t)$ to denote the state which arises at time $t$ when following an optimal policy in the heterogeneous volunteer setting. We will show via induction that for all $t \in [T]$, $\bar{d}_t \geq d_t^*$, $\bar{k}_t = k^*_t$, $\bar{v}_t = v^*_t$, and $\bar{m}_t \geq m^*_t$.

The first three inequalities in the base case hold by assumption: $\bar{d}_0 = d_0^*$, $\bar{k}_0 = k^*_0$, and $\bar{v}_0 = v^*_0$. Because the matching process is identical in both settings, we must have $\bar{m}_0 = m^*_0$.

We now assume the inductive hypothesis holds for $t = \tau$. The donation dynamics in both settings are identical, and the number of donations in period $\tau+1$ is increasing in both the prior period's number of donations and number of matches. By our inductive hypothesis, $\bar{d}_\tau \geq d^*_\tau$ and $\bar{m}_\tau \geq m^*_\tau$. Thus, $\bar{d}_{\tau+1} \geq d^*_{\tau+1}$. 

The dynamics for fixed matches differ slightly in the two settings.
However, starting with the definition of $k^*_{\tau+1}$ given by equation \eqref{eq:kthetero}, we have 
\begin{align}
    k^*_{\tau+1} &= (1-\gamma)\left(\sum_{j \in\{p,s\}} k_\tau^{j^*} + z_\tau^*\adoptp^j(m_\tau^{j*} - k_t^{j*})\right) \nonumber \\
    &\leq (1-\gamma)(k_\tau^* +\bar{\adoptp}(m_\tau^* - k_\tau^*)) \nonumber \\
    &= (1-\gamma)(1-\bar{\adoptp})k_\tau^* +(1-\gamma)\bar{\adoptp}m_\tau^* \nonumber \\
    &\leq (1-\gamma)(1-\bar{\adoptp})\bar{k}_\tau +(1-\gamma)\bar{\adoptp}\bar{m}_\tau \label{eq:ktmyopicapprox1} \\
    & = \bar{k}_{\tau+1} \nonumber
\end{align}
We note that the the inequality in Line \eqref{eq:ktmyopicapprox1} holds by the inductive hypothesis. The final equality holds due to the dynamics in that setting: the platform always allows all donations to be adopted, and volunteers form an adopted match with probability $\bar{\adoptp}$.

Similarly, for volunteers, the dynamics are again slightly different in the two settings. Starting with the definition of $v^*_{\tau+1}$ given by equation \eqref{eq:vthetero}, we have
\begin{align}
    v^*_{\tau+1} &= \sum_{j \in\{p,s\}}(1-\alpha)v_\tau^{j*} + (\alpha'^j+\gamma')m_\tau^{j*} - (\gamma - \alpha + \alpha'^j)\left(k_\tau^{j*} + z_\tau^*\adoptp^j(m_\tau^{j*}-k_\tau^{j*})\right) \nonumber \\
    &\leq (1-\alpha)v_\tau^* + (\bar{r} + \alpha)m_\tau^* \label{eq:vtmyopicapprox1} \\
    &\leq (1-\alpha)\bar{v}_\tau + (\bar{r} + \alpha)\bar{m}_\tau \label{eq:vtmyopicapprox2} \\
    &= \bar{v}_{\tau+1} \nonumber
\end{align}
The inequality in Line \eqref{eq:vtmyopicapprox1} comes from the definition of $\bar{r}$ as an upper bound on the growth rate of a matched volunteer, regardless of the volunteer type and the match type. The subsequent inequality, Line \eqref{eq:vtmyopicapprox1}, comes from our inductive hypothesis, and the final equality holds due to the dynamics in that setting.

We have shown that $\bar{d}_{\tau+1} \geq d^*_{\tau+1}$, $\bar{k}_{\tau+1} \geq k^*_{\tau+1}$, and $\bar{v}_{\tau+1} \geq v^*_{\tau+1}$. The matching process is not only identical in both settings, but it is also increasing in each of the three state variables. Thus, $\bar{m}_{\tau+1} \geq m^*_{\tau+1}$, which completes the proof of Claim \ref{clm:myopicapproxupper}.
\end{proof}
\fi

Now consider a different setting where the growth rate of a matched volunteer is given by $\underline{r} := \min\{\gamma'-\gamma,\gamma'-\alpha+\alpha'^p,\gamma'-\alpha+\alpha'^s\} $, regardless of match type, and suppose volunteers are not willing to adopt at all, i.e., $\underline{\adoptp} = 0$.  In this setting, the platform has no decision to make. We use $\underline{m}_t$ to denote the number of matches achieved in period $t$ in this setting.

\begin{claim}[Lower Bound on repeated optimal myopic Policy] \label{clm:myopicapproxlower}
For all $t \in [T]$, the number of matches $\tilde{m}_t$ in period $t$ when following a repeated optimal myopic policy is at least $\underline{m}_t$.
\end{claim}

\begin{proof}{Proof:}
The proof of Claim \ref{clm:myopicapproxlower} is nearly identical to the proof of Claim \ref{clm:myopicapproxupper}. Once again, we aim to show a similar pathwise dominance, this time of $\tilde{m}_t$ over $\underline{m}_t$. 

Let $(\tilde{d}_t, \tilde{k}_t, \tilde{v}_t)$ denote the state which arises at time $t$ when repeatedly following the optimal myopic policy. Similarly, we will use $(\underline{d}_t, \underline{k}_t, \underline{v}_t)$ to denote the state which arises at time $t$ in the setting described above. We will show via induction that for all $t \in [T]$, $\tilde{d}_t \geq \underline{d}_t$, $\tilde{v}_t \geq \underline{v}_t$, and $\tilde{m}_t \geq \underline{m}_t$. (Recall that no volunteers will form an adopted match in the setting described above, which trivially implies $\underline{k}_t \leq \tilde{k}_t$ for all $t$).

The first two inequalities in the base case hold by assumption: $\tilde{d}_0 \geq \underline{d}_0$ and $\tilde{v}_0 \geq \underline{v}_0$. Since the matching function is increasing in all state variables (i.e., $m_0$ is increasing in $d_0$, $k_0$, and $v_0$), we must have $\tilde{m}_0 \geq \underline{m}_0$.

We now assume the inductive hypothesis holds for $t = \tau$. The donation dynamics in both settings are identical, and the number of donations in period $\tau+1$ is increasing in both the prior period's number of donations and number of matches. By our inductive hypothesis, $\tilde{d}_\tau \geq \underline{d}_\tau$ and $\tilde{m}_t \geq \underline{m}_t$. Thus, $\tilde{d}_{\tau+1} \geq \underline{d}_{\tau+1}$. 

For volunteers, the dynamics are slightly different in the two settings. Starting with the definition of $\tilde{v}_{\tau+1}$ given by equation \eqref{eq:vthetero}, we have
\begin{align}
    \tilde{v}_{\tau+1} &= \sum_{j \in\{p,s\}}(1-\alpha)\tilde{v}_\tau^{j} + (\alpha'^j+\gamma')\tilde{m}_\tau^{j} - (\gamma - \alpha + \alpha'^j)\left(\tilde{k}_\tau^{j} + \tilde{z}_\tau \adoptp^j(\tilde{m}_\tau^{j}-\tilde{k}_\tau^{j})\right) \nonumber \\   
    &\geq (1-\alpha)\tilde{v}_\tau + (\underline{r} + \alpha)\tilde{m}_\tau \label{eq:vtmyopicapprox3} \\
    &\geq (1-\alpha)\underline{v}_\tau + (\underline{r} + \alpha)\underline{m}_\tau \label{eq:vtmyopicapprox4} \\
    &= \underline{v}_{\tau+1} \nonumber
\end{align}
The inequality in Line \eqref{eq:vtmyopicapprox3} comes from the definition of $\underline{r}$ as an lower bound on the growth rate of a matched volunteer, regardless of the volunteer type and the match type. The subsequent inequality, Line \eqref{eq:vtmyopicapprox4}, comes from our inductive hypothesis, and the final equality holds due to the dynamics in that setting.

We have shown that $\tilde{d}_{\tau+1} \geq \underline{d}_{\tau+1}$ and $\tilde{v}_{\tau+1} \geq \underline{v}_{\tau+1}$. Of course, we must also have $\tilde{k}_{\tau+1} \geq \underline{k}_{\tau+1}$. Since the matching function is increasing in all state variable, $\tilde{m}_{\tau+1} \geq \underline{m}_{\tau+1}$. This completes the proof of Claim \ref{clm:myopicapproxlower}. \halmos
\end{proof}

 We now establish a bound on the gap in the number of matches in the two settings, i.e., an upper bound on the difference $\bar{m}_t - \underline{m}_t$. To aid in this part of the proof, let us define constants $B_1 := \max\{\bar{r}, \beta'-\beta\}$ and $B_2:= \max\{\underline{r} + \alpha, \underline{r} +\beta\}$.

\begin{claim}[Upper Bound on the Gap Between Optimal and repeated optimal myopic Policies] \label{clm:myopicapproxgap}
For all $t \in [T]$, 
\begin{equation}
    \bar{m}_t - \underline{m}_t \leq \min\{d_0, v_0\}(B_1-\underline{r})t(1+B_1)^{t-1} + \left(B_2\frac{(1+\underline{r})^{t}-1}{\underline{r}} + 1\right) \frac{\log(2)}{c}
\end{equation}
\end{claim}

\begin{proof}{Proof:}
We prove Claim \ref{clm:myopicapproxgap} by first relating the gap in matches ($\bar{m}_t - \underline{m}_t$) to the largest gap on a single side of the market ($\max\{\bar{d}_t - \underline{d}_t, \bar{v}_t - \underline{v}_t\}$). We then upper-bound the latter expression. 
%To help with this proof, note that $\bar{m}_t$ can be at most $\bar{v}_t$, which is bounded by $v_0(1+\bar{r})^t$.

Based on inequality \eqref{eq:proof:matchdiff} in the proof of Theorem \ref{thm:approxregimeB}, 
\begin{equation}
   \bar{m}_t - \underline{m}_t \leq \max\{\bar{d}_{t} - \underline{d}_{t}, \bar{v}_{t} -\underline{v}_{t}\} + \frac{\log(2)}{c} \label{eq:mtgapapproxmyopic}
\end{equation} 

This bound takes into account the difference between a perfect matching and the matching function described in \eqref{eq:matchingfunc}. 

We now show via induction that the following bound holds: 
\begin{equation}
    \max\{\bar{d}_{t} - \underline{d}_{t}, \bar{v}_{t} -\underline{v}_{t}\} \leq \min\{d_0, v_0\}(B_1-\underline{r})t(1+B_1)^{t-1} + B_2\frac{\log(2)}{c}\frac{(1+\underline{r})^{t}-1}{\underline{r}}. \label{eq:vtgapapproxmyopic}
\end{equation} 

This holds by equality at $t = 0$. Now we assume it holds for $t \leq k$ and we aim to show that it continues to hold at $t = k+1$. We begin with the volunteer side of the market:
		\begin{align}
		    \bar{v}_{k+1} - \underline{v}_{k+1} &= (1-\alpha)(\bar{v}_{k} - \underline{v}_{k}) + (\underline{r}+\alpha)(\bar{m}_{k} - \underline{m}_{k}) + (\bar{r} - \underline{r})\bar{m}_k \nonumber \\
		    &\leq (1-\alpha)(\max\{\bar{d}_{k} - \underline{d}_{k}, \bar{v}_{k} -\underline{v}_{k}\}) + (\underline{r}+\alpha)(\bar{m}_{k} - \underline{m}_{k}) + (\bar{r} - \underline{r})\bar{m}_k  \nonumber \\
		    &\leq (1-\alpha)(\max\{\bar{d}_{k} - \underline{d}_{k}, \bar{v}_{k} -\underline{v}_{k}\}) + (\underline{r}+\alpha)(\max\{\bar{d}_{k} - \underline{d}_{k}, \bar{v}_{k} -\underline{v}_{k}\} + \frac{\log(2)}{c}) + (B_1 - \underline{r})\bar{m}_k  \nonumber \\
		    &\leq (1+\underline{r})(\max\{\bar{d}_{k} - \underline{d}_{k}, \bar{v}_{k} -\underline{v}_{k}\}) + (\underline{r}+\alpha) \frac{\log(2)}{c} + (B_1 - \underline{r})\min\{\bar{d}_k, \bar{v}_k\}  \nonumber \\
		    &\leq (1+\underline{r})(\max\{\bar{d}_{k} - \underline{d}_{k}, \bar{v}_{k} -\underline{v}_{k}\}) + B_2 \frac{\log(2)}{c} + (B_1 - \underline{r})\min\{\bar{d}_k, \bar{v}_k\}  \nonumber
		\end{align}
		
For donors, we have
\begin{align}
    \bar{d}_{k+1} - \underline{d}_{k+1} &= (1-\beta)(\bar{d}_{k} - \underline{d}_{k}) + (\beta')(\bar{m}_{k} - \underline{m}_{k}) \nonumber \\
    &\leq (1-\beta)(\bar{d}_{k} - \underline{d}_{k}) + (\beta + \underline{r})(\bar{m}_{k} - \underline{m}_{k}) + (\max\{\bar{r}, \beta'-\beta\} - \underline{r})\bar{m}_k \nonumber \\
    &\leq (1-\beta)(\max\{\bar{d}_{k} - \underline{d}_{k}, \bar{v}_{k} -\underline{v}_{k}\}) + (\beta + \underline{r})(\max\{\bar{d}_{k} - \underline{d}_{k}, \bar{v}_{k} -\underline{v}_{k}\} + \frac{\log(2)}{c} ) + (B_1 - \underline{r})\min\{\bar{d}_k, \bar{v}_k\} \nonumber \\
    &= (1+\underline{r})(\max\{\bar{d}_{k} - \underline{d}_{k}, \bar{v}_{k} -\underline{v}_{k}\}) + (\beta + \underline{r})\frac{\log(2)}{c} + (B_1 - \underline{r})\min\{\bar{d}_k, \bar{v}_k\} \nonumber \\
	 &\leq (1+\underline{r})(\max\{\bar{d}_{k} - \underline{d}_{k}, \bar{v}_{k} -\underline{v}_{k}\}) + B_2\frac{\log(2)}{c} + (B_1 - \underline{r})\min\{\bar{d}_k, \bar{v}_k\} \nonumber
\end{align}
		 We now apply our inductive hypothesis. In addition, we make use of the fact that $\min\{\bar{d}_k, \bar{v}_k\} \leq \min\{d_0, v_0\}(1+B_1)^k$.
		 
\begin{align}
    \max\{\bar{d}_{k+1} - \underline{d}_{k+1}, \bar{v}_{k+1} -\underline{v}_{k+1}\} &\leq (1+\underline{r})\left(\min\{d_0, v_0\}(B_1-\underline{r})k(1+B_1)^{k-1} + B_2\frac{\log(2)}{c}\frac{(1+\underline{r})^{k}-1}{\underline{r}}\right) \nonumber \\
    & \quad + B_2 \frac{\log(2)}{c} + (B_1 - \underline{r})\min\{d_0, v_0\}(1+B_1)^k \nonumber \\
    &\leq \min\{d_0, v_0\}(B_1-\underline{r})(k(1+B_1)^{k-1} + (1+B_1)^k) \nonumber \\
    &\quad + B_2\frac{\log(2)}{c} \left(\frac{(1+\underline{r})^{k+1}-(1+\underline{r})}{\underline{r}} + 1\right) \nonumber \\
    &\leq \min\{d_0, v_0\}(B_1-\underline{r})(k+1)(1+B_1)^{k} + B_2\frac{\log(2)}{c} \left(\frac{(1+\underline{r})^{k+1}-1}{\underline{r}}\right) \nonumber
\end{align}
This completes the proof of inequality \ref{eq:vtgapapproxmyopic}. In combination with inequality \ref{eq:mtgapapproxmyopic}, this completes the proof of the claim.
\halmos
\end{proof}

As the final step, we use the three claims above to establish a lower bound on the ratio between the discounted number of matches when following a repeated optimal myopic policy and the discounted number of matches when following the optimal policy.

\begin{align}
    \frac{\sum_{t=1}^T \delta^{t-1}\tilde{m}_t}{\sum_{t=1}^T \delta^{t-1}m_t^*} 
    &\geq 
    \frac{\tilde{m}_1+\sum_{t=2}^T \delta^{t-1}\underline{m}_t}{\tilde{m}_1+\sum_{t=2}^T \delta^{t-1}\overline{m}_t} \label{eq:myopicapprox1} \\
    & = 1- \frac{\sum_{t=2}^T \delta^{t-1}(\overline{m}_t-\underline{m}_t)}{\tilde{m}_1+\sum_{t=2}^T \delta^{t-1}\overline{m}_t}. \label{eq:myopicapprox2}
\end{align}

Line \eqref{eq:myopicapprox1} comes from noting that by definition, a repeated optimal myopic policy maximizes $m_1$. Thus, $\tilde{m}_1 \geq m_1^*$. {For $t \geq 2$, we rely on the upper and lower bounds established in Claims \ref{clm:myopicapproxupper} and \ref{clm:myopicapproxlower} to replace $m^{*}_t$ with $\bar{m}_t$ and $\tilde{m}_t$ with $\underline{m}_t$, respectively. Line \ref{eq:myopicapprox2} is purely algebraic.}

Using the bound established in Claim \ref{clm:myopicapproxgap} as well as properties of summations, 
\begin{align}
   \sum_{t=2}^T \delta^{t-1}(\bar{m}_t-\underline{m}_t) &\leq \sum_{t=2}^T \delta^{t-1}\min\{d_0, v_0\}(B_1-\underline{r})t(1+B_1)^{t-1} + \sum_{t=2}^T \delta^{t-1}\left(B_2\frac{(1+\underline{r})^{t}-1}{\underline{r}} + 1\right)\frac{\log(2)}{c} \nonumber \\
   &\leq \frac{2\delta(1+B_1)}{(1-\delta(1+B_1))^2}\min\{d_0, v_0\}(B_1-\underline{r}) + \sum_{t=2}^T \delta^{t-1}\left(B_2\frac{(1+\underline{r})^{t}}{\underline{r}} \right)\frac{\log(2)}{c} \nonumber \\
   &\leq \frac{2\delta(1+B_1)}{(1-\delta(1+B_1))^2}\min\{d_0, v_0\}(B_1-\underline{r}) + B_2\frac{\delta(1+\underline{r})^2}{(1-\delta(1+\underline{r}))\underline{r}} \frac{\log(2)}{c}\label{eq:myopicapproxnum}
\end{align}
Further, we can lower bound the denominator due to the fact that perfect matching will always lead to growth (assuming a non-degenerate instance):
\begin{align}
    \tilde{m}_1+\sum_{t=2}^T \delta^{t-1}\bar{m}_t \quad \geq \quad \tilde{m}_1+\sum_{t=2}^T \delta^{t-1}\tilde{m}_1 \quad = \quad  \tilde{m}_1 \frac{1-\delta^T}{1-\delta}
    \label{eq:myopicapproxdenom}
\end{align}
Combining the inequalities in lines \ref{eq:myopicapprox2}, \ref{eq:myopicapproxnum}, and \ref{eq:myopicapproxdenom} completes the proof of Proposition \ref{prop:approxrepeatmyopic} by establishing
$$\frac{\sum_{t=1}^T \delta^{t-1}\tilde{m}_t}{\sum_{t=1}^T \delta^{t-1}m_t^*} 
\geq 
1 - \frac{\delta (1-\delta)}{\tilde{m}_1(1-\delta^T)}\left((B_1-\underline{r})\frac{2(1+B_1)\min\{d_0,v_0\}}{(1-\delta(1+B_1))^2} + \frac{B_2 (1+\underline{r})^2 \log(2)}{\underline{r}(1-\delta (1+\underline{r}))c} \right)
$$
}
\end{subsection}

%%%%%%%%%%%%%%%%%%%%%%%%%%%%%%%%%%%%%%%%%%%%%%%

%   ^New approx proof^

%%%%%%%%%%%%%%%%%%%%%%%%%%%%%%%%%%%%%%%%%%%%%%%

\end{section}

\end{APPENDICES}

\end{document}